\documentclass[conference]{IEEEtran}
\usepackage{amsmath,amssymb,amsfonts,amsthm}
\usepackage{graphicx}
\usepackage{textcomp}
\usepackage{xcolor}

\def\BibTeX{{\rm B\kern-.05em{\sc i\kern-.025em b}\kern-.08em
    T\kern-.1667em\lower.7ex\hbox{E}\kern-.125emX}}
    
\usepackage{subcaption}
\usepackage{multirow}
\usepackage{tabularx}
\usepackage{booktabs}
\usepackage{xspace}
\usepackage{algorithm}
\usepackage[noend]{algpseudocode}
\usepackage{algorithmicx}
\usepackage{hyperref}
\usepackage{ifthen}
\usepackage{cleveref}
\usepackage{enumitem}

\usepackage{thmtools,thm-restate}

\usepackage{pifont}\newcommand{\cmark}{\ding{51}}\newcommand{\xmark}{\ding{55}}

\newcolumntype{R}{>{$}r<{$}} \newcolumntype{C}{>{$}c<{$}} \newcolumntype{L}{>{$}l<{$}} 

\usepackage[numbers]{natbib}

\newtheorem{theorem}{Theorem}

\newtheorem{definition}{Definition}

\graphicspath{{figures/}}

\newif\ifdisablenotes
\newif\ifdisablearxiv
\newif\ifenablepagenumbers

\disablenotestrue

\disablearxivfalse

\enablepagenumberstrue

\ifdisablenotes
\newcommand{\jhnote}[1]{}
\newcommand{\borja}[1]{}
\newcommand{\gnote}[1]{}
\newcommand{\todo}[1]{}
\else
\newcommand{\jhnote}[1]{\textcolor{blue}{Jamie: #1}}
\newcommand{\borja}[1]{\textcolor{green}{Borja: #1}}
\newcommand{\gnote}[1]{\textcolor{red}{Gio: #1}}
\newcommand{\todo}[1]{\textcolor{orange}{TODO: #1}}
\fi

\ifenablepagenumbers
\newcommand{\mypagestyle}{\thispagestyle{plain}\pagestyle{plain}}
\else
\newcommand{\mypagestyle}{}
\fi

\newcommand{\mymaketitle}{\author{
\IEEEauthorblockN{Borja Balle\textsuperscript{*}}
\IEEEauthorblockA{\textit{DeepMind}}
\and
\IEEEauthorblockN{Giovanni Cherubin\textsuperscript{*}\textsuperscript{\dag}}
\IEEEauthorblockA{\textit{Microsoft Research}}
\and
\IEEEauthorblockN{Jamie Hayes\textsuperscript{*}}
\IEEEauthorblockA{\textit{DeepMind}}
}
\maketitle
\begingroup\renewcommand\thefootnote{*}
\footnotetext{Equal contribution}
\endgroup
\begingroup\renewcommand\thefootnote{\dag}
\footnotetext{Work done while at the Alan Turing Institute}
\endgroup
}

\newcommand{\main}[1]{\textcolor{blue}{S&P submission}: #1\\ \textcolor{blue}{END}}
\newcommand{\arxiv}[1]{\textcolor{blue}{Arxiv version}: #1\\ \textcolor{blue}{END}}

\ifdisablearxiv
\renewcommand{\main}[1]{#1}
\renewcommand{\arxiv}[1]{}
\else
\renewcommand{\main}[1]{}
\renewcommand{\arxiv}[1]{#1}
\fi

\newcommand{\ip}[2]{\langle #1, #2 \rangle}
\DeclareMathOperator{\argmax}{argmax}
\DeclareMathOperator{\argmin}{argmin}
\newcommand{\norm}[1]{\| #1 \|}
\newcommand{\cN}{\mathcal{N}}
\newcommand{\cU}{\mathcal{U}}
\newcommand{\Vol}{\mathsf{Vol}}

\newcommand{\model}{\ensuremath{\theta}\xspace}
\newcommand{\train}{\ensuremath{T}\xspace}
\newcommand{\Xset}{\ensuremath{\mathcal{X}}\xspace}
\newcommand{\Yset}{\ensuremath{\mathcal{Y}}\xspace}

\newcommand{\trainobjective}{\ensuremath{C}\xspace}
\newcommand{\trainobjectivepoint}{\ensuremath{c}\xspace}
\newcommand{\dimpoint}{\ensuremath{d}\xspace}

\newcommand{\lattackloss}{\ell}
\newcommand{\attackloss}[1]{\ensuremath{\lattackloss(#1)}}
\newcommand{\Daux}{\ensuremath{\bar{D}}\xspace}
\newcommand{\zaux}{\ensuremath{\bar{z}}\xspace}
\newcommand{\modelaux}{\ensuremath{\bar{\model}}\xspace}
\newcommand{\attackmodel}{\ensuremath{\phi}\xspace}
\newcommand{\mia}{\ensuremath{M}\xspace}

\newcommand{\attackin}{shadow model}

\newcommand{\aux}{\mathsf{aux}}

\usepackage{amssymb}
\usepackage{pifont}

\mathchardef\mhyphen="2D
\newcommand{\Zset}{\ensuremath{\mathcal{Z}}\xspace}
\newcommand{\Dtrain}{\ensuremath{D}\xspace}
\renewcommand{\train}{\ensuremath{A}\xspace}
\newcommand{\Dfixed}{\ensuremath{D_{\mhyphen}}\xspace}
\newcommand{\Xfixed}{\ensuremath{\bar{X}}\xspace}
\newcommand{\Yfixed}{\ensuremath{\bar{Y}}\xspace}
\newcommand{\reconstruct}{\ensuremath{R}\xspace}
\newcommand{\prior}{\ensuremath{\pi}\xspace}
\newcommand{\zguess}{\ensuremath{\hat{z}}\xspace}

\newcommand{\R}{\mathbb{R}}
\renewcommand{\Pr}{\mathbb{P}}
\newcommand{\E}{\mathbb{E}}
\newcommand{\onesv}{\mathbf{1}}

\mypagestyle

\begin{document}

\title{Reconstructing Training Data\\ with Informed Adversaries}

\mymaketitle

\begin{abstract}
Given access to a machine learning model, can an adversary reconstruct the model's training data? This work studies this question from the lens of a powerful informed adversary who knows all the training data points except one. By instantiating concrete attacks, we show it is feasible to reconstruct the remaining data point in this stringent threat model. For convex models (e.g.\ logistic regression), reconstruction attacks are simple and can be derived in closed-form. For more general models (e.g.\ neural networks), we propose an attack strategy based on training a reconstructor network that receives as input the weights of the model under attack and produces as output the target data point. We demonstrate the effectiveness of our attack on image classifiers trained on MNIST and CIFAR-10, and systematically investigate which factors of standard machine learning pipelines affect reconstruction success. Finally, we theoretically investigate what amount of differential privacy suffices to mitigate reconstruction attacks by informed adversaries.
Our work provides an effective reconstruction attack that model developers can use to assess memorization of individual points in general settings beyond those considered in previous works (e.g.\ generative language models or access to training gradients); it shows that standard models have the capacity to store enough information to enable high-fidelity reconstruction of training data points; and it demonstrates that differential privacy can successfully mitigate such attacks in a parameter regime where utility degradation is minimal.
\end{abstract}

\begin{IEEEkeywords}
machine learning, neural networks, reconstruction attacks, differential privacy
\end{IEEEkeywords}

\section{Introduction}

Machine learning (ML) models have the capacity to memorize their training data \cite{DBLP:conf/iclr/ZhangBHRV17}, and such memorization is sometimes unavoidable while training highly accurate models \cite{DBLP:conf/stoc/Feldman20,DBLP:conf/nips/FeldmanZ20,DBLP:conf/stoc/BrownBFST21}.
When the training data is sensitive, sharing models that 
exhibit memorization
can lead to privacy breaches.
To design mitigations enabling privacy-preserving deployment of ML models we must understand how these breaches arise and how much information they leak about individual data points.

Membership leakage is considered the gold standard for privacy in ML, both from the point of view of empirical privacy evaluation (e.g., via membership inference attacks (MIA) \cite{DBLP:conf/sp/ShokriSSS17}) as well as mitigation (e.g., differential privacy (DP) \cite{DBLP:conf/tcc/DworkMNS06}).
Membership information represents a minimal level of leakage: it allows an adversary to infer a single bit determining if a given data record was present in the training dataset.
Models trained on health data represent a prototypical application where membership can be considered sensitive: the presence of an individual's record in a dataset might itself be indicative of whether they were tested or treated for a medical condition.

Reconstruction of training data from ML models sits at the other extreme of the individual privacy leakage spectrum: a successful attack enables an adversary to reconstruct all the information about an individual record that a model might have seen during training.
The possibility of extracting training data from models can pose a serious privacy risk even in applications where membership information is not directly sensitive.
For example, reconstruction of individual images from a model trained on pictures that were privately shared in a social network can be undesirable even if that individual's membership in the social network is public information.

Existing evidence of the feasibility of reconstruction attacks is sparse and focuses on specialized use cases. For example,
recent work on generative language models highlights their capacity to memorize and regurgitate some of their training data~\cite{DBLP:conf/uss/Carlini0EKS19,DBLP:conf/uss/CarliniTWJHLRBS21}, while works on gradient inversion
show that adversaries with access to model gradients (e.g.\ in federated learning (FL) \cite{DBLP:conf/aistats/McMahanMRHA17}) can use this information to reconstruct training examples \cite{zhu2019deep}.
Similarly, attribute inference attacks reconstruct a restricted subset of attributes of a training data point given the rest of its attributes~\cite{fredrikson2014privacy}, while property inference attacks infer global information about the training distribution rather than individual points \cite{ganju2018property, suri2021formalizing}.

\begin{figure}[t]
\captionsetup{width=1\columnwidth}
  \centering
    \includegraphics[width=1\linewidth]{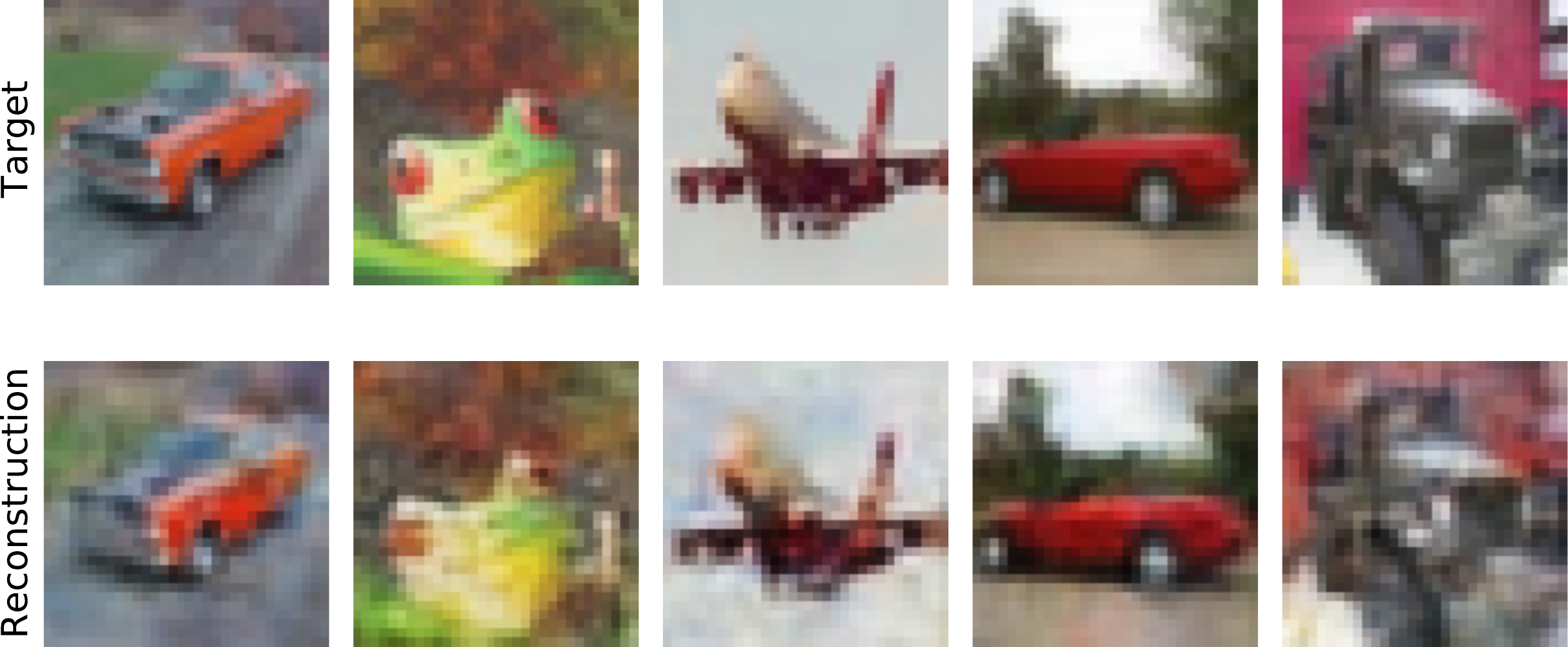}
\caption{Examples of training data points reconstructed from a 55K parameter CNN classifier trained on CIFAR-10.
}
\label{fig:fig_1_rec}
\end{figure}

Our work proposes a general approach to study the feasibility of reconstruction attacks against ML models without assumptions on the type of model or access to intermediate gradients, and initiates a study of mitigation strategies capable of preventing this kind of attacks.
The starting point is the instantiation of an \emph{informed adversary} that, knowing all the records in a training data set except one, attempts to reconstruct the unknown record after obtaining white-box access to a released model.
This choice of adversary is inspired by the (implicit) threat model in DP \cite{DBLP:conf/sp/NasrSTPC21}.

Working with such a powerful, albeit unrealistic, adversary enables us to demonstrate the feasibility of reconstruction, both in theory against convex models as well as experimentally against standard neural network architectures for image classification.
Furthermore, the use of an informed adversary
makes our work relevant for provable mitigations:
effective defenses against optimal informed adversaries will also protect against attacks run by less powerful and more realistic adversaries.

\subsection{Overview of Contributions and Paper Outline}

We start by introducing and motivating the informed adversary threat model (\Cref{sec:threat-model}).
Our first contribution is a theoretical analysis of reconstruction attacks against simple ML models like linear, logistic, and ridge regression (\Cref{sec:reconstruction-convex}).
We show that for a broad class of generalized convex linear models, access to the maximum likelihood solution enables an informed adversary to recover the target point exactly.

In the convex setting, the attack relies on solving a simple system of equations.
Extending reconstruction attacks to neural networks requires a different approach due to the inherent non-convexity of the learning problem.
In \Cref{sec:reconstruction-general}, we propose a generic approach to reconstruction attacks based on \emph{reconstructor networks} (RecoNN): networks that are trained by the adversary to output a reconstruction of the target point when given as input the parameters of a released model.

Our second contribution is to show that it is feasible to attack standard neural network classifiers using reconstructor networks;
we present effective RecoNN architectures and training procedures, and show they can extract high-fidelity training images from classifiers trained on MNIST\footnote{A minimal implementation of our reconstruction attack on MNIST is available at {\color{blue} \url{https://github.com/deepmind/informed_adversary_mnist_reconstruction}}.} and CIFAR-10.
\Cref{fig:fig_1_rec} provides an illustration of reconstructions produced by a RecoNN-based attack against a convolutional neural network (CNN) classifier trained on CIFAR-10.
These experiments
provide compelling evidence
that image classification models can store in their weights enough information to reconstruct individual training data points.

\todo{Feedback from Adria Gascon: emphasize early that we evaluate on images but the attack can be applied to any kind of data.}
\gnote{Agreed, but some reviewers may ask for experiments...
Fingers crossed.}

\Cref{sec:experiments} describes our third contribution: an in-depth analysis around what factors affect the success of our RecoNN-based attack.
These include hyper-parameter settings in the model training pipeline,
degree of access to model parameters,
and quality and quantity of side knowledge available to the adversary.
We also explore how different levels of knowledge
about the internal randomness of stochastic gradient descent (SGD) affect reconstruction;
we observe that knowing the model's initialization significantly improves the quality of reconstructions, while knowing the randomness used for mini-batch sampling is not necessary for good reconstruction.

As part of our experiments, we also investigate the use of DP-SGD \cite{DBLP:conf/ccs/AbadiCGMMT016} as a mitigation to protect against reconstruction attacks.
We find that large values of $\epsilon$ suffice
to defend
against our best RecoNN-based attacks
-- in fact, values that are much larger than what is
necessary to protect against membership inference attacks by informed adversaries \cite{DBLP:conf/sp/NasrSTPC21}.
\Cref{sec:defenses} supports this observation by introducing a definition of \emph{reconstruction robustness}, analyzing its relation to the (R{\'e}nyi) DP parameters of the training algorithm, and showing that, under mild conditions on the adversary's side knowledge, $\epsilon = o(d)$ suffices to prevent reconstruction of $d$-dimensional data records.

 \section{Reconstruction with Informed Adversaries}
\label{sec:threat-model}

We start by instantiating and justifying the \emph{informed adversary} threat model for reconstruction attacks against ML models, and by comparing it to related attacks in the literature.
Notation for the most important concepts introduced in this section is summarized in \Cref{tab:notation}.
At its core, our threat model assumes a powerful adversary with white-box access to a model released by a \emph{model developer}.
The developer owns a dataset $\Dtrain \in \Zset^n$ of $n$ training records from some domain $\Zset$, and a (possibly randomized) training algorithm $\train : \Zset^n \to \Theta$.
They train (the parameters of) a model $\theta = \train(\Dtrain)$, and then release it as part of a system or service.
For example, records in $D$ may be feature-label pairs in standard supervised learning settings, and $\train$ may implement an optimization algorithm (e.g.\ SGD or Adam) for a loss function associated with $\Dtrain$ and $\Theta$.

\begin{table}[t]
    \centering
    \caption{Summary of notation}
    \resizebox{\linewidth}{!}{
    \begin{tabular}{clcl}
        \hline
        \multicolumn{2}{c}{\textbf{Model Developer}} & \multicolumn{2}{c}{\textbf{Reconstruction Adversary}} \\
        \hline
        $\Zset$ & Data domain & $\Dfixed$ & Training dataset minus target point \\
        $\Theta$ & Model domain & $z$ & Target point \\
        $\Dtrain$ & Training dataset & $\reconstruct$ & Reconstruction algorithm\\
        $n$ & size of training set (includes target point) & $\aux$ & Side knowledge about $z$\\
        $\train$ & Training algorithm & $\zguess$ & Candidate reconstruction \\
        $\model$ & Released model & $\lattackloss$ & Reconstruction error \\
        \hline
    \end{tabular}
    }
    \label{tab:notation}
\end{table}

\subsection{Threat Model}
A reconstruction adversary with access to the
released model aims to
infer enough information about its training data to reconstruct one of the examples in $\Dtrain$.
In this paper, we consider a powerful adversary who already has full knowledge about all but one of the training points.
Formally, they have access to the following information to carry out the attack.

\begin{definition}[Informed reconstruction adversary]\label{def:informed-adversary}
Let $\model$ be a model trained on dataset $\Dtrain$ of size $n$ using algorithm $\train$.
Let $z \in \Dtrain$ be an arbitrary training data point and $\Dfixed = D \setminus \{z\}$ denote the remaining $n-1$ points;
we refer to $z$ as the \emph{target point}.
An \emph{informed reconstruction adversary} has access to:
    \begin{enumerate}[label=\alph*)]
    \item The \emph{fixed dataset} $\Dfixed$;
    \item The \emph{released model}'s parameters $\model$;
    \item The model's training algorithm $\train$;
\item (Optional) Side knowledge $\aux$ about the target point.
\end{enumerate}
\end{definition}

We first discuss each piece of knowledge we give to our attacker,
and then analyze in depth how our adversary relates to other threat models arising in
other privacy attacks.

\paragraph{Fixed dataset}
Arguably, the assumption that gives our attacker the greatest advantage
is knowing all the training data except for the target point.
There are two main reasons
to consider such a stringent threat model.
First, since our ultimate goal for studying ML vulnerabilities is to design effective mitigations, by evaluating the resilience of ML models in this strong
threat model we ensure their resilience against weaker
(and more realistic) attackers.
Second, our setup captures the
implicit
threat model
used in the DP definition;
indeed, DP bounds
the ability of a mechanism at
preventing the disclosure of membership information about one data record from an
adversary who knows all the other records in the database.

\paragraph{White-box model access}
White-box access to the model is motivated by several real-world scenarios.
First, the practice of publishing models online (e.g.\ to facilitate their
use or favor public scrutiny) is increasingly widespread.
Second, proprietary models shipped as part of hardware or software components can be vulnerable to reverse-engineering;
it would be naive to assume that sufficiently motivated adversaries will never obtain white-box access to such models.
Finally, FL settings may give real-world attackers
access to similar information to the one we capture in our threat model.

\paragraph{Training algorithm}
Privacy (and security) through obscurity is generally regarded as a bad practice.
Thus, we assume the adversary has access to the model developer's training algorithm $A$, including any associated hyper-parameters (e.g.\ learning rate, regularization, batch size, number of iterations, etc).
Access to $A$ can be in the form of a concrete (e.g.\ open source) implementation.
Nevertheless, black-box access
(e.g.\ through a SaaS API)
suffices for
the general reconstruction attack presented in \Cref{sec:reconstruction-general}.
In cases where $\train$ is randomized, we will evaluate
attacks with and without knowledge of the different sources of randomness used when training the released model.
In stochastic optimization algorithms these typically include model initialization and mini-batch sampling.
Knowledge of $\train$'s internal randomness could come from the model developer using a hard-coded random seed in a public implementation.
Alternatively, knowledge about the model's initialization will also be available
whenever the released model is obtained by fine-tuning a publicly available model (e.g.\ in transfer learning scenarios),
or in FL settings where the adversary has successfully compromised an intermediate model by taking part in the training protocol.

\paragraph{Side knowledge about target point}
Privacy attacks do not happen in a vacuum, so adversaries will often have some prior information about the target point before observing the released model.
For starters, knowledge of $\Dfixed$ and $\train$ provides the adversary with syntactic and semantic context for a learning task in which the model developer
deemed it useful to include the target point.
In our investigations, we often consider adversaries with additional side knowledge abstractly represented by $\aux$.
From a practical perspective, the attack presented in \Cref{sec:reconstruction-general} takes $\aux$ to be a dataset $\Daux$ of points disjoint from $\Dfixed$.
For example, these could come from a public academic dataset or from scraping relevant websites.
Our experiments in \Cref{ssec:reconctruction-factors} show that
these additional points do not necessarily need to come
from the same distribution as the training
data.
In our theoretical investigation (\Cref{sec:defenses}), we model the adversary's side knowledge as a probabilistic prior $\pi$ from which the target is assumed to be sampled.

\subsection{Reconstruction Attack Protocol and Error Metric}\label{ssec:attack-protocol}

\Cref{algo:reconstruction} formalizes the interaction between model developer and reconstruction adversary in our threat model.
After the model $\theta$ is trained on $D = \Dfixed \cup \{z\}$, the adversary runs their attack algorithm $\reconstruct$ using all the information discussed in the previous section, and produces a \emph{candidate} reconstruction $\zguess$ for the target point $z$.
The protocol returns a measure of the attack's success based on a \emph{reconstruction error} function $\ell$; smaller error means the reconstruction is more faithful.

\begin{algorithm}
\caption{Reconstruction attack with an informed adversary. (Auxiliary side knowledge $\aux$ is optional).
} 
\begin{algorithmic}

\Procedure{Reconstruction}{$\train, \reconstruct, \Dfixed, z; \aux$}
    \State $\model \gets \train(\Dfixed \cup \{z\})$
    \State $\zguess \gets \reconstruct(\model, \Dfixed, \train; \aux)$
    \State \Return \attackloss{z, \zguess}
\EndProcedure
\end{algorithmic}
\label{algo:reconstruction}
\end{algorithm}

Privacy expectations are
contextual, and depend on the information content and modality of the sensitive data.
Perfect reconstruction may not be necessary for the user to claim their privacy has been violated;
e.g., a privacy breach may occur if the image of a car's license plate is revealed via an attack, even if the reconstructed background is inaccurate.
In particular, the error function $\lattackloss$ can encode not only proximity between the feature representations of the target and candidate points, but also the correctness with which an attack can recover a (private) property of interest about the target. 
Our experiments on image classifiers use the MSE between pixels as a measure of reconstruction, as well as the similarity between outputs of machine learning models on $z$ and $\zguess$ (through the LPIPS and KL metrics cf.\ \Cref{ssec: metrics}).
In general, an appropriate choice of $\lattackloss$ and a threshold for declaring successful reconstruction is a policy question that will depend on the particular application: it should capture the minimum level of leakage that would cause a significant harm to the involved individual.

\subsection{Relation to Attribute Inference}\label{sec:aia}

Reconstruction can be seen as a generalization of attribute
inference attacks (AIA)~\cite{fredrikson2014privacy,fredrikson2015model,DBLP:conf/csfw/YeomGFJ18,DBLP:conf/cvpr/ZhangJP0LS20}, also sometimes referred to as model inversion attacks.
In AIA, an attacker that knows part of a data record $z$
aims to reconstruct the entire record by exploiting
(white-box or black-box) access to a model $\model$ whose training dataset contained $z$. It is also common for the attack goal of a model inversion attack to try and reveal training data information in aggregate, possibly isolated to a specific target label. 
Although no individual training records are reconstructed through this attack, privacy can be leaked if aggregated training information with respect to a target label is sensitive (e.g. facial recognition where each label is associated with an identity).
The standard threat model in AIA does not include an informed adversary,
but we can get a more direct comparison with our model by considering
an \textit{informed} AIA adversary.
Such an adversary is identical to \Cref{def:informed-adversary}
but also receives as input partial information about the target point $z$, which we denote by $\eta(z)$.
This can be incorporated in \Cref{def:informed-adversary} via the side knowledge $\aux$,
showing that informed AIA corresponds to reconstruction in our model with a particular type of side knowledge.
We conclude that any investigation into mitigating general reconstruction attacks in our threat model will also be useful in protecting against informed AIA, and, by extension, standard AIA.

\subsection{Relation to Membership Inference}\label{sec:mia}

In membership inference attacks (MIA) \cite{DBLP:conf/sp/ShokriSSS17,DBLP:conf/csfw/YeomGFJ18,salem2018ml,DBLP:conf/sp/NasrSH19}, an attacker with access to a released
model $\model$ and a \textit{challenge example} $z \in \Zset$ guesses
if $z$ was part of the model's training data.
Like in AIA, standard MIA does not assume an informed adversary.
Introducing an \textit{informed} MIA adversary yields a model matching the adversary in the threat model behind DP \cite{DBLP:conf/sp/NasrSTPC21}.
This adversary is identical \Cref{def:informed-adversary}, with the exception that it also receives two  candidates $z_0, z_1 \in \Zset$ for the additional data point that was used for training the model, and the developer decides which one to use uniformly at random.
The corresponding interaction protocol between model developer and adversary is summarized in \Cref{algo:informed-mia}, where the adversary uses a MIA algorithm $M$ and the result provides a bit representing whether it guessed correctly.

\borja{We don't compare along white-box/black-box axis}

\begin{algorithm}
\caption{Informed Membership Inference Attack
}
\begin{algorithmic}[1]

\Procedure{Informed-MIA}{$\train, M, \Dfixed, z_0, z_1$}
    \State $b \gets \mathrm{Unif}(\{0,1\})$
    \State $\model \gets \train(\Dfixed \cup \{z_b\})$
    \State $\hat{b} \gets \mia(\model, \Dfixed, \train, z_0, z_1)$
    \State \Return $b = \hat{b}$
\EndProcedure
\end{algorithmic}
\label{algo:informed-mia}
\end{algorithm}

We remark that this attacker is much more
powerful than the one in standard MIA.
In particular, if the model's
training algorithm $\train$ is deterministic,
then there is a trivial strategy:
the attacker trains models
on $\Dfixed \cup \{z_0\}$ and $\Dfixed \cup \{z_1\}$ and checks which of the two
matches the released model $\model$.
This is coherent with the observation that randomized algorithms are necessary to (non-trivially) provide DP.
Note also that accurate reconstruction provides an informed MIA.
Indeed, assume, for example, that $\ell$ satisfies the triangle inequality and reconstruction succeeds at achieving error less than $\ell(z_0, z_1) / 2$.
Then the reconstruction adversary uses $\theta$ to obtain
a candidate $\zguess$, and then guess $z_0$ if $\ell(\zguess, z_0) < \ell(\zguess, z_1)$ and $z_1$ otherwise.

The contrapositive implication of the above is that if this powerful notion
of MIA is not possible, then accurate reconstruction is also not possible.
Furthermore, the existence of a standard MIA attacker implies the existence of
an informed one.
This argument indicates that protecting against informed MIA will protect against both standard MIA and accurate reconstruction, thus motivating the use of DP -- a mitigation against informed MIA -- as a strong privacy protection.
The experiments in \Cref{sec:experiments} and the theoretical investigation developed in~\Cref{sec:defenses} will, however, illustrate that values of the DP parameter $\epsilon$ that are too large to protect against informed MIA can still protect against accurate reconstruction.

\subsection{Further Related Work}

Attacks for reconstructing training data have been studied in the context of generative language models (LM).
Carlini et al.~\cite{DBLP:conf/uss/Carlini0EKS19} proposed a \emph{targeted} black-box reconstruction attack where the adversary knows part of a training example (i.e.\ a text prompt) and infers the rest (e.g.\ a credit card number).
Their attack assumes partial knowledge of the target record
(as with AIA) and a threat model where the adversary has significant computational power but no additional knowledge of the training data.
An \emph{untargeted} version of this attack was later performed against GPT-2 \cite{radford2019language} by repeatedly sampling from the model and comparing the samples with the training data \cite{DBLP:conf/uss/CarliniTWJHLRBS21}.
Both works crucially exploit the generative
aspect of LMs to carry out reconstruction;
our attacks are more general and require no such assumptions, making them suitable to attack standard image classification models.

Many works have investigated what an attacker can infer from inspecting the intermediate gradients in FL settings or
multiple model snapshots during training \cite{wang2019beyond,geiping2020inverting,wainakh2021user, BeguelinWTRPOKB20, Salem20}.
These attacks focus on inferring training points,
their labels, or related properties.
The task our reconstruction adversary has to solve
is harder: whilst a gradient leakage adversary has access to information involving only a mini-batch of training points, our attacks
needs to invert the entire training procedure.

Finally, \emph{property inference attacks} (PIA) are a generalization of AIA where the adversary infers
properties about the training set \cite{ganju2018property, suri2021formalizing}.
These attacks are effective at recovering overall statistics (e.g.\ the percentage of training records
coming from a minority group, the average value of a feature across the data) but in general do not compromise the privacy of individuals.

 \section{Reconstruction in Convex Settings}
\label{sec:reconstruction-convex}

In this section, we focus on attacking convex supervised learning models.
We discuss a general reconstruction attack strategy against a broad family of convex models when the empirical risk minimization (ERM) problem has a unique minimum and is solved to optimality.
Specifically, we show there exists a closed form solution to perform reconstruction attacks against Generalized Linear Models (GLMs) without any additional side knowledge
about the target point.
This attack applies to popular models such as linear regression, ridge regression, and logistic regression.

\subsection{Reconstruction Strategy for Convex Models}

Consider an ML model $\model$ trained by exactly solving the ERM problem.
Formally, let
$\trainobjective(\hat{\model}) = \sum_{z \in \Dtrain} \trainobjectivepoint(z, \hat{\model})$
be a risk function for some loss
$\trainobjectivepoint$,
and let
$\model \in \argmin_{\hat{\theta} \in \Theta} \trainobjective(\hat{\model})$.
If the loss is strictly convex, this optimization admits a unique global minimum.
Further, if the loss is differentiable and there is no constraint on the parameters (i.e.\ $\Theta = \R^{d'}$), then the optimum is characterized by the system of equations $\nabla C(\theta) = 0$.

This simplified scenario enables a direct strategy to perform a reconstruction attack.
Recall the adversary has white-box access to
the released model $\model$ and knowledge of the fixed dataset $\Dfixed$.
This allows them to write the following system of equations which will be satisfied by the target point $z$:
\begin{align}\label{eq:convex-attack-strategy}
    \nabla_{\theta} c(z, \theta) = -\textstyle{\sum_{z' \in \Dfixed}} \nabla_{\theta} c(z', \theta) \enspace.
\end{align}
Since in supervised training every point $z = (x, y)$ is represented by a feature vector $x \in \R^\dimpoint$ and a label $y \in \R$, this provides $d'$ equations from which the adversary wants to recover $d + 1$ unknowns ($d$ features plus the label).
Note that this strategy is independent of
the algorithm that was used for training the model as long as the model was trained to optimality.
Next we show a closed-form solution
for this attack exists in the case of GLMs
fitted with an intercept term.

\subsection{Closed-Form Reconstruction Against GLMs}
Consider fitting a GLM derived from a canonical exponential family with canonical link function $g$
(see, e.g.\ \cite{mccullagh2019generalized}).
The GLM parameters are trained via (regularized) ERM by minimizing the maximum likelihood objective
$\trainobjective(\hat{\model}) = -\sum_{(x,y) \in \Dtrain} \left(b(\ip{x}{\hat{\model}}) - \ip{x}{\hat{\model}} y\right)  + \lambda \norm{\hat{\model}}^2$,
where $b$ is a function satisfying $b' = g^{-1}$,
and $\lambda \geq 0$ is a regularization parameter. 
For example, $g^{-1}$ is the identify function for linear regression and the sigmoid function for logistic regression.
This optimization admits a unique minimum when either $\lambda > 0$, $b$ is strictly concave (as in the examples above) or the data is in general position \cite{wedderburn1976existence}.
In any of these cases \eqref{eq:convex-attack-strategy} connects the unknown $z = (x,y)$ with $\theta$ and $\Dfixed$.
Assuming the model is trained with an intercept parameter\arxiv{\footnote{In appendix,
we show an attack against linear regression without intercept parameter (\Cref{thm:lr-reconstruction-no-intercept}), which although assumes the adversary knows $y$.}}
(i.e.\ the first coordinate of each feature vector is equal to $1$) this results in a system of $\dimpoint$ equations with $\dimpoint$ unknowns. The following solution for this system gives an effective reconstruction attack.

\begin{restatable}[Reconstruction attack against GLMs]{theorem}{attackglm}
\label{thm:convex-attack}
Let $\model$ be the unique optimum of $\trainobjective(\hat{\model})$
and $\Dfixed$ the training data set except for one point $z=(x,y)$.
Suppose $\Xfixed \in \R^{(n-1) \times d}$ contains as rows the features of all points in $\Dfixed$ where its first column satisfies $\Xfixed_1 = \vec{1}$, and similarly for the labels $\Yfixed \in \R^{n-1}$.
Then taking $B = g^{-1}(\Xfixed \model) - \Yfixed$ we get:
\begin{align*}
    x = \frac{\Xfixed^\top B + \lambda\model}{\Xfixed_1^\top B + \lambda\model_1} \enspace,
    \qquad
    y = g^{-1}(\ip{x}{\model}) + \lambda \Xfixed_1^\top B \model_1 \enspace.
\end{align*}

\end{restatable}

We defer all proofs to the appendix.
Two important takeaways from this result are: 1) an informed adversary needs no additional side knowledge about $z$ to effectively attack a GLM trained with intercept; and, 2) whether the model overfits the data or generalizes well plays no role in the attack's success.

 \section{A General Reconstruction Attack}
\label{sec:reconstruction-general}

We describe a reconstruction attack against general ML models.
Intuitively, our attack stems from the observation that the 
influence of the target point $z$ on the released model $\theta$
is similar to the influence an alternative point $\bar{z}$ would have on the model $\bar{\theta} = \train(\Dfixed \cup \{\bar{z}\})$.
By repeatedly training models on different points, our attack collects enough information about the mapping from training points to model parameters to invert it at the model of interest $\theta$.
We give a high-level introduction to our attack strategy using \emph{reconstructor networks} (RecoNN).

\subsection{General Attack Strategy}

Let us use the shorthand notation $\train_{\Dfixed} : \Zset \to \Theta$ with
$\train_{\Dfixed}(z) = \train(\Dfixed \cup \{ z\})$
to emphasize that, from the point of view of an informed adversary,
when $\Dfixed$ is fixed $\train$ effectively becomes a mapping from target points to model parameters.
An ideal reconstruction attack would invert the training
procedure
and output
$\zguess = \train_{\Dfixed}^{-1}(\model)$;
whenever $\train$ is easy to invert, this will produce a perfect reconstruction as in the setting analyzed in \Cref{sec:reconstruction-convex}.
In general, however, the training process is not (easily) invertible, due to the non-convexity of the optimization problem solved by $A$, or to the presence of randomness in the training process.
In such settings, our general reconstruction attack relies on \emph{approximately} solving this inverse problem by producing a function
$\attackmodel: \Theta \rightarrow \Zset$
that associates model weights to a guess for the target point in a similar way to the (ideal) inverse mapping $\train_{\Dfixed}^{-1}$.
Note that the adversary in this threat model is extremely powerful; for example, they could enumerate (a fine discretization of) $\Zset$ and pick the candidate $\hat{z}$ that produces the model $\hat{\theta} = \train_{\Dfixed}(\hat{z})$ closest to $\theta$. However, for high-dimensional data this enumerative approach is infeasible, so we focus on attacks that can be executed in practice.

In this paper, we instantiate the search for $\phi$ as a learning problem, effectively using ``neural networks to attack neural networks''.
To solve this learning problem, we first design a RecoNN architecture for neural networks whose inputs lie in the parameter space $\Theta$ of the released model and outputs lie in the domain $\Zset$ of the training data;
typically we can encode both using numerical vectors.
The adversary then uses its knowledge of $\Dfixed$ and $A$, together with side knowledge in the form of \emph{shadow target} points $\Daux$ disjoint from $\Dfixed$, to generate a collection of \emph{shadow models}.
These shadow model and target pairs comprise the training data for the RecoNN,
which is then applied to the released model to obtain a candidate reconstruction $\zguess$ for the (previously unseen) target point $z$.

\subsection{Training Reconstructor Networks}
\label{ssec: attack_desc}

Consider an informed adversary in our threat model (\Cref{def:informed-adversary}).
As side knowledge about $z$, we assume the attacker has $k$ additional shadow targets
$\Daux = \{\zaux_1, \ldots, \zaux_k\}$ from $\Zset$.
Ideally, if we think that the attack's success will depend on the RecoNN's ability to exhibit statistical generalization,
these points would be sampled from the same distribution as the target point $z$.
Nonetheless, we will see in our experimental evaluation that this requirement is not strictly necessary to achieve good reconstructions (\Cref{ssec:reconctruction-factors}).
The general reconstruction attack
proceeds
as follows
(see also \Cref{fig:general-attack}):
\begin{enumerate}
\item For $i = 1,\ldots,k$, train model $\modelaux_i = \train_{\Dfixed}(\zaux_i)$ on the fixed dataset plus the $i$th shadow target from the adversary's side knowledge pool $\Daux$.
Together, we refer to the collection of shadow model-target pairs $S = \{(\modelaux_i, \zaux_i)\}_{i=1}^k$ 
as the \emph{attack training data}.
    \item
Train a RecoNN $\attackmodel$ using $S$ as examples of successful reconstrutions.
    Abusing our notation, we use $R$ to denote the training algorithm used by the adversary: $\attackmodel = R(S)$.
\item Obtain a reconstruction candidate by applying the RecoNN to the target model: $\zguess = \attackmodel(\model)$.
\end{enumerate}

\begin{figure}[t]
\captionsetup{width=1\columnwidth}
  \centering
    \includegraphics[width=0.75\columnwidth]{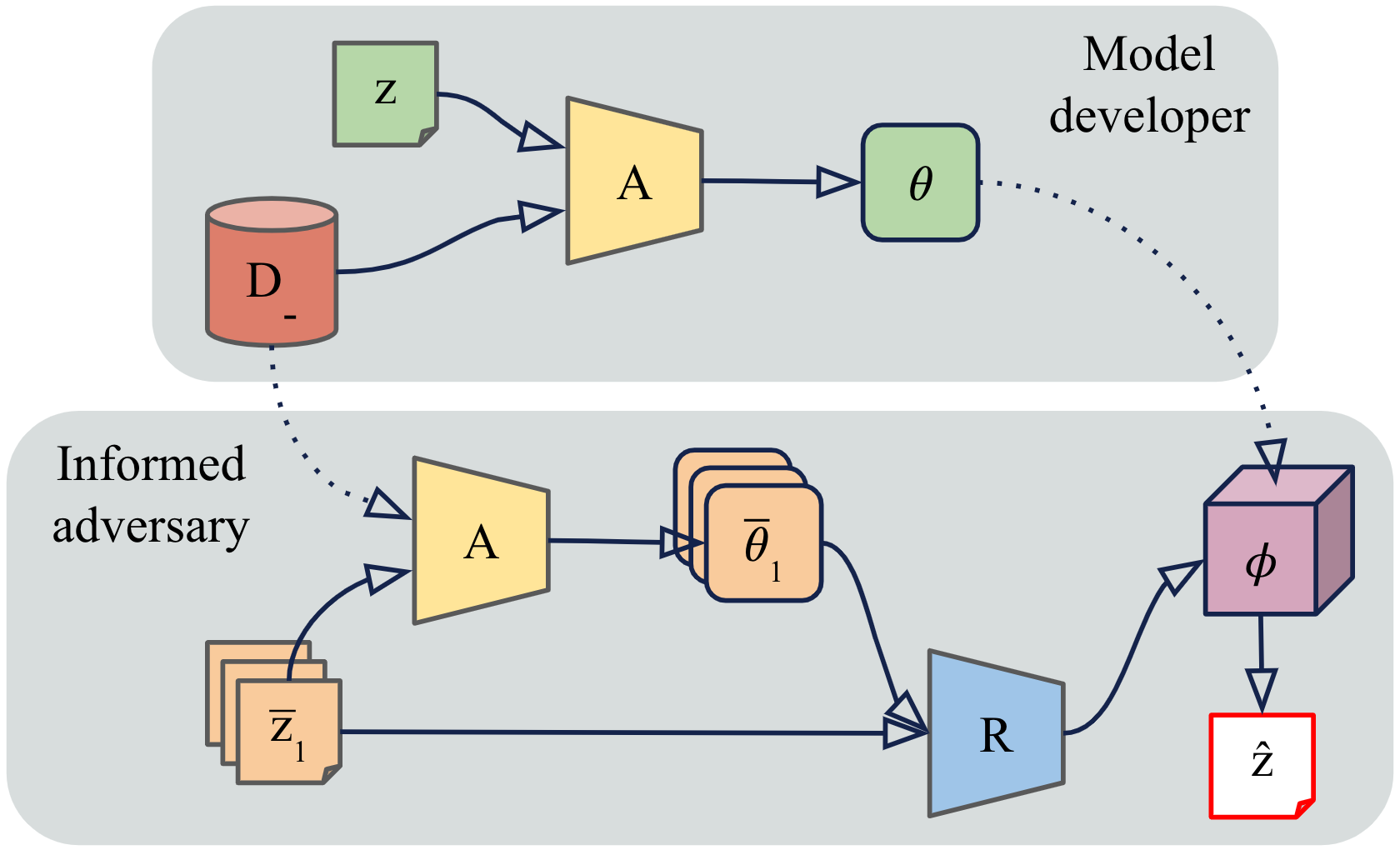}
\caption{
Overview of RecoNN-based attack.
}
\label{fig:general-attack}
\end{figure}

In all our experiments, we consider classification tasks where $z = (x,y) \in \Xset \times \Yset$ with $\Xset \subset \R^d$ and $\Yset$ is a finite set of labels.
We also make the simplifying assumption that $y$ can be inferred from $x$, and focus only on reconstructing $x$.

\paragraph*{Related work}
The idea of using ``neural networks to attack neural networks''
has been used in the literature to implement a number of attacks, including (black-box and white-box) membership inference \cite{DBLP:conf/sp/ShokriSSS17,salem2018ml,DBLP:conf/sp/NasrSH19}, model inversion \cite{DBLP:conf/cvpr/ZhangJP0LS20}, and property inference \cite{ganju2018property, suri2021formalizing}.
Our use of RecoNNs is related to \cite{ganju2018property}, where an invariant representation of a released neural network parameters is fed into another neural network to perform a PIA, although the output of our attack is often a high-dimensional object (e.g.\ an image) instead of single scalar.
In preliminary experiments we did not see an improvement from using this invariant representation as a pre-processing step; standard normalization was sufficient for a successful attack.
Similarly, the use of shadow models
trained by the adversary to imitate the behavior of the released model is a common approach in MIA and AIA, although most works do not consider an informed adversary with knowledge of $\Dfixed$.
Despite the attack being an instantiation of the shadow model technique, it is not a foregone conclusion that this approach will work for reconstruction attacks.
Reconstruction is a more difficult task than membership inference, and it entails a considerable amount of engineering, data curation, and ML training insight to carry out, as we will discuss.

\section{Experimental Setup}
\label{ssec: exp_setup}

We discuss the default experimental settings, and how we will evaluate reconstruction attacks.

\subsection{Default Settings}
We evaluate our reconstruction attacks on the MNIST and CIFAR-10 datasets using fully connected (i.e.\ multi-layer perceptron) and convolutional neural networks (CNN) as the released (and shadow) models.
Our experiments investigate
the influence that training hyperparameters for $\train$ have on the
effectiveness of reconstruction.
Default model architectures and hyperparameters for both released and reconstructor models are summarized in \Cref{tab:experimental-setup}.
Most of these choices are standard and were selected based on preliminary experiments.
In the following we highlight the most important details.

\paragraph{Dataset splits}
We split each dataset into three disjoint parts: fixed dataset ($\Dfixed)$, shadow dataset ($\Daux$), and test targets dataset;
the latter contains $1K$ points, both for MNIST and CIFAR-10.
We train one released model per test target and report average performance of our attack across test targets.

\paragraph{Released model training}
The training algorithm for released and shadow models is standard gradient descent with momentum. By default, we use full batches (i.e.\ no mini-batch sampling) to keep the algorithm deterministic.
Additionally, by default we assume the adversary
knows the
model initialization step, so both released and shadow models are trained from the same starting point.
We explore the effect of mini-batching and random initialization separately in \Cref{ssec:reconctruction-factors}.

The architecture is an MLP for MNIST and a CNN for
CIFAR-10.
On average, the released models achieve over 94\% accuracy on MNIST and 40\% on CIFAR-10 without significant overfitting (generalization gap is less 1\% on MNIST and 5\% on CIFAR-10). 
The reason for the subpar performance on CIFAR-10 is partially\footnote{Training without random mini-batches, no regularization and a small CNN architecture also contribute to this effect.} because the models are trained with only 10\% of the data used in standard evaluations -- this constraint comes from the need to reserve a large disjoint set of shadow points to train RecoNN.
We experiment with a larger CIFAR-10 fixed set size ($50K$) in \Cref{ssec:reconctruction-factors}; in this setting the released models achieve $\sim50\%$ test accuracy.

We expect reconstructing CIFAR-10 targets will be a more challenging task than MNIST.
CIFAR-10 images have a richer, more complex structure, and so capturing and reconstructing the intricacies of such an image may be difficult. 
Additionally, the underlying released model is larger;
hence: 1) a larger reconstructor network is required, which comes with higher computational costs for the adversary; 2) the shadow dataset may need to be larger, to facilitate learning on high dimensional data (i.e.\ on the shadow models' weights).

\paragraph{Reconstructor network training}
When training the reconstructor, shadow model parameters across layers are flattened and concatenated together.
We also re-scale each coordinate in this representation to zero mean and unit variance; we found this pre-processing step to be important, as some of the parameters can be extremely small.
For MNIST, we use a mean absolute error (MAE) + mean squared error (MSE) loss between shadow targets and reconstructor outputs as the training objective.
For CIFAR-10 we modify the reconstructor training objective by adding an LPIPS loss \cite{zhang2018unreasonable} and a GAN-like Discriminator loss to improve visual quality of reconstructed images.
We use a patch-based Discriminator \cite{isola2017image} with the architecture given in \Cref{tab: cifar10_attack_discrim_model}, and train it using mean squared error loss \cite{mao2017least} and a learning rate of $10^{-5}$.
The patch-based discriminator aims to distinguish shadow targets from reconstructor generated candidates. At a high-level, we can view the reconstructor network as a generative model with a latent space defined over a distribution of shadow models;
this enables us to apply ideas from Generative Adversarial Networks (GANs) training.
Our discriminator training set-up is as in \cite{isola2017image} -- we alternate between one gradient descent step on the discriminator, and one step on the reconstructor network. 
From visual inspection, we found using a discriminator improves sharpness of CIFAR-10 reconstructed images, even if it does not strictly improve the MSE metric.

\subsection{Criteria for Attack Success}
\label{ssec: metrics}

In our experiments, we use several evaluation
metrics $\ell$
to capture various aspects of information leakage from reconstruction attacks.
When reporting an average metric we measure performance of a single reconstructor network on $1K$ released model and target point pairs.

\paragraph{Mean squared error (MSE)}
We report the MSE between a target and its reconstruction.
In the context of images,
while discovery of private information does not necessarily perfectly coincide with a decreasing MSE
between the original and reconstructed training point, in general the two are correlated (\Cref{ssec: flagship_exp}).

\paragraph{LPIPS}
We report the LPIPS metric~\cite{zhang2018unreasonable} as it has been shown to be closer to the human’s visual systems determination of image similarity in comparison to the MSE distance.
LPIPS is measured by comparing deep feature representations
from visual models trained with similarity judgements made
by human annotators.

\paragraph{KL}
After running the attack, a real-world
adversary may need to post-process the reconstructed
image;
e.g.\
if they wanted to extract a license
plate from the reconstructed image, they may need to
run a downstream image classifier.
We therefore include a similarity metric between the outputs of a highly accurate classifier on the target and reconstructed image based on the Kullback–Leibler (KL) divergence between predicted class probabilities.
For MNIST, we use a LeNet classifier \cite{726791} achieving 99.4\% test accuracy, and for CIFAR-10 use a Wide ResNet \cite{ZagoruykoK16} achieving 94.7\% test accuracy.

\begin{figure}[t]
\captionsetup{width=0.5\textwidth}
  \centering
\begin{subfigure}[t]{.25\textwidth}
\centering
    \includegraphics[width=0.99\linewidth]{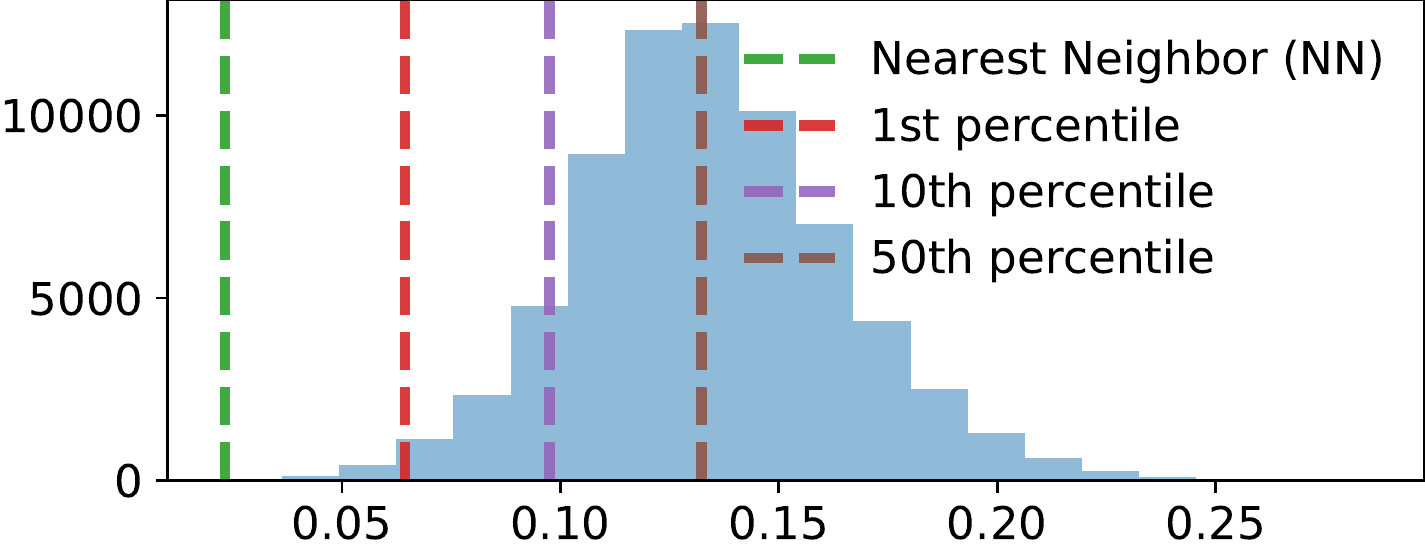}
    \caption{MNIST. $|\Daux \cup \Dfixed|=69K$.}
    \label{fig:mnist_oracle_hist}
\end{subfigure}\begin{subfigure}[t]{.25\textwidth}
\centering
    \includegraphics[width=0.99\linewidth]{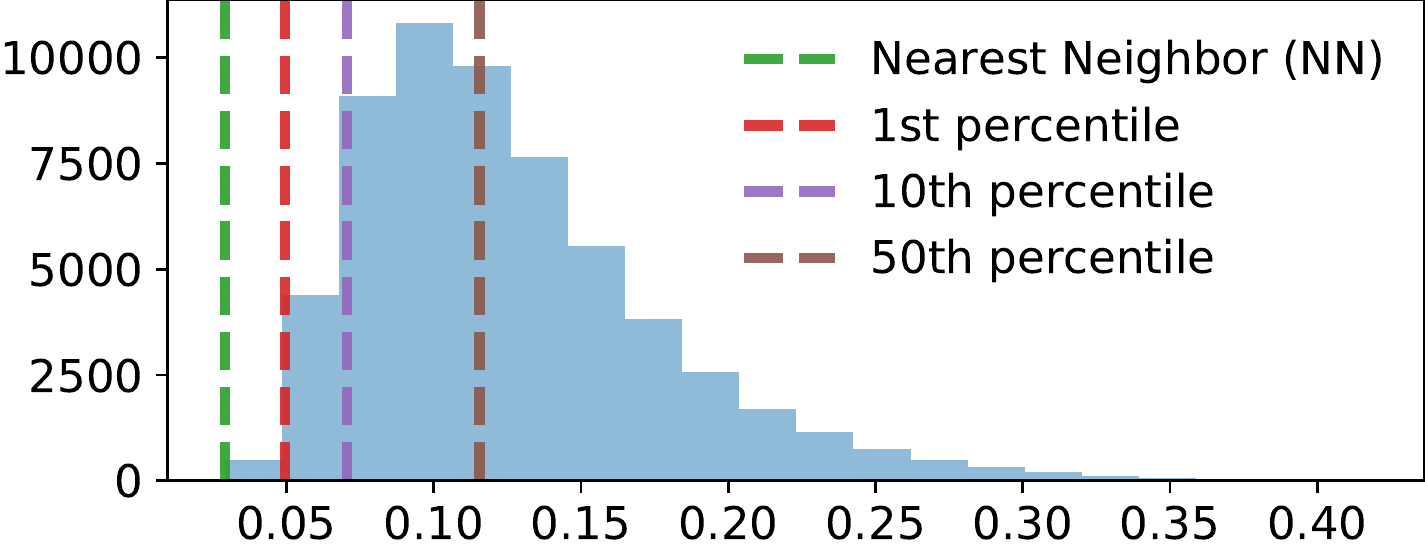}
    \caption{CIFAR-10. $|\Daux \cup \Dfixed|=59K$.}
    \label{fig:cifar10_oracle_hist}
\end{subfigure}\caption{For each test target compute the MSE to all points in adversary's available pool of data ($\Daux \cup \Dfixed$) in the default setting. Plot the histogram of MSEs (averaged over all test targets) along with some highlighted order statistics.
}
\label{fig:oracle_hist}
\end{figure}

\paragraph{Nearest Neighbor Oracle}\label{sec:baseline}

To
contextualize MSE reconstruction metrics
we consider an oracle that exploits all the data available to the adversary in the default setting
and guesses the point $\zguess \in \Dfixed \cup \Daux$ that has the smallest
MSE distance to $z$.
The MSE distance between $z$ and its nearest neighbor $\zguess$ serves as a conservative threshold for successful reconstruction: although faithful reconstructions with larger MSE are certainly possible, falling below the threshold means the reconstruction is closer to the target than to any other point previously available to the adversary, so the attack must have extracted unique information about the target point from the released model.
\Cref{fig:oracle_hist} provides average histograms (over 1K test targets) of MSEs between a target point and all points in $\Dfixed \cup \Daux$.
The green line corresponds to the average MSE to the nearest neighbor across all test targets ({0.0232} on MNIST and {0.0291} on CIFAR-10); if reconstructions have a smaller MSE than this distance we will judge the target to have been successfully reconstructed.
For reference, we also highlight the 1st, 10th and 50th percentile MSEs, which will be helpful to contextualize experiments throughout \Cref{sec:experiments}.

 \begin{figure}[t]
\centering
\begin{subfigure}[t]{.24\textwidth}
\centering
    \includegraphics[width=0.99\linewidth]{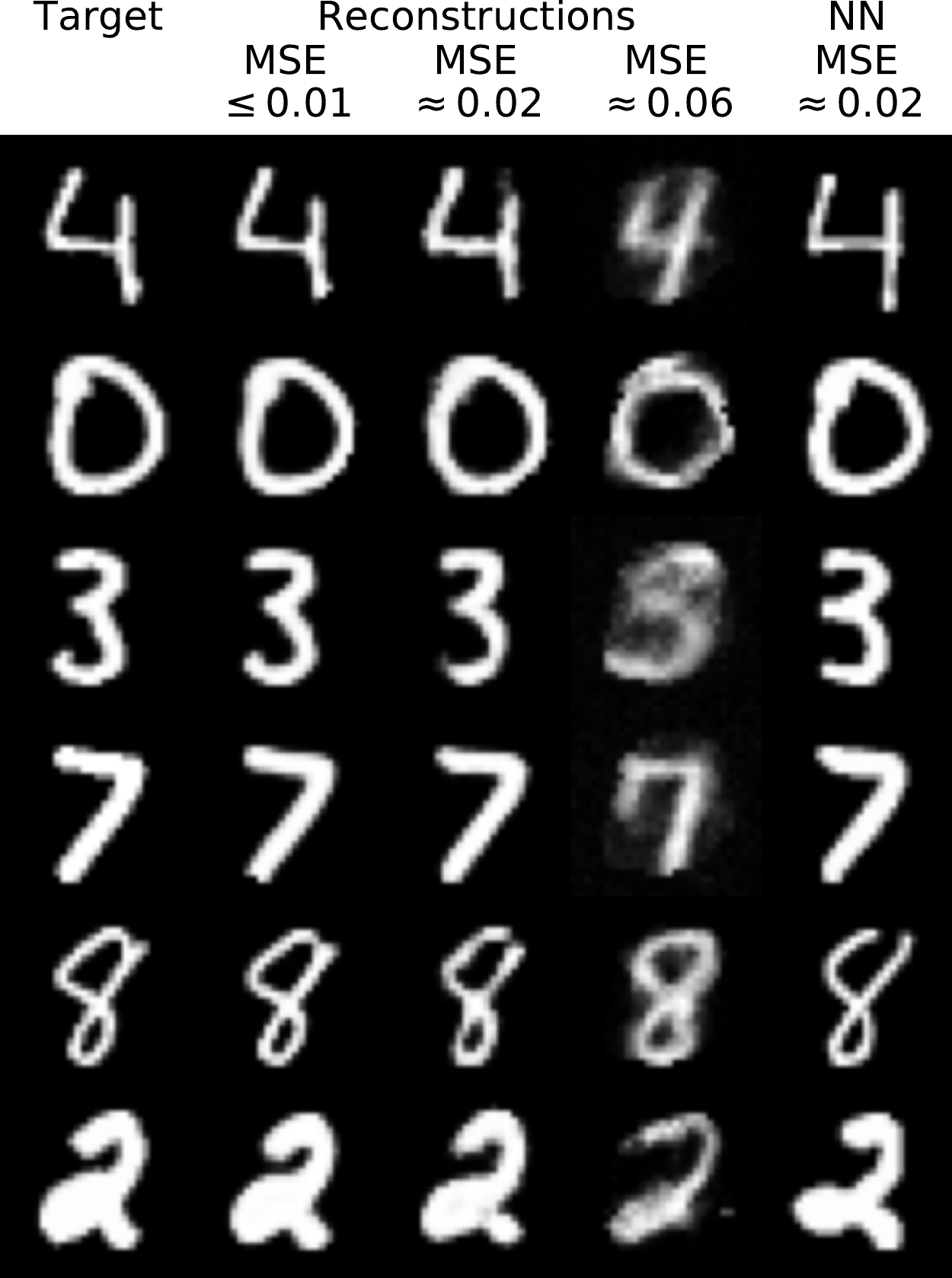}
    \caption{MNIST}
\end{subfigure}\begin{subfigure}[t]{.24\textwidth}
\centering
    \includegraphics[width=0.99\linewidth]{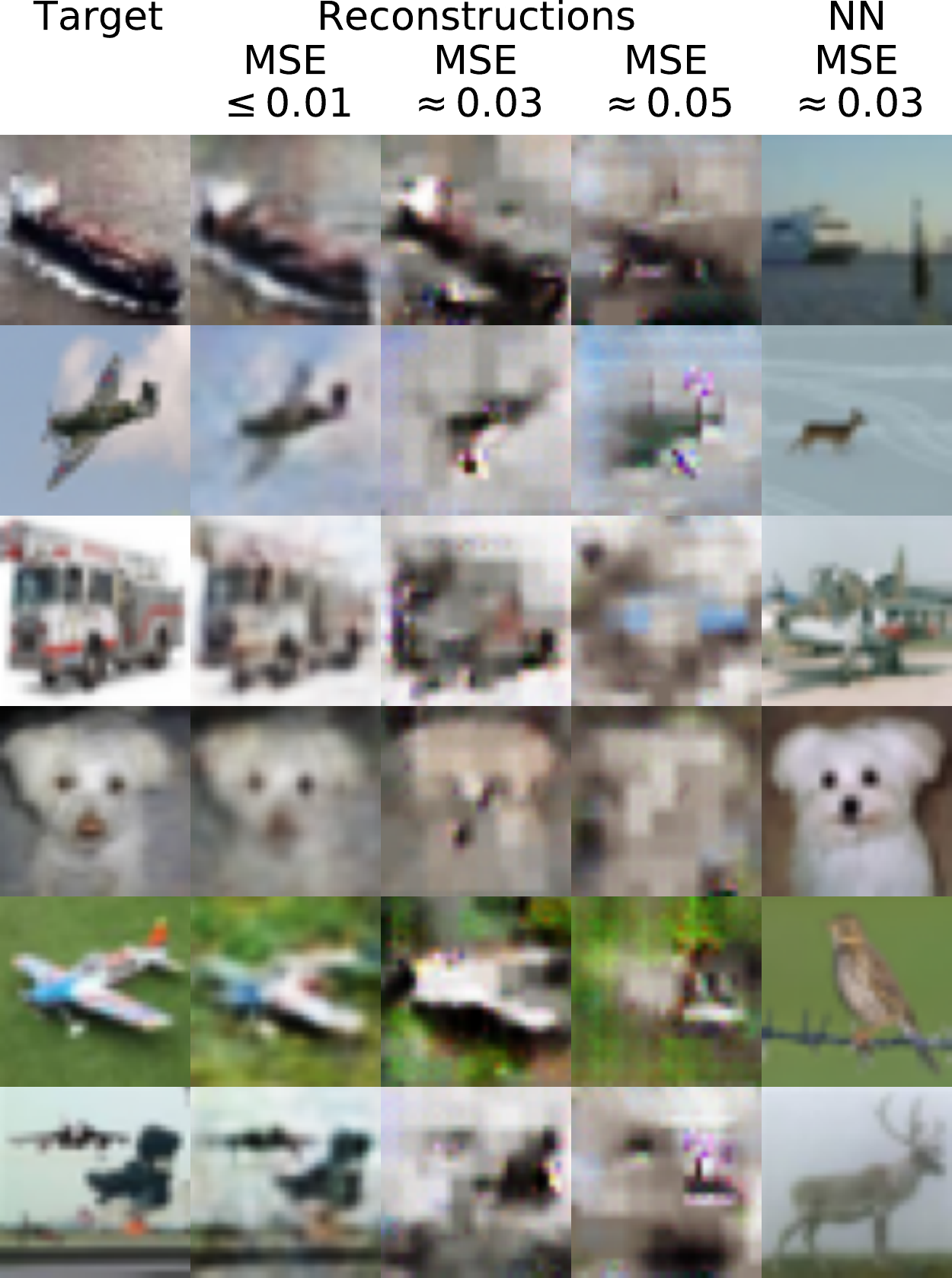}
    \caption{CIFAR-10}
\end{subfigure} \caption{
Visualization of reconstructions for six random targets selected from the test set.
The first column shows the targets, the second shows the default reconstruction attack, the third shows reconstructions around the same MSE provided by the NN oracle. The fourth column corresponds to reconstructions with distance approximately equal to the 1st percentile (\Cref{fig:oracle_hist}). The last column shows the NN oracle.}
\label{fig:flagship-experiment}
\end{figure}

\begin{table*}[t]
\centering
\caption{Effect of different factors on the success of reconstruction attacks.
}
\label{tab:summary}
\resizebox{\textwidth}{!}{
\begin{tabular}{llcccc}
\toprule
\multirow{2}{*}{Factor} & \multirow{2}{*}{Description} & \multicolumn{2}{c}{MNIST} & \multicolumn{2}{c}{CIFAR-10} \\
\cmidrule{3-4} \cmidrule{5-6}
& & MSE & Success & MSE & Success \\
\cmidrule{1-6}
---
& Nearest neighbor (NN) oracle & $0.0232$ & -- & $0.0291$ & --\\
---
& Default hyper-parameters and architectures (cf.\ \Cref{ssec: exp_setup})& $0.0089$ & \checkmark & $0.0049$ & \checkmark\\
Fixed set size & Change size of fixed set to: $1K$ (MNIST) $50K$ (CIFAR-10 + shadows from CIFAR-100) & $0.0094$ & \checkmark &  $0.0039$ & \checkmark \\
Size \& architecture & Larger MLP (MNIST) and CNN (CIFAR-10) & $0.0079$ & \checkmark &
$0.0047$ & \checkmark\\
Released layers & Restrict attack to use subset of released model layers & $0.0124$ & \checkmark & $0.0257$ & \checkmark \\
Epochs & Increase number of released model training epochs: $250$ (MNIST) $200$ (CIFAR-10) & $0.0121$ & \checkmark & $0.0094$ & \checkmark \\
Activation & Change released model activations to ReLU & $0.0182$ & \checkmark &  $0.0324$ & \xmark \\
Learning rate & Decrease released model learning rate: $0.01$ (MNIST) $0.001$ (CIFAR-10) & $0.0049$ & \checkmark & $0.0055$ & \checkmark \\
Random initialization & Adversary does not know initial released model parameters & $0.0695$ & \xmark & $0.0931$ & \xmark \\
Model access & Only allow logit-based black-box access to released model & $0.0110$ & \checkmark & $0.0198$ & \checkmark \\
\bottomrule
\end{tabular}
}
\end{table*}

\section{Empirical Studies in Reconstruction}
\label{sec:experiments}

We now conduct extensive experiments investigating
how the released model architecture and its training hyperparameters impact reconstruction quality.
We first demonstrate the feasibility of the
reconstruction attack against models trained on our
default experimental setup.
Then we discuss an in-depth study on which factors, such
as training set size or released model's hyperparameters,
affect the success of reconstruction.
Finally, we investigate DP as a mitigation against reconstruction attacks.
Our findings are summarized in \Cref{tab:summary}.

\subsection{Feasibility of Reconstruction Attacks}
\label{ssec: flagship_exp}

We first carry out the general reconstruction attack
under the default experimental settings
(cf.\ \Cref{ssec: exp_setup}).

\Cref{fig:flagship-experiment} shows
examples of targets and respective reconstructions;
we use the nearest neighbor (NN) oracle
as a baseline.
We observe a good overall reconstruction quality on both datasets.
Running the attack against 1K test targets, we observe an average reconstruction MSE of $0.0089$ (MNIST) and $0.0049$ (CIFAR-10).
These numbers, compared to the NN oracle baselines,
demonstrate our attack is effective.
To account for the variance across experimental runs (e.g. different random selections of fixed sets across experiments), we repeated this experimental procedure ten times with differing fixed sets, initial released model parameters, and evaluation sets.
We saw minimal variance in results; importantly, reconstructions were consistently better than the NN oracle.

To help the reader calibrate MSE values
to reconstruction quality, \Cref{fig:flagship-experiment}
shows poor reconstructions with MSE close to the oracle NN's MSE (third column)
and to its 1st percentile (fourth column);
these reconstructions were obtained in preliminary experiments
with weaker RecoNN instances.
\arxiv{More examples are in \Cref{fig:cifar10_more_examples}.}

\paragraph*{Relation between reconstruction metrics}
With the same experimental setup as above,
we also evaluate results across
our other metrics (\Cref{ssec: metrics}) on MNIST.
We observe that MSE
and LPIPS are strongly correlated
(\Cref{fig:mnist_attack_train_size_l2_lpips}).
\Cref{fig:mnist_attack_train_size_l2_kl} also shows that a small MSE implies a small KL but the converse is not true;
in other words, it is possible for two images that are not identical to have similar predictions.
Since these metrics exhibit significant correlations, we only report a subset of them in subsequent experiments, focusing mostly on comparing the MSE metric with the NN oracle.
We observe similar trends on CIFAR-10, although MSE vs LPIPS correlation is weaker; this partially motivated including the LPIPS loss when training RecoNNs on CIFAR-10\arxiv{ (c.f. \Cref{app: cifar10_correlation})}.

\subsection{What Factors Affect Reconstruction}
\label{ssec:reconctruction-factors}
We study which factors may improve or impact reconstruction success;
these are summarized in \Cref{tab:summary}.

\paragraph{Attack training set size}

\begin{figure*}[t]
\captionsetup{width=0.99\textwidth, justification=centering}
  \centering
\begin{subfigure}[t]{.2\textwidth}
\centering
    \includegraphics[width=0.99\linewidth]{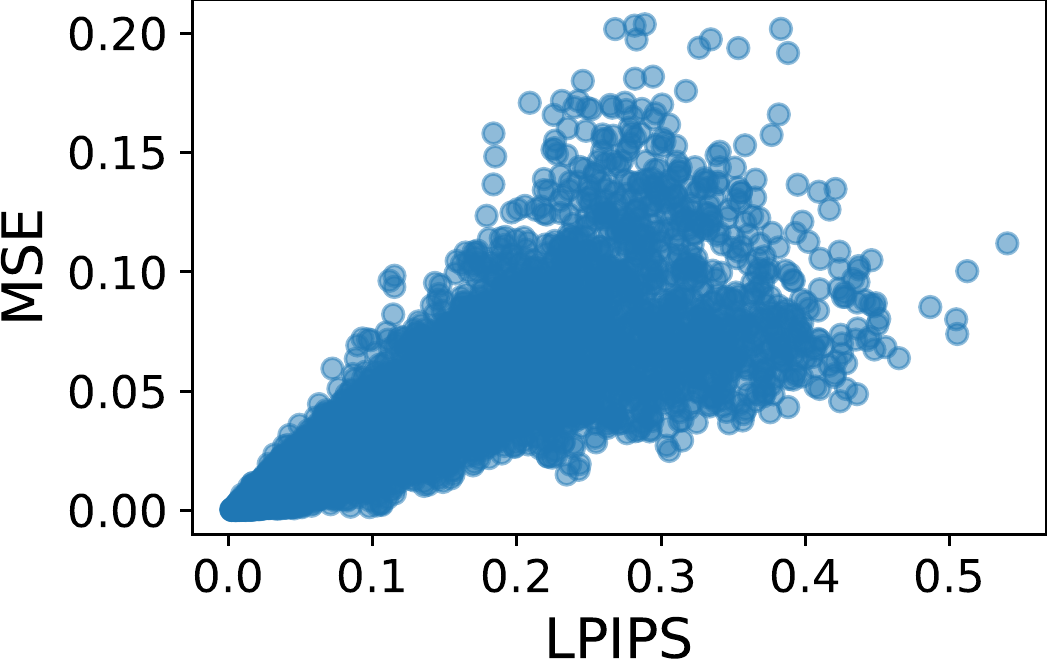}
    \caption{Correlation between MSE and LPIPS.}
    \label{fig:mnist_attack_train_size_l2_lpips}
\end{subfigure}\begin{subfigure}[t]{.2\textwidth}
\centering
    \includegraphics[width=0.99\linewidth]{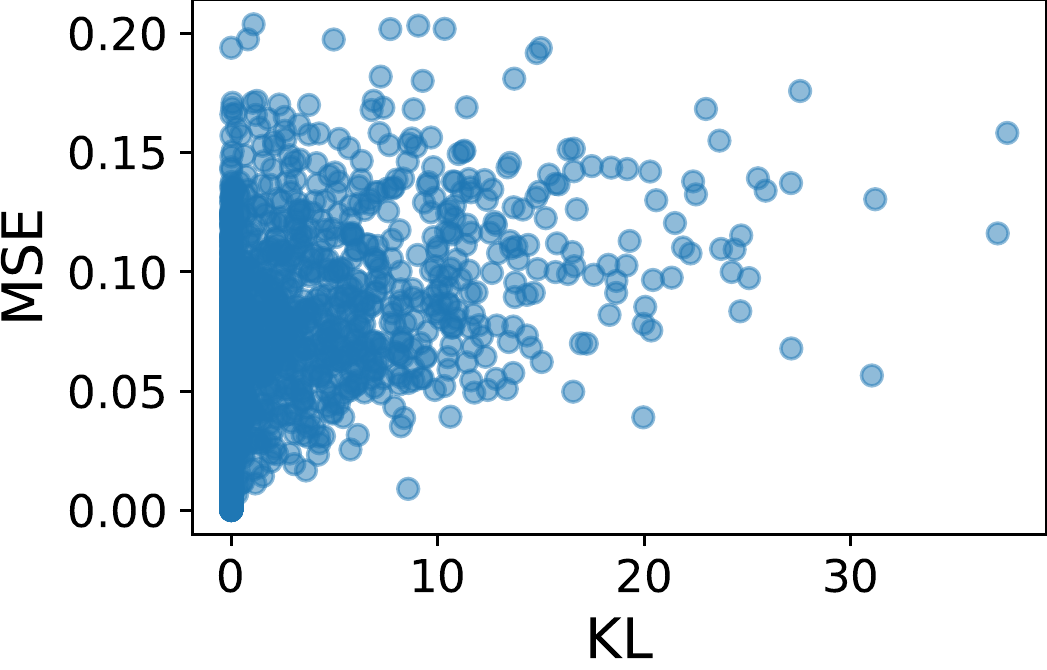}
    \caption{Correlation between MSE and KL metric.}
    \label{fig:mnist_attack_train_size_l2_kl}
\end{subfigure}\begin{subfigure}[t]{.2\textwidth}
\centering
    \includegraphics[width=0.99\linewidth]{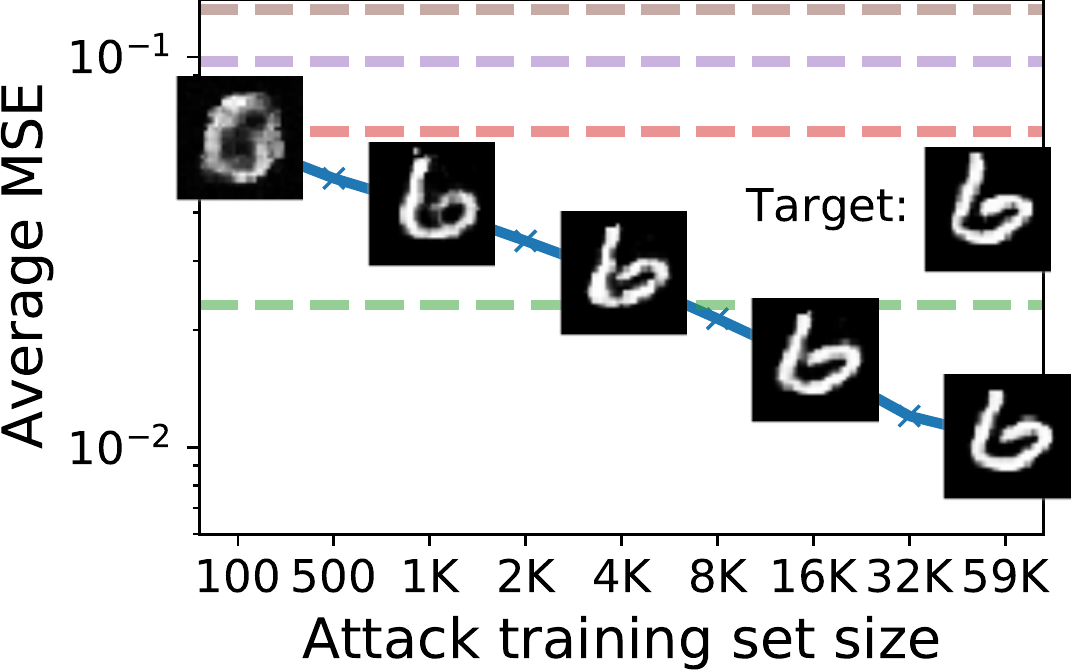}
    \caption{Average MSE against attack training set size.}
\label{fig:mnist_attack_train_size_l2_train_size}
\end{subfigure}\begin{subfigure}[t]{.2\textwidth}
\centering
    \includegraphics[width=0.99\linewidth]{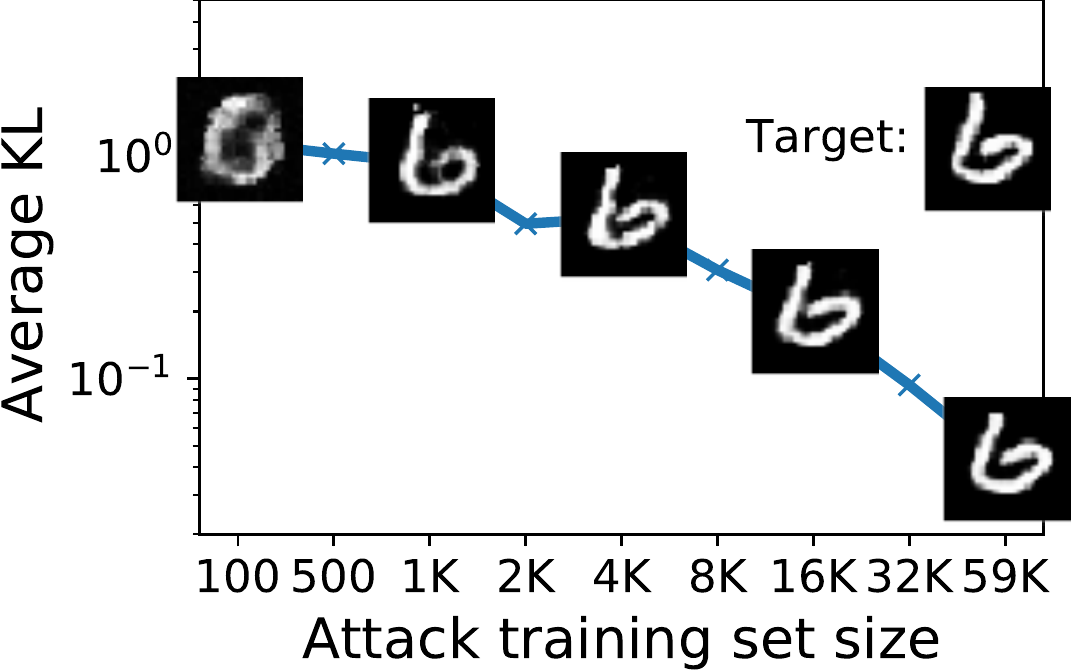}
\caption{Average KL against attack training set size.}
    \label{fig:mnist_attack_train_size_kl_train_size}
\end{subfigure}\captionsetup{width=0.99\textwidth}
\caption{
Correlation between metrics, and quality of reconstruction as a function of the number of shadow models.
}
\label{fig:mnist_attack_train_size}
\end{figure*}

Recall that the general reconstruction attack assumes the
attacker has access to $k$ shadow data points
$\Daux$
from the same distribution
as the target point.
From this knowledge, the attacker generates a collection of
shadow model-target pairs (the attack training data),
which is used to train the RecoNN.
Note that the size of the attack training data depends
both on the knowledge of the attacker (simply, the attacker may
not have access to many examples),
and on their
computational power:
they need to train one shadow model per data point
to create the attack training data.

We explore the fidelity of reconstruction on MNIST as
the amount of attack training data $k$ ranges from
$100$ to $59K$.
\Cref{fig:mnist_attack_train_size_l2_train_size} shows the average MSE between reconstructions over the $1K$ released model targets as $k$ varies.
Clearly, the attack becomes better as more training data is available. 
However, high fidelity reconstructions occur already with $1K$ shadow models; in our plots we include reconstructed examples at different values of $k$ illustrating this.
Reconstructions that are (on average) better than the NN oracle only require $8K$ shadow models.
Because the correlation between MSE and KL is not symmetric, we also plot the average KL against attack training set size and observe a similar monotonic decrease (\Cref{fig:mnist_attack_train_size_kl_train_size}).
We observe similar trends on CIFAR-10 when increasing the attack training set size; $5K$ shadow models is enough to generate reconstructions below the 1st percentile oracle MSE (\texttildelow0.05) and $10K$ shadow models will generate reconstructions below the NN oracle MSE (\texttildelow0.03).
See \Cref{app: transfer_learning} for full results on CIFAR-10.

\paragraph{Out-of-distribution (OOD) data on CIFAR-10}
\label{app:cifar10-ood}

The previous experiment indicates that reconstructions are poor when an adversary has relatively little side-information ($<1K$ points)
to create shadow models.
We now investigate if these additional points
must come from the same distribution as the fixed set and target sample.
If the attack succeeds even when $\Daux$ comes from a different distribution, they can potentially create a larger pool of shadow targets for the attack.
In addition, when reasonable OOD data is scarce or not available, the attacker could instead use $\Dfixed$ to train a generative model and use it to generate shadow targets from a similar distribution.

To relax the assumption that shadow targets
come from the same distribution
as the released model's training data
we use CIFAR-100, a standard OOD benchmark for CIFAR-10 \cite{DBLP:journals/corr/abs-2106-03004}, to construct the adversary's side knowledge.
\main{
In this experiment, the fixed data $\Dfixed$ and the 1K test targets are still selected from CIFAR-10,
but the shadow targets in $\Daux$ are OOD points corresponding to images sampled from CIFAR-100 annotated with a random CIFAR-10 label.
We measure attack success on the 1K released models with in-distribution CIFAR-10 targets.
As we observe a negligible difference in MSE between the two cases, we conclude that the success of the attack does not require having access to the correct prior distribution.
We exploit OOD data in later experiments when evaluating how the size of the fixed set affects reconstruction.
}

\arxiv{
\begin{figure}[H]
\captionsetup{width=0.5\textwidth}
  \centering
\begin{subfigure}[t]{.16\textwidth}
\centering
    \includegraphics[width=0.99\linewidth]{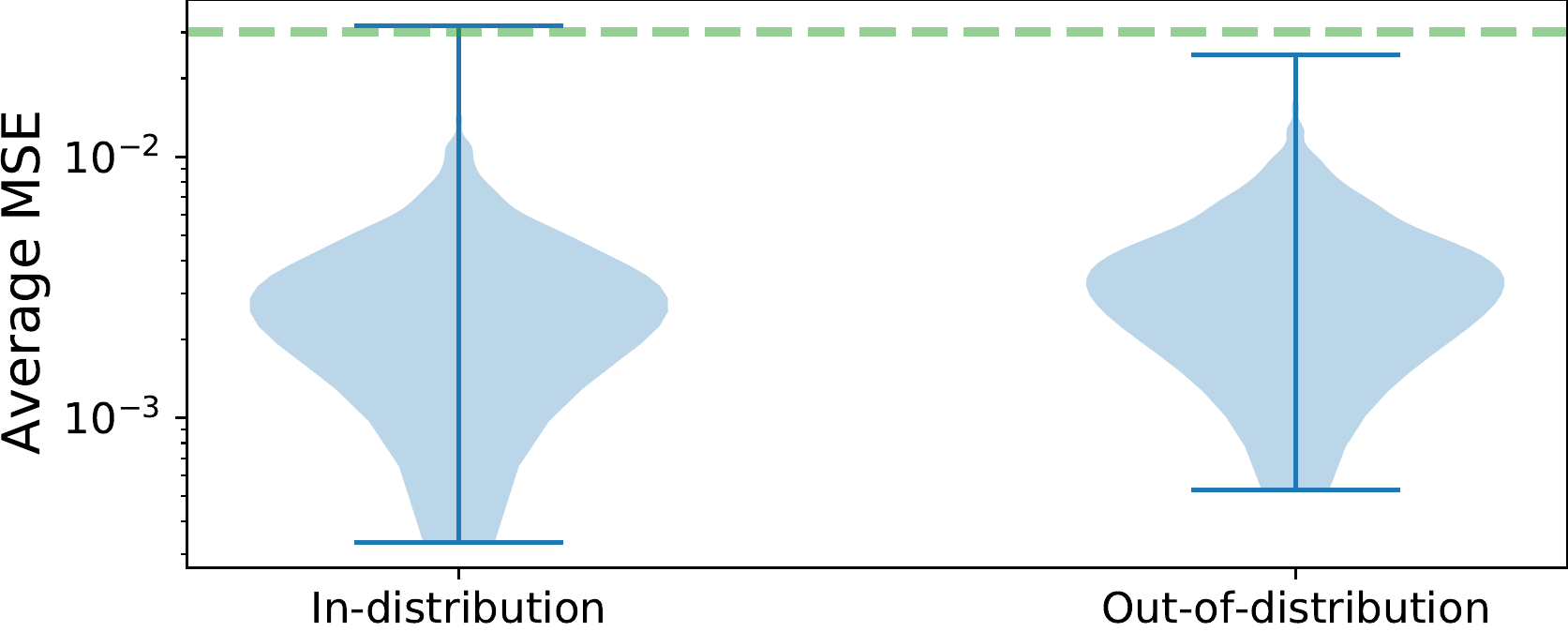}
\label{fig:}
\end{subfigure}\begin{subfigure}[t]{.16\textwidth}
\centering
    \includegraphics[width=0.99\linewidth]{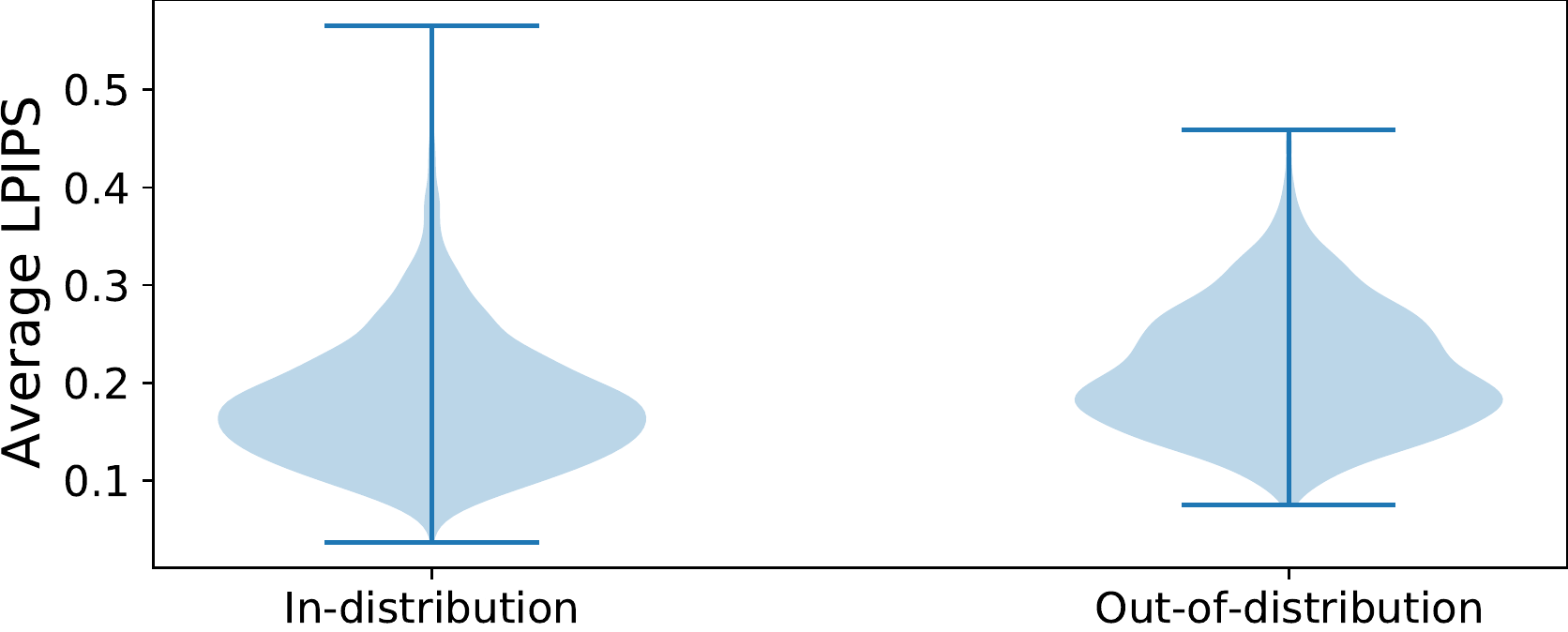}
\label{fig:}
\end{subfigure}\begin{subfigure}[t]{.16\textwidth}
\centering
    \includegraphics[width=0.99\linewidth]{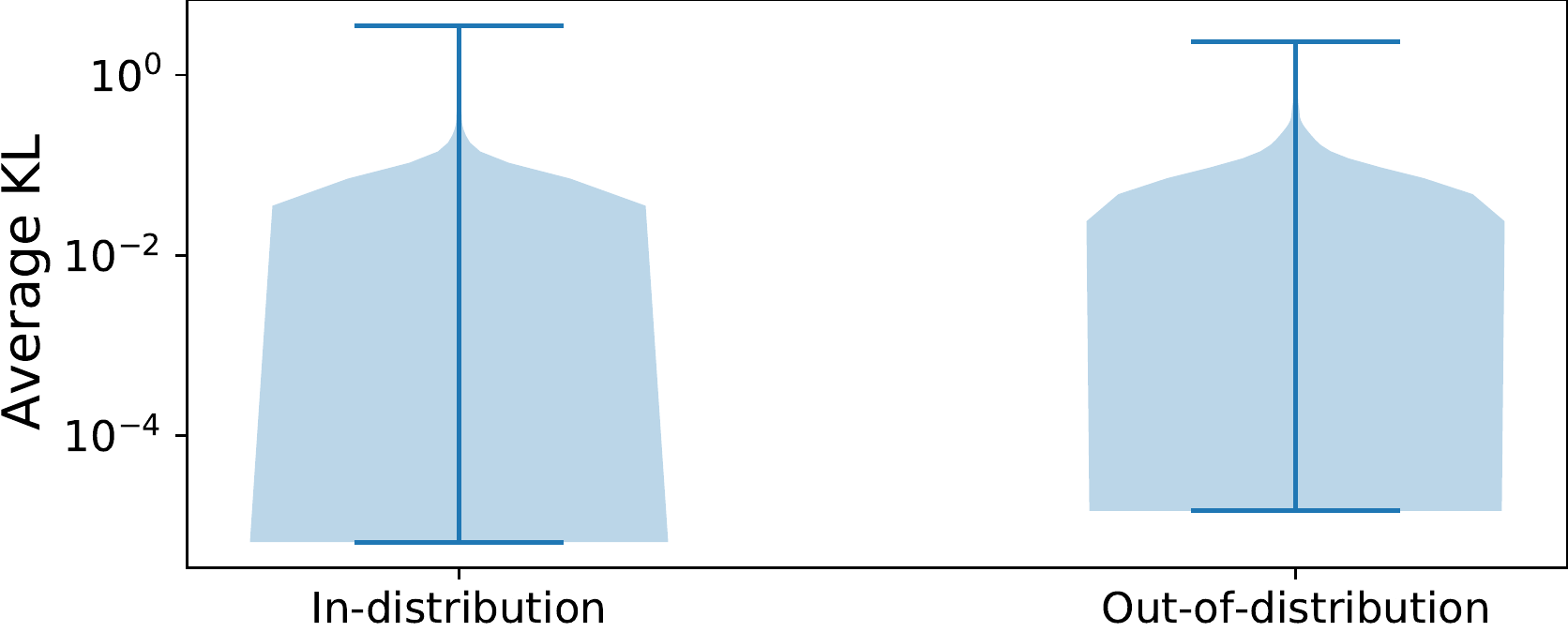}
\label{fig:}
\end{subfigure} \caption{Out-of-distribution experiments. Shadow models are created by either sampling the target from in-distribution (CIFAR-10) or out-of-distribution (CIFAR-100) data.}
\label{fig:cifar10_ood}
\end{figure}
\Cref{fig:cifar10_ood} shows
the difference between creating an attack training set with a fixed dataset $\Dfixed$ sampled from in-distribution data (CIFAR-10) and the additional shadow target points sampled from either in-distribution data (CIFAR-10) or out-of-distribution data (CIFAR-100).
We measure attack success on the 1$K$ released models with in-distribution targets (CIFAR-10).
We observe a negligible difference between the two, and so conclude the success of the attack is not predicated on access to the correct prior distribution.
We further exploit OOD data in the next part, when
we evaluate how the size of the fixed set
affects reconstruction.
}

\paragraph{Influence of training hyper-parameters}
\Cref{tab:summary} summarizes
what factors in training affect reconstruction.
The appendix expands upon these and gives empirical insights.

\noindent \textit{Fixed set size.} We measure the role of the fixed set size
by reducing from $10K$ to $1K$ (MNIST) and increasing from $5K$ to $50K$ (CIFAR-10).
We observe almost no difference in MSE in both cases; e.g.\ CIFAR-10 target points can be reconstructed even if there are $50K$ other points in the training set.

\noindent \textit{Model size and architecture.} We assess whether the size and architecture of the released model affect reconstruction. 
For MNIST, we increase the size of the hidden layer from $10$ to $100$; this increases the number of trainable parameters tenfold.
For CIFAR-10, we double the size from $50K$ to $100K$ by increasing the width of the first linear layer.
The rest of the architecture is kept to the defaults (\Cref{tab: cifar10_released_model}).
We observe almost no difference in reconstruction success
when attacking these larger released models.
Nevertheless, this attack has a bigger
computational cost: the size
of the RecoNN for CIFAR-10 increases
from $226M$ to over $400M$ parameters.

\noindent \textit{Layers.} Instead of allowing the RecoNN to process all parameters from a released model, we restrict to only the second layer for MNIST and convolutional layers for CIFAR-10.
This significantly reduces the input size to the reconstructor network, by $98\%$ on MNIST and $84\%$ on CIFAR-10.
We observe that this does not substantially affect the reconstruction fidelity, demonstrating that memorization of training points is not localized to a specific layer or small group of neurons.

\noindent \textit{Epochs.} The number of epochs
has a small impact on reconstruction.
For both MNIST and CIFAR-10 there is a slight increase in MSE if we more than double the number of training epochs, although targets are still successfully reconstructed. 
\arxiv{We investigate this relationship in more detail in \Cref{app: cifar10_factors} and  \Cref{app: finegrained_cifar10}.}
\main{We investigate this relationship in more detail in \Cref{app: finegrained_cifar10}.}

\noindent \textit{Activation.}
One may wonder why we used ELU activations in the released model instead of the more common ReLUs. 
We noticed that released models with ReLU activations tend to be harder to attack in comparison to other activation functions,
resulting in poor quality reconstructions on CIFAR-10 (i.e. MSE larger than the NN oracle).
It is well known that ReLUs induce sparse gradients; we observed that $>60\%$ of weights are not updated during training when the loss is computed with respect to the target. 
We suspect this is why RecoNN is less effective against ReLU activated models:
there is less mutual information between the model parameters and the target in comparison to models trained with other activations.
We discuss this in further detail in \Cref{app: relu_act}.

\noindent \textit{Learning rate.} Decreasing the learning rate of the released model did not affect the attack in the deterministic training setting. 
If randomness is introduced via mini-batch sampling, we will see that the learning rate impacts reconstruction.
\arxiv{In \Cref{app: mnist_batch_size_lr_epoch_fixed_size}, we better investigate the role of learning rate;
we find that a larger rate can harm the success of the attack in settings where the released is trained with mini-batches.}

\paragraph{Randomness from data sub-sampling}

\begin{figure}[t]
\captionsetup{width=0.5\textwidth}
  \centering
\begin{subfigure}[t]{.24\textwidth}
\centering
    \includegraphics[width=0.99\linewidth]{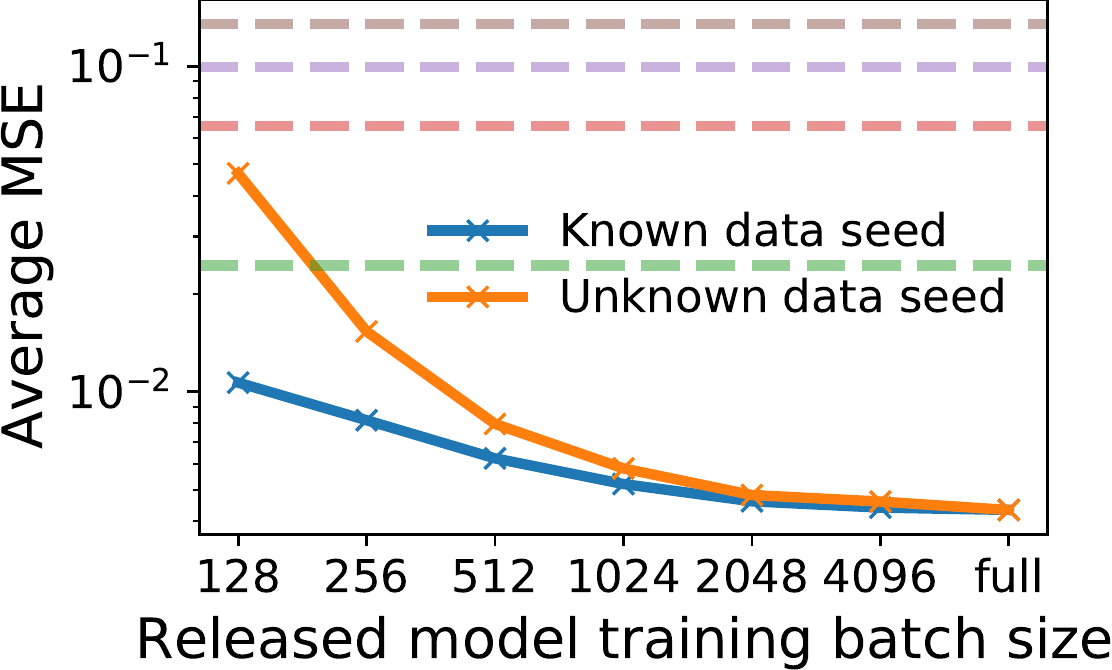}
\end{subfigure}\begin{subfigure}[t]{.24\textwidth}
\centering
    \includegraphics[width=0.99\linewidth]{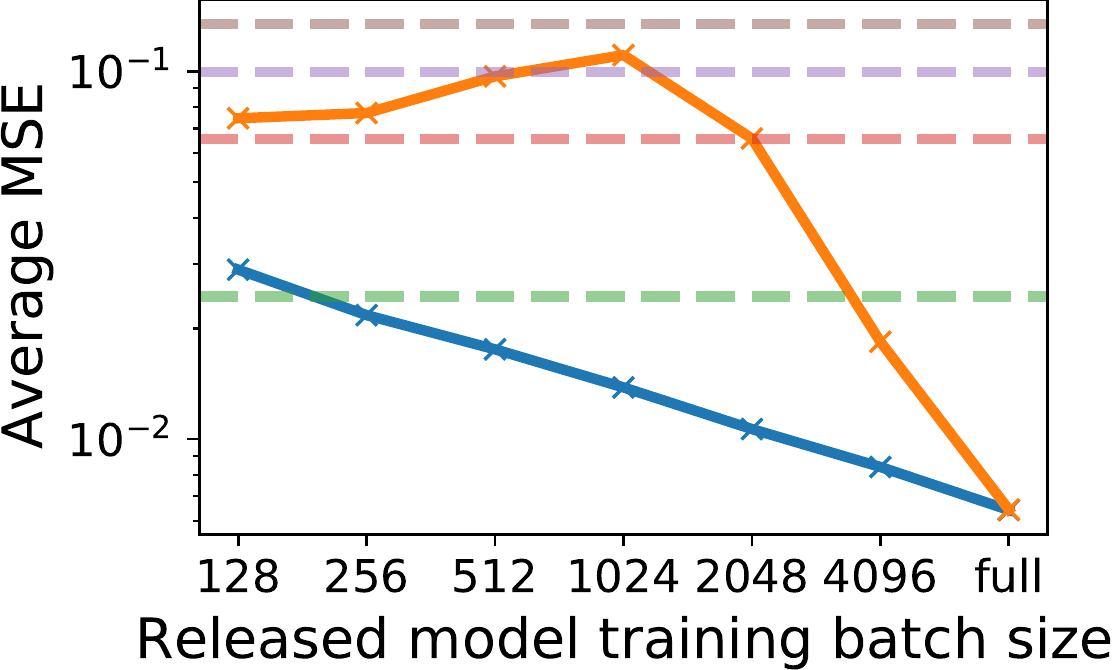}
\end{subfigure}\caption{MSE and released model training batch size when the adversary knows/does not know the data sub-sampling random seed. MSE is sensitive to the learning rate and momentum.
Learning rate: 0.01 (left), 0.2 (right).
Momentum: 0 (both).
} 
\label{fig:mnist_batch_size_main}
\end{figure}

\begin{figure}[t]
\captionsetup{width=0.5\textwidth}
  \centering
\begin{subfigure}[t]{.12\textwidth}
\centering
    \includegraphics[width=0.99\linewidth]{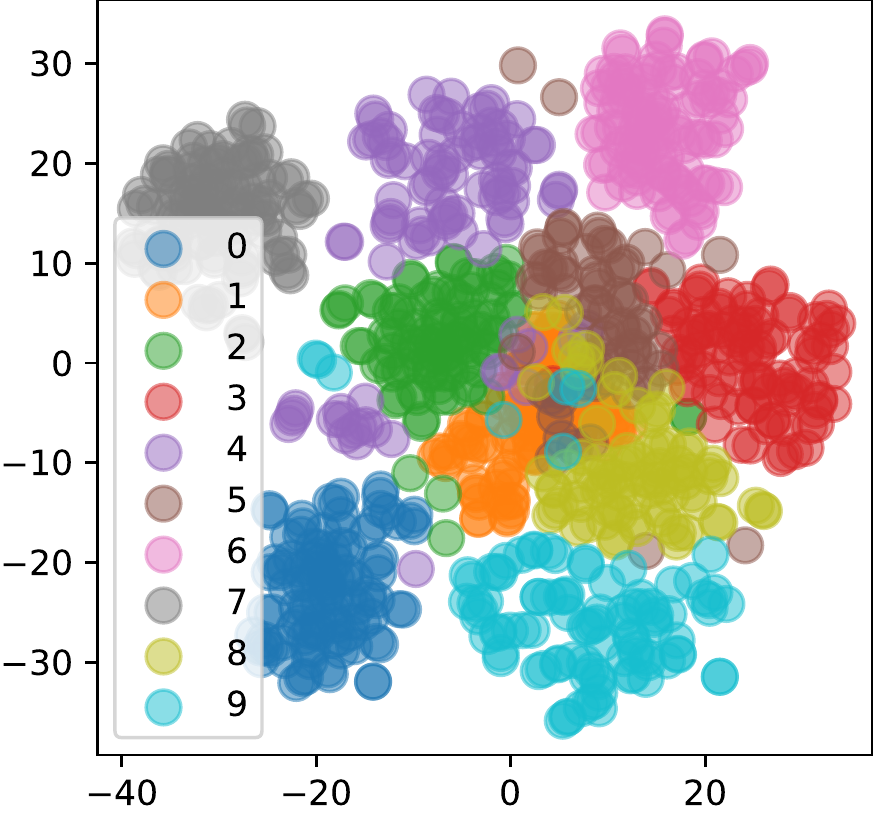}
\end{subfigure}\begin{subfigure}[t]{.12\textwidth}
\centering
    \includegraphics[width=0.99\linewidth]{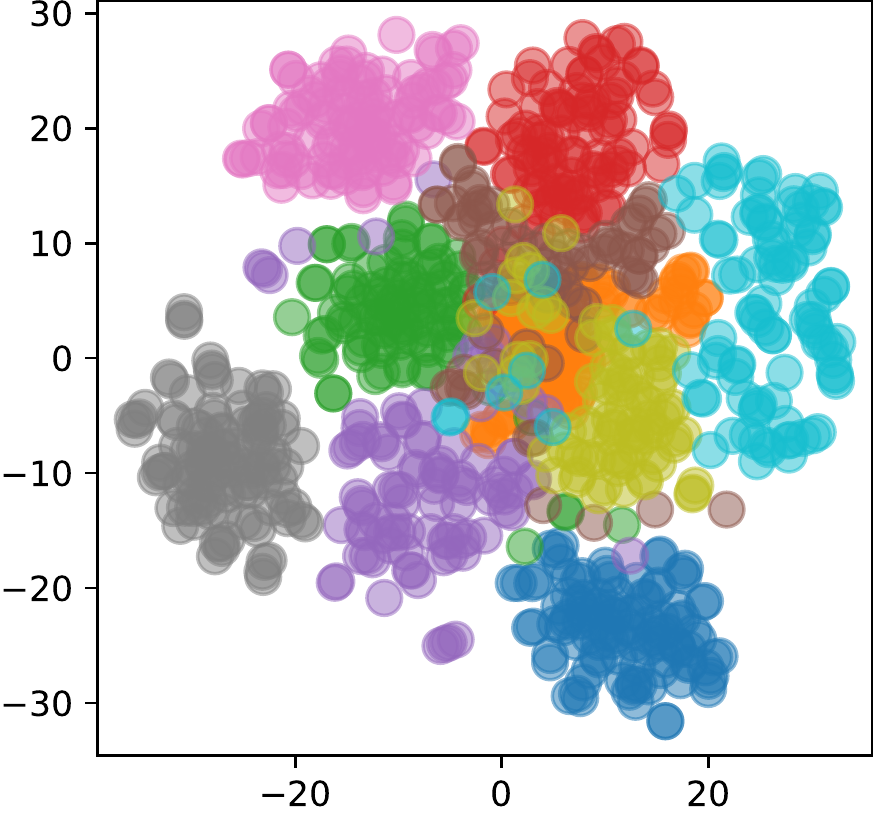}
\end{subfigure}\begin{subfigure}[t]{.12\textwidth}
\centering
    \includegraphics[width=0.99\linewidth]{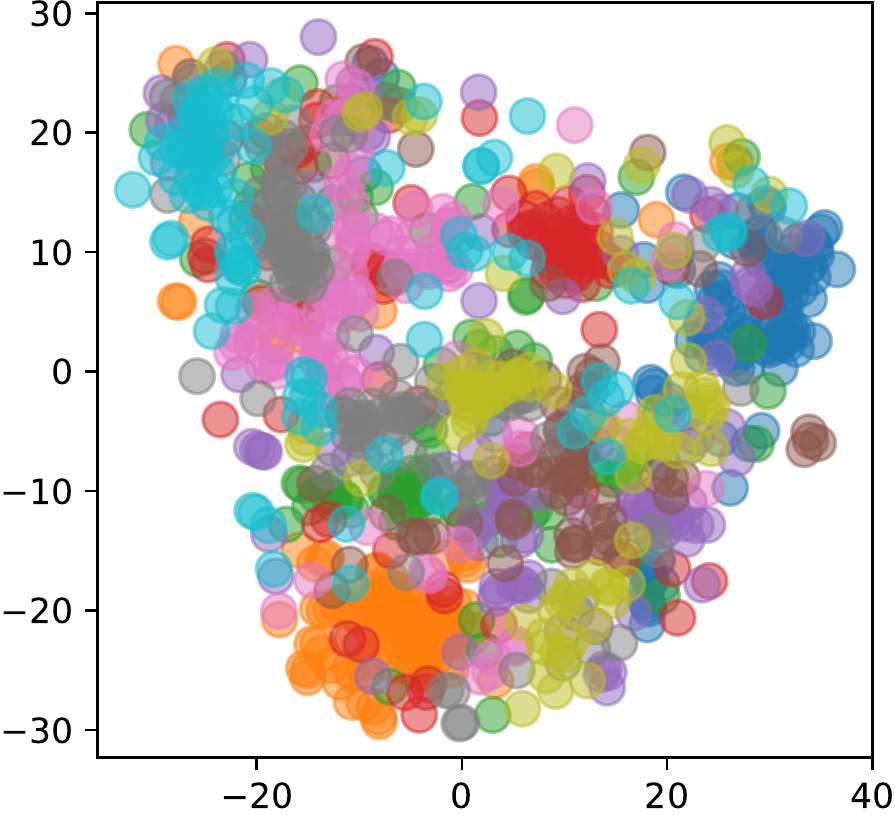}
\end{subfigure}\begin{subfigure}[t]{.12\textwidth}
\centering
    \includegraphics[width=0.99\linewidth]{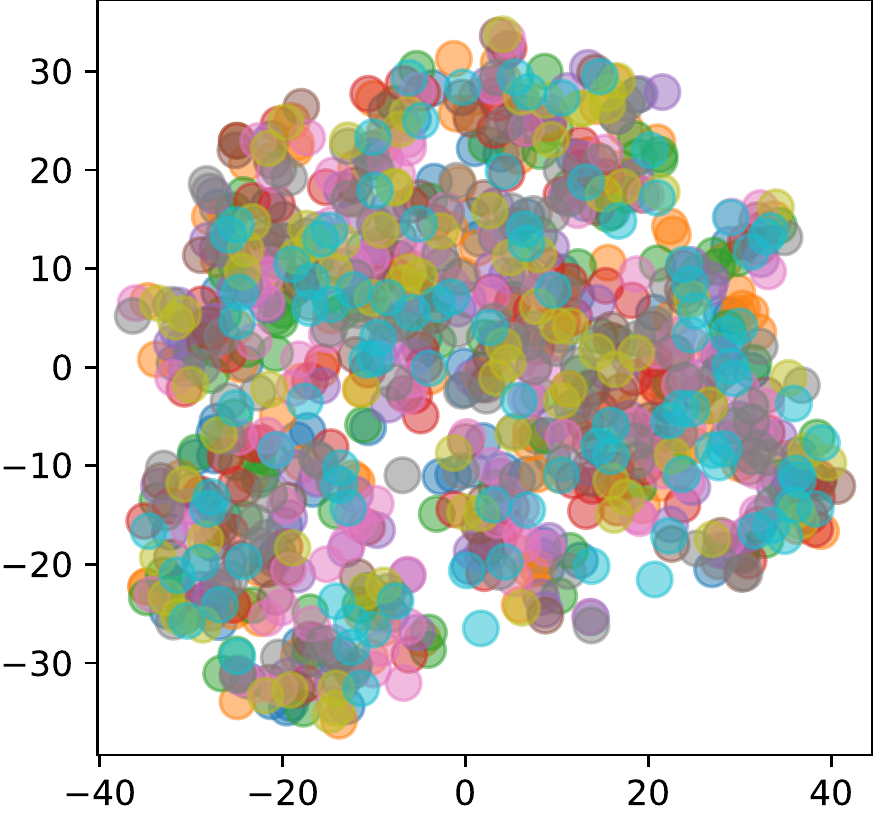}
\end{subfigure}\caption{TSNE embeddings of 1K released models trained with 1024 batch-size.
From left to right (known sub-sampling seed, learning rate):
(Yes, 0.01), (No, 0.01), (Yes, 0.2), (No, 0.2).
}
\label{fig:mnist_tsne}
\end{figure}

We explore how randomness stemming from data sub-sampling affects the attack on MNIST, by removing the assumption that the released model is trained with full batch gradient descent. 
We consider settings where the adversary knows the random seed used to shuffle the data
(this corresponds to SGD but with no randomness),
and settings where the adversary does not know the random seed.
Results in \Cref{fig:mnist_batch_size_main}
indicate that when the adversary knows the data shuffling seed, reconstruction attacks are successful even for small batch sizes.
Without knowing the seed, attack success depends on the training hyper-parameters, such as the choice of the learning rate.
It appears that attacking models with randomness from sub-sampling is more difficult than determinstically trained released models, and that larger learning rates also increase the hardness of the reconstruction task.
Loss landscapes of neural networks are extremely non-convex and contain many local optima~\cite{DBLP:conf/nips/Li0TSG1}; if more randomness is introduced, this will increase the opportunity for different shadow models to reach different optima.
This increases the difficulty of reconstruction as these shadow models will not be representative of the optima attained by the released model, and training with a larger learning rate will exacerbate this issue.
In \Cref{fig:mnist_tsne}, we show plot TSNE embeddings of parameters for all $1K$ released models for each of the two learning rates given in \Cref{fig:mnist_batch_size_main} and the two randomness settings (known and unknown seed) for a batch size of $1024$.
We represent each released model with a color depending on the label of the respective target. 
For a small learning rate, labels are grouped together in both known and unknown seed settings, implying the local optima these models realize are similar;
this makes it easier for the RecoNN to learn and subsequently generalize to the released model.
Conversely, in the large learning rate setting there is a stark difference between known and unknown seed settings:
if the seed is known,
groupings of labels still happen,
and a successful attack is possible;
however, if the seed is unknown,
the local optima reached by each released model has less structure that
the reconstructor network can learn on.

\arxiv{In \Cref{app: mnist_batch_size_lr_epoch_fixed_size} we
show comprehensive results with more learning
rates and evaluated on more metrics.}

\paragraph{Randomness from model initialization}

We explore how initialization randomness can affect the attack on MNIST. 
Firstly, we remove the assumption that the adversary knows the initial parameters of the released model; in practice, this means training each released and shadow model with a new random seed controlling the model's initial parameters.
By default, each linear and convolutional layer is initialized with Lecun Normalization, which is the default in the Haiku library~\cite{haiku2020github}.
In our experiments, we evaluated
other common initialization procedures (e.g., Glorot, He),
which did not change any of our findings;
we omit these results.
We refer the reader to \Cref{fig:flagship-experiment} for visual inspection of reconstructions at the two error rates reported in \Cref{tab:summary}, and conclude that the attack is unable to successfully reconstruct without knowledge of initialization, as they are far larger than the NN oracle described in \Cref{ssec: metrics}.

One may conjecture that the current attack pipeline is not suitable for this setting:
we only train a single shadow model per shadow target, which may fail to capture the variance in shadow model parameters over different initializations for the same shadow target.
For this reason, we further created an attack training set of $5M$ shadow model-target pairs, consisting of $10K$ shadow targets, where each target has $500$ shadow models all differing in initial parameters. 
Even so, this approach did not improve the MSE reported in \Cref{tab:summary}.
In \Cref{app: init_randomness},
we discuss evidence suggesting that reconstruction may not be possible without knowing the initial released model parameters.
A similar observation was made by Jagielski et al. \cite{jagielski2020auditing}, who run attacks to find lower bounds of the privacy budget $\epsilon$ in DP-SGD.
They observed that the bounds become tighter with less randomness from model initialization.

\subsection{Black-box Access to Released Model}
\label{ssec: blackbox_exp}
\arxiv{
The attack assumes white-box access to the released model parameters, and so a natural question arises: can we construct an attack that achieves a similar MSE distance without white-box access?
Our attack is constructed by learning the relation between the released model parameters and the unknown target point; the white-box attack uses these parameters directly by flattening released model parameters, concatenating each layer, normalizing, and passing this to the reconstructor network.
However, we could instead use other representations that contain information about the released model.
}
We design a black-box attack by limiting the adversary's access to only the logits predicted by the released model. 
For each \attackin{}, using a set of 200 (500) images from $\Daux$ for MNIST (CIFAR-10), the adversary collects the logit outputs of each image, concatenates them together, and uses this as the feature representation of the model, instead of the flattened weights.
This reduces the dimensionality of the feature vector from 8K to 2K for MNIST and 55K to 5K for CIFAR-10.
The average MSE using this logit representation approach is $0.011$ for MNIST and $0.0198$ for CIFAR-10, which is only marginally worse than the MSE of white-box attacks with default settings, and still much better that the NN oracle. 
We conclude that black-box reconstruction attacks
are feasible and have comparable performance
to white-box ones.

\begin{figure}[t]
\captionsetup{width=1.\linewidth}
\centering
  \centering
  \includegraphics[width=0.7\linewidth]{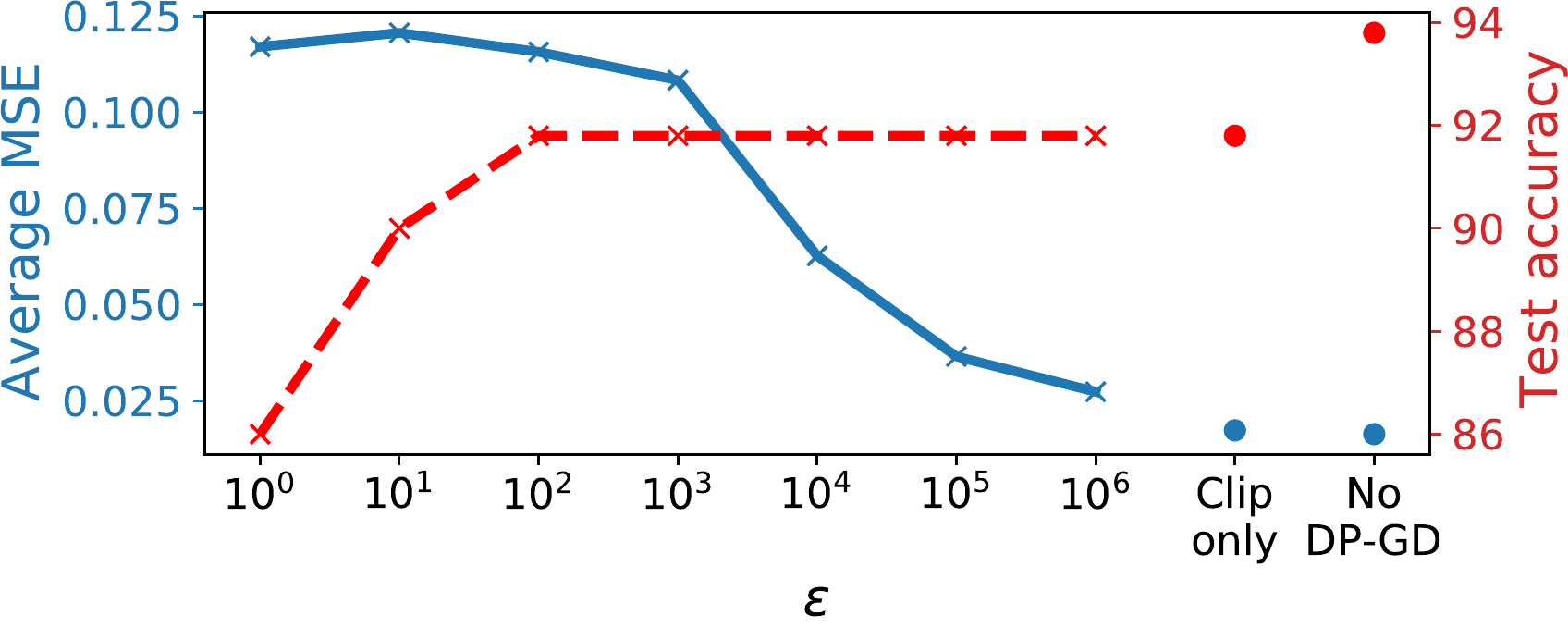}
  \caption{Average MSE of reconstructions and test accuracy of released model using $(\epsilon,\delta)$-DP on the MNIST dataset.}
  \label{fig:dpsgd_mnist}
\end{figure}

\begin{figure}[t]
\centering
  \centering
  \includegraphics[width=0.99\linewidth]{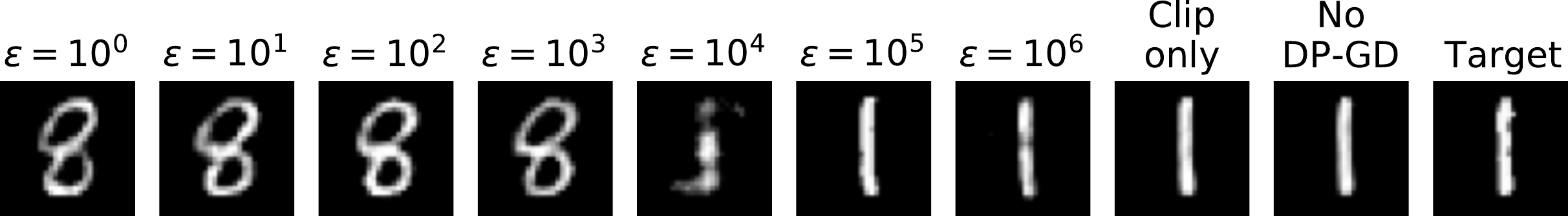}
    \caption{Example of MNIST reconstructions under DP.}
    \label{fig:dpgd_example}
\end{figure}

\subsection{Released Model Trained with Differential Privacy}
\label{ssec: mnist_dp_exp}

Having discussed what factors help and
hinder reconstruction,
we now evaluate on MNIST the resilience
of models trained with DP.
The released model training set-up is identical to
before (\Cref{ssec: attack_desc}), except we train with full batch DP gradient descent (DP-GD) with clipped gradients \citep{DBLP:conf/ccs/AbadiCGMMT016}.
Gradients are clipped to have a maximum $\ell_2$ norm of $1$, and Gaussian noise (unknown to the adversary) is added to make the model $(\epsilon, \delta)$-DP with $\delta = 10^{-5}$.
\Cref{fig:dpsgd_mnist} shows that even a large $\epsilon$
successfully mitigates reconstruction attacks, and that in these $\epsilon$ regimes the reduction in utility (measured by test accuracy) is negligible (\Cref{app: cifar10_dp} reports similar results on CIFAR-10).
Interestingly, for high levels of privacy, the reconstruction attack generates realistic but wildly incorrect reconstructions
(\Cref{fig:dpgd_example}).
These findings motivate our theoretical investigation into what level of DP is sufficient to protect against reconstruction attacks.

 \section{Towards Formal Guarantees Against Reconstruction Attacks}
\label{sec:defenses}

Mitigations that (provably) protect released models against reconstruction attacks
can (and should) be implemented within the training algorithm used by the model developer.
Protections that defend against effective reconstruction by informed adversaries will also protect against attacks by weaker, more realistic adversaries.
In this section, we propose a definition of \emph{reconstruction robustness} against informed adversaries, and compare it to the privacy guarantees afforded by DP.
As will soon become apparent, the strength of mitigations against reconstruction is necessarily going to be \emph{relative} to the strength of the prior information available to the adversary.

\subsection{Reconstruction Robustness}

Our main definition focuses on bounding the success probability of achieving accurate reconstruction by any (informed) adversary. The definition is parameterized by the side information available to the adversary, captured by a probabilistic prior $\prior$ from which the target $z$ is sampled, and by the adversary's goal expressed as a measure of reconstruction error $\lattackloss$.

\begin{definition}\label{def:rero}
Let $\pi$ by a prior over $\Zset$ and $\lattackloss : \Zset \times \Zset \to \R_{\geq 0}$ a reconstruction error function.
A randomized mechanism $M : \Zset^{n} \to \Theta$ is $(\eta, \gamma)$-ReRo (\emph{reconstruction robust}) with respect to $\pi$ and $\lattackloss$ if for any dataset $\Dfixed \in \Zset^{n-1}$ and any reconstruction attack $R : \Theta \to \Zset$ we have
\begin{equation}\label{eqn:rero-def}
    \Pr_{Z \sim \pi, \theta \sim M(\Dfixed \cup \{Z\})}[\lattackloss(Z, \reconstruct(\theta)) \leq \eta] \leq \gamma \enspace.
\end{equation}
\end{definition}

Suppose $M$ is an $(\eta, \gamma)$-ReRo mechanism.
The definition prevents any reconstruction attack with knowledge\footnote{Knowledge of $\pi$ and $D$ in the attack is implicit through the fact that \eqref{eqn:rero-def} has to hold for any reconstruction attack.} of the prior  $\pi$, the dataset $\Dfixed$ and the output $\theta = M(\Dfixed \cup \{Z\})$ to attain a reconstruction error lower than $\eta$ on an unknown target $Z \sim \pi$ with probability larger than $\gamma$.
A good ReRo mechanism is one with large $\eta$ and very small $\gamma$, i.e.\ one where even ``decent'' reconstructions are impossible with high probability.
In practice a tension between these two parameters is expected, at least for mechanisms providing some form of utility when computing a function depending on all the inputs.

Definition~\ref{def:rero} assumes the reconstruction attack is deterministic. We could consider randomized attacks instead, but note that determinism is not a limitation when trying to capture worst-case attacks: the $R$ that maximizes $\Pr[\lattackloss(Z, \reconstruct(\theta)) \leq \eta]$ is given by the (deterministic) \emph{maximum a posteriori} attack:
\begin{align*}
    R^*(\theta) = \argmax_{\hat{z} \in \Zset} \Pr_{Z \sim \pi}[\lattackloss(Z, \hat{z}) \leq \eta | M(\Dfixed \cup \{Z\}) = \theta] \enspace.
\end{align*}
Similarly, the definition protects against adversaries with full knowledge of the prior $\prior$. Since the optimal attack run by an adversary with a wrong prior is necessarily weaker than the optimal attack with a correct prior, assuming the adversary knows $\pi$ is preferable when designing mitigations.

Our main results provide two connections between reconstruction robustness and DP.
The first observation is that DP implies reconstruction robustness.
Quantitatively, we show that the ReRo parameters of a R{\'e}nyi DP (RDP) mechanism depend in a simple way on its privacy parameters and another quantity capturing the relation between $\prior$ and $\lattackloss$.
The second observation is that any mechanism that is robust against exact reconstruction with respect to a sufficiently rich family of priors supported on pairs of points must satisfy DP.
Together, both results stress the importance of correctly modelling an adversary's prior knowledge in effectively protecting against reconstruction attacks.
In particular, we show that very weak DP guarantees suffice to protect against reconstruction when the adversary has limited knowledge about the target point.

\subsection{From DP to ReRo}\label{sec:dp2rero}

We now show that differentially private mechanisms provide reconstruction robustness. Let us recall the definitions of approximate and R{\'e}nyi DP.

\begin{definition}[\cite{DBLP:conf/tcc/DworkMNS06,DBLP:conf/eurocrypt/DworkKMMN06,DBLP:conf/csfw/Mironov17}]
Let $M : \Zset^n \to \Theta$ be a randomized mechanism, $\epsilon > 0$, $\delta \in [0,1]$ and $\alpha > 1$. We say that:
\begin{enumerate}
    \item $M$ is $(\epsilon, \delta)$-DP if for any datasets $D, D' \in \Zset^n$ differing in a single record and any event $E \subseteq \Theta$ we have
    \begin{align*}
        \Pr[M(D) \in E] - e^{\epsilon} \Pr[M(D') \in E] \leq \delta \enspace.
    \end{align*}
    When $\delta = 0$ we simply say the mechanism is $\epsilon$-DP.
    \item $M$ is $(\alpha, \epsilon)$-RDP if for any datasets $D, D' \in \Zset^n$ differing in a single record we have
    \begin{align*}
        \E_{\theta \sim M(D')}\left[\left(\frac{\Pr[M(D) = \theta]}{\Pr[M(D') = \theta]}\right)^{\alpha}\right] \leq e^{(\alpha - 1) \epsilon} \enspace.
    \end{align*}
\end{enumerate}
\end{definition}

The effect of the prior on ReRo bounds obtained from DP is through an anti-concentration property. For prior $\pi$, error function $\lattackloss$ and error threshold $\eta$, define the \emph{baseline error} as
\begin{equation*}
  \kappa_{\pi, \lattackloss}(\eta) = \sup_{z_0 \in \Zset} \Pr_{Z \sim \pi}[\lattackloss(Z, z_0) \leq \eta] \enspace. 
\end{equation*}
When $\pi$, $\lattackloss$ or $\eta$ are clear from the context we may drop them to unclutter our notation.
Whenever $\lattackloss$ is a metric on $\Zset$, an upper bound on $\kappa$ provides a measure of \emph{anti-concentration} of the prior by guaranteeing that no single point has too much of probability mass concentrated around it;
bounds for $\kappa$ for some prior distributions are given in Section~\ref{sec:high-uncertainty}.
Another interpretation of $\kappa$ is as the success probability of the best \emph{oblivious} reconstruction attack that ignores the output of $M$.
By this interpretation, the next theorem says that if a mechanism is RDP, the best reconstruction attack cannot have success probability much larger than the best oblivious attack.

\begin{restatable}[]{theorem}{dprero}
\label{thm:dp2rero}
Fix $\pi$, $\lattackloss$ and $\eta > 0$, and let $\kappa = \kappa_{\pi,\ell}(\eta)$.
If a mechanism $M$ satisfies $(\alpha, \epsilon)$-RDP then it also satisfies $(\eta, \gamma)$-ReRo with respect to $\pi$ and $\lattackloss$ with $\gamma = \left(\kappa \cdot e^{\epsilon} \right)^{\frac{\alpha - 1}{\alpha}}$.
\end{restatable}

Taking $\alpha \to \infty$ and recalling that $(\infty,\epsilon)$-RDP is equivalent to $\epsilon$-DP \cite{DBLP:conf/csfw/Mironov17} we obtain the following corollary.

\begin{restatable}[]{corollary}{reropuredp}
\label{cor:rero-pure-dp}
Fix $\pi$, $\lattackloss$ and $\eta > 0$, and let $\kappa = \kappa_{\pi,\ell}(\eta)$.
If a mechanism $M$ satisfies $\epsilon$-DP then it also satisfies $(\eta, \gamma)$-ReRo with respect to $\pi$ and $\lattackloss$ with $\gamma = \kappa \cdot e^{\epsilon}$.
\end{restatable}

Another way to interpret Theorem~\ref{thm:dp2rero} is through the lens of zero-concentrated DP (zCDP) \cite{DBLP:conf/tcc/BunS16}.
A mechanism is $\rho$-zCDP if it satisfies $(\alpha, \alpha \rho)$-RDP for every $\alpha > 1$.
This definition provides a natural and convenient way to express the privacy afforded by the ubiquitous Gaussian mechanism \cite{DBLP:conf/eurocrypt/DworkKMMN06}.
Applying Theorem~\ref{thm:dp2rero} to a $\rho$-zCDP mechanism and optimizing $\alpha$ to minimize the upper bound yields the following.

\begin{restatable}[]{corollary}{rerozcdp}
\label{cor:rero-zcdp}
Fix $\pi$, $\lattackloss$ and $\eta > 0$, and let $\kappa = \kappa_{\pi,\ell}(\eta)$.
If a mechanism $M$ satisfies $\rho$-zCDP with $\rho < \log(1/\kappa)$
then it also satisfies $(\eta, \gamma)$-ReRo with respect to $\pi$ and $\lattackloss$ with $\gamma = e^{-(\sqrt{\log(1/\kappa)} - \sqrt{\rho})^2}$.
\end{restatable}

\subsection{From ReRo to DP}

Next we investigate the reverse implication: does a strong enough level of reconstruction robustness imply a standard definition of privacy protection like DP?
We show that this is indeed the case if one insists on protecting against \emph{exact} reconstruction simultaneously for a family of priors concentrated on pairs of data points.
From this lens, the result says that as soon as a mechanism exhibits strong enough reconstruction robustness to prevent membership inference it must necessarily satisfy DP.

Before stating the result we introduce the following notation.
Given $p \in (0,1)$ and $z, z' \in \Zset$, $z \neq z'$, let $\pi_{p,z,z'}$ denote the prior over $\Zset$ that assigns probability $p$ to $z$ and $1-p$ to $z'$.
We also let $\lattackloss_{0/1}(z, z') = \onesv[z \neq z']$.

\begin{restatable}[]{theorem}{rerotodp}
\label{thm:rero2dp}
Fix $\epsilon \geq 0$, $\eta \in (0,1)$ and $\gamma \in [0,1]$.
Let $\Pi_{\epsilon} = \{ \pi_{p, z, z'} : z, z' \in \Zset, z \neq z' \}$ be the class of all priors on $\Zset$ concentrated on pairs of points with $p = \frac{1}{e^{\epsilon} + 1}$.
If a mechanism $M : \Zset^n \to \Theta$ is $(\eta, \gamma)$-ReRo with respect to $\lattackloss_{0/1}$ and every prior $\pi \in \Pi_{\epsilon}$, then $M$ satisfies $(\epsilon,\delta)$-DP with $\delta = \max\{0, (e^{\epsilon} + 1) \gamma - e^{\epsilon}\}$.
\end{restatable}

\subsection{ReRo Against High-Dimensional, High-Uncertainty Priors}
\label{sec:high-uncertainty}

A standard ``rule of thumb'' says that DP only provides a meaningful protection when $\epsilon$ is a small constant.
On the other hand, our experiment on models trained with DP-SGD (\Cref{ssec: mnist_dp_exp}, \Cref{app: cifar10_dp}) shows that much larger values of $\epsilon$ are successful at mitigating our RecoNN-based attack.
This could be interpreted as a limitation of our attack in the presence of weak levels of DP protection.
An alternative explanation is that DP with large values of $\epsilon$ can protect against reconstruction attacks if the reconstruction target is high-dimensional and the adversary's prior knowledge contains a large degree of uncertainty.
We formalize this intuition by instantiating the bounds from~\Cref{sec:dp2rero} on two natural priors where $\kappa$ is easy to bound: uniform and Gaussian priors.
A similar analysis in the context of local DP was presented in \cite{DBLP:journals/corr/abs-1812-00984} (see \Cref{ssec:defense-related-work} for a detailed comparison).

\paragraph{Uniform priors}
Suppose training data points in $\Zset$ are represented by $d$-dimensional real vectors and all the adversary knows about the target point $z$ is a norm bound of the form $\norm{z}_2 \leq 1$. Then it makes sense for the adversary to take as prior the uniform distribution $\cU(B_1^d(0))$ over the Euclidean $d$-dimensional unit ball $B_1^d(0)$ centered at zero. For simplicity, suppose also that reconstruction error is measured in terms of the Euclidean distance $\ell_2$. Then we have the following.

\begin{restatable}[]{proposition}{uniprior}
\label{prop:uniprior}
Fix a constant $\eta \in (0,1)$.
Suppose $M$ is a mechanism satisfying $\epsilon$-DP with $\epsilon = o(d)$ or $\rho$-zCDP with $\rho = o(d)$.
Then $M$ is $(\eta, \gamma)$-ReRo with respect to $\cU(B_1^d(0))$ and $\ell_2$ with $\gamma = e^{-\Omega(d)}$.
\end{restatable}

This result shows that, in high-dimensional settings where an informed adversary's knowledge about the target datapoint is only in the form a syntactic constraint like $\norm{z}_2 \leq 1$, privacy parameters sub-linear in the dimension suffice to make the reconstruction success probability negligible. 

\paragraph{Gaussian priors}
Another natural prior to consider is a ($d$-dimensional, isotropic) Gaussian distribution $\cN(w,\sigma^2 I_d)$ specifying the adversary's prior knowledge about the location $w$ of the target point with some degree of uncertainty controlled by $\sigma$.
Taking again $\ell_2$ as the measure of reconstruction error, we obtain the following.

\begin{restatable}[]{proposition}{normalprior}
\label{prop:normalprior}
Fix a constant $\eta > 0$.
Suppose $M$ is a mechanism satisfying $\epsilon$-DP with $\epsilon = o(d)$ or $\rho$-zCDP with $\rho = o(d)$.
Then $M$ is $(\eta, \gamma)$-ReRo with respect to $\cN(w,\sigma^2 I_d)$ and $\ell_2$ with $\gamma = e^{-\Omega(d)}$ as long as $\sigma \geq \frac{2 \eta}{\sqrt{d}}$.
\end{restatable}

The idea that large values of $\epsilon$ can protect against reconstruction when the adversary's prior contains significant uncertainty (i.e.\ it is diffused) was previously noticed in \cite{DBLP:journals/corr/abs-1812-00984} in the context of local DP (LDP) with priors close to uniform.
Inspired by FL applications where adversaries get access to LDP gradients, the authors propose a notion of protection against \emph{reconstruction breaches} that is more stringent than \Cref{def:rero}: it asks that the adversary cannot effectively reconstruct a particular feature of interest about the target point no matter what the output of the mechanism is -- in contrast, ReRo uses an \emph{average-case} requirement over the outputs of the mechanism.
Technically, \cite[Lemma 2.2]{DBLP:journals/corr/abs-1812-00984} shows that the bound in Corollary~\ref{cor:rero-pure-dp} also holds for this worst-case notion of protection against reconstruction.\footnote{Although the bound in \cite{DBLP:journals/corr/abs-1812-00984} is stated in terms of $\epsilon$-LDP, it is easy to see that the same holds for central $\epsilon$-DP in the presence of an informed adversary.}
Such worst-case guarantees, however, are not attainable under relaxations of $\epsilon$-DP like RDP because the latter does not enforce an almost sure bound on the privacy loss: instead, it just guarantees that the privacy loss will be small with high probability.
Thus, \Cref{thm:dp2rero} and \Cref{prop:uniprior} are natural generalizations of the results from \cite{DBLP:journals/corr/abs-1812-00984} to RDP, which is the default notion of privacy provided by DP-SGD and other popular private ML algorithms \cite{DBLP:conf/iclr/PapernotSMRTE18,DBLP:journals/corr/abs-1908-10530}.

\subsection{Is Reconstruction Robustness Useful in Practice?}
\label{sec:rero-useful}

To deploy the bounds from Theorem~\ref{thm:dp2rero} two things are necessary: the description of a criterion for reconstruction error $\ell$ with an associated threshold $\eta$, and an understanding of the success rate of $\eta$-approximate reconstruction by the adversary prior to the release.
Equipped with $\ell$ and $\eta$, one can then engage in a conversation with stakeholders and domain experts to determine what success rate of reconstruction is reasonable to adjudicate to a potential adversary before the release is made.
An interesting feature of \Cref{thm:dp2rero} is that it reduces adversarial modelling to a question about determining a \emph{single number} $\kappa_{\pi,\ell}(\eta)$.
Furthermore, it is possible that one does not need to be overly conservative in estimating this number.
After all, the theorem bounds the success probability of the wost-case adversary which, in particular, knows all the fixed dataset.
Realistic adversaries will often have less knowledge of the fixed dataset, so it might be possible to trade-off knowledge of the fixed dataset with the amount of diffusion required from the prior. We leave this question for future work.

\subsection{Further Related Work}\label{ssec:defense-related-work}

\paragraph{Threat modelling and privacy semantics}\label{sec:threat-modelling}
The use of informed adversaries in formal privacy analyses can be tracked back to the \emph{sub-linear queries} (SuLQ) framework \cite{DBLP:conf/pods/BlumDMN05}. SuLQ was later subsumed by DP \cite{DBLP:conf/tcc/DworkMNS06}, where mentions to a concrete adversary were expressly avoided in the definition that is widely used nowadays \cite{mcsherry_2021}.
Nonetheless, \cite[Appendix A]{DBLP:conf/tcc/DworkMNS06} provides a ``semantically flavored'' definition equivalent to DP which involves the likelihood ratio between the prior and posterior beliefs of an informed adversary about any property of the target data point.
The adversarial model put forward in Section~\ref{sec:threat-model} uses the same notion of informed adversary.

In other frameworks where the adversary is not (necessarily) informed (e.g.\ Pufferfish privacy \cite{DBLP:journals/tods/KiferM14} and inferential privacy \cite{DBLP:conf/innovations/GhoshK17}), side knowledge about the whole dataset is encoded in a probabilistic prior capturing information about the individual entries in the dataset as well as their statistical dependencies.
These frameworks extend the semantic approach to DP by replacing the prior-vs-posterior condition with an odds ratio condition -- such modification is motivated by the observation that prior-vs-posterior bounds cannot hold in general for uninformed adversaries unless the prior distribution over the dataset assumes the records are mutually independent.
Alternatively, \cite{DBLP:journals/jpc/Kasiviswanathan14} provides posterior-vs-posterior semantics for DP in the presence of an uninformed adversary with an arbitrary prior.
In the definition of reconstruction robustness, our use of an informed adversary with a prior over the target data point circumvents the complications arising from dependencies between points in the training data: the prior captures the adversary's \emph{residual} uncertainty about the target point after observing the fixed dataset.
On the opposite direction, several authors have proposed approaches where the adversary's uncertainty with respect to the input data of a mechanism is leveraged to increase the privacy provided to individuals \cite{DBLP:conf/cikm/Duan09,DBLP:conf/asiacrypt/BhaskarBGLT11,DBLP:conf/focs/BassilyGKS13,DBLP:journals/corr/abs-1905-00650}.
Implicitly, these works assume a less powerful adversary than the one considered in this paper.

Most of the semantic definitions we discussed formalize the privacy protection goal without assuming the adversary is interested in a particular inference task; that is, protection applies simultaneously to all possible inferences about the target point(s).
In contrast, the use of an explicit reconstruction error $\lattackloss$ makes the definition of reconstruction robustness syntactic in nature.
\Cref{ssec:attack-protocol} briefly discusses how the problem of designing an appropriate error function for each application can be approached.
A similar dilemma arises in location privacy, where distortion-based notions include an explicit measure of reconstruction error \cite{DBLP:conf/wpes/ShokriFJH09,DBLP:conf/sp/ShokriTBH11}.
Nonetheless, as Theorem~\ref{thm:rero2dp} shows, by considering a very stringent reconstruction goal and a set of sufficiently informative priors one can recover semantic privacy notions from reconstruction robustness.

The connection between DP and protection against membership inference is perhaps best understood via its hypothesis testing interpretation \cite{wasserman2010statistical,DBLP:conf/icml/KairouzOV15}.
A comprehensive discussion of the adversary \emph{implicit} in the definition of DP from the hypothesis testing standpoint can be found in \cite{DBLP:conf/sp/NasrSTPC21}.
Interestingly, \cite{DBLP:conf/aistats/BalleBGHS20} shows that, unlike standard DP, RDP does not admit a hypothesis testing interpretation.
A semantic (Bayesian) interpretation of RDP in terms of moment bounds on the odds ratio is presented in \cite{DBLP:conf/csfw/Mironov17}.
Theorem~\ref{thm:dp2rero} provides an alternative characterization of the privacy protection afforded by RDP in terms of resilience to reconstruction attacks.

\paragraph{DP and protection against reconstruction}

How standard DP offers concrete protection against reconstruction attacks has been studied in other contexts.
Indeed, one of the original motivations for the definition of DP was to defeat database reconstruction attacks in the context of interactive query mechanisms \cite{DBLP:conf/pods/DinurN03,DBLP:conf/stoc/DworkMT07,dwork2017exposed,DBLP:journals/jpc/CohenN20,dpblogreconstruction}.
In such attacks, the adversary receives (noisy) answers to a sequence of specially crafted queries against a database and, if the noise is small enough, uses the answers to (partially) reconstruct every record in the database.
The success of these attacks is contingent on the adversary's ability to control these queries; in contrast, in ML applications like the ones we consider the computation performed by the mechanism is completely under the model developer's control.

The quantitative information flow literature seeks to provide information-theoretic bounds on data leakage in information processing systems \cite{DBLP:conf/fossacs/Smith09,alvim2020science}.
When applied to differentially private mechanisms, these ideas yield bounds on the protection against \emph{exact} reconstruction when $\Zset$ is finite.
In particular, when specialized to informed adversaries and translated into our terminology, \cite[Theorem 3]{DBLP:conf/birthday/ElSalamounyCP14} shows that any $\epsilon$-DP mechanism is $(\eta, \gamma)$-ReRo with $\eta \in (0,1)$ with respect to $\lattackloss_{0/1}$ and any prior $\pi$ with
$\gamma \leq \frac{|\Zset| \kappa e^{\epsilon}}{|\Zset| + e^{\epsilon} - 1}$.
Taking $|\Zset| \to \infty$ recovers the bound from Corollary~\ref{cor:rero-pure-dp} in the case of $\lattackloss_{0/1}$.
Our results can thus be interpreted as a generalization of this line of work where no assumptions about $\Zset$ are necessary.

 \section{Conclusions}

Our work provides compelling evidence that standard ML models can memorize enough information about their training data to enable high-fidelity reconstructions in a very stringent threat model.
By instantiating an informed adversary that
learns to map
model parameters to training images, we successfully attacked standard MNIST and CIFAR-10 classifiers with up to $100K$ parameters, and showed the attack is significantly robust to
changes in the training hyper-parameters.
Two aspects of our attack we would like to improve in future work are its data and computational efficiency, and its scalability to larger, more performant released models.
This would not lead to real-world adversaries mounting practical attacks due to the nature of our threat model, but it would enable model developers to assess potential privacy leakage in models before deployment.
Extending our attacks to reconstruct $N > 1$ targets simultaneously would also be interesting,
but we expect this to be substantially harder.
For example, in this setting our attacks against convex models lead to a problem with more unknowns than equations.
On the defenses side, we empirically showed that DP training with large values of $\epsilon$
can effectively mitigate our reconstruction attacks.
Our theoretical discussion, stemming from a new definition
of reconstruction robustness and a study of its connection to (R)DP,
shows this is a general phenomenon:
informed reconstruction attacks can be prevented
with large values of $\epsilon$
under some assumptions on the adversary. Validating such assumptions in particular applications would open the door to practical models which are accurate and resilient against reconstruction attacks.

\section*{Acknowledgment}
The authors would like to thank: Leonard Berrada, Adri{\`a} Gasc{\'o}n and Shakir Mohamed for feedback on an earlier version of this manuscript; Brendan McMahan for suggesting the idea that random initialization in SGD might make privacy attacks harder which inspired some of our experiments; and Olivia Wiles for discussions on how to improve reconstructor network training on CIFAR-10.
This work was done while G.C. was at the Alan Turing Institute.

\bibliographystyle{IEEEtran}
\bibliography{short_biblio}

\begin{thebibliography}{10}
\providecommand{\url}[1]{#1}
\csname url@samestyle\endcsname
\providecommand{\newblock}{\relax}
\providecommand{\bibinfo}[2]{#2}
\providecommand{\BIBentrySTDinterwordspacing}{\spaceskip=0pt\relax}
\providecommand{\BIBentryALTinterwordstretchfactor}{4}
\providecommand{\BIBentryALTinterwordspacing}{\spaceskip=\fontdimen2\font plus
\BIBentryALTinterwordstretchfactor\fontdimen3\font minus
  \fontdimen4\font\relax}
\providecommand{\BIBforeignlanguage}[2]{{%
\expandafter\ifx\csname l@#1\endcsname\relax
\typeout{** WARNING: IEEEtran.bst: No hyphenation pattern has been}%
\typeout{** loaded for the language `#1'. Using the pattern for}%
\typeout{** the default language instead.}%
\else
\language=\csname l@#1\endcsname
\fi
#2}}
\providecommand{\BIBdecl}{\relax}
\BIBdecl

\bibitem{DBLP:conf/iclr/ZhangBHRV17}
C.~Zhang, S.~Bengio, M.~Hardt, B.~Recht, and O.~Vinyals, ``Understanding deep
  learning requires rethinking generalization,'' in \emph{International
  Conference on Learning Representations (ICLR)}, 2017.

\bibitem{DBLP:conf/stoc/Feldman20}
V.~Feldman, ``Does learning require memorization? a short tale about a long
  tail,'' in \emph{{ACM} Symposium on Theory of Computing (STOC)}, 2020.

\bibitem{DBLP:conf/nips/FeldmanZ20}
V.~Feldman and C.~Zhang, ``What neural networks memorize and why: Discovering
  the long tail via influence estimation,'' in \emph{Conference on Neural
  Information Processing Systems (NeurIPS)}, 2020.

\bibitem{DBLP:conf/stoc/BrownBFST21}
G.~Brown, M.~Bun, V.~Feldman, A.~D. Smith, and K.~Talwar, ``When is
  memorization of irrelevant training data necessary for high-accuracy
  learning?'' in \emph{{ACM} Symposium on Theory of Computing (STOC)}, 2021.

\bibitem{DBLP:conf/sp/ShokriSSS17}
R.~Shokri, M.~Stronati, C.~Song, and V.~Shmatikov, ``Membership inference
  attacks against machine learning models,'' in \emph{{IEEE} Symposium on
  Security and Privacy (SP)}, 2017.

\bibitem{DBLP:conf/tcc/DworkMNS06}
C.~Dwork, F.~McSherry, K.~Nissim, and A.~D. Smith, ``Calibrating noise to
  sensitivity in private data analysis,'' in \emph{Theory of Cryptography
  Conference (TCC)}, 2006.

\bibitem{DBLP:conf/uss/Carlini0EKS19}
N.~Carlini, C.~Liu, {\'{U}}.~Erlingsson, J.~Kos, and D.~Song, ``The secret
  sharer: Evaluating and testing unintended memorization in neural networks,''
  in \emph{{USENIX} Security Symposium}, 2019.

\bibitem{DBLP:conf/uss/CarliniTWJHLRBS21}
N.~Carlini, F.~Tram{\`{e}}r, E.~Wallace, M.~Jagielski, A.~Herbert{-}Voss,
  K.~Lee, A.~Roberts, T.~B. Brown, D.~Song, {\'{U}}.~Erlingsson, A.~Oprea, and
  C.~Raffel, ``Extracting training data from large language models,'' in
  \emph{{USENIX} Security Symposium}, 2021.

\bibitem{DBLP:conf/aistats/McMahanMRHA17}
B.~McMahan, E.~Moore, D.~Ramage, S.~Hampson, and B.~{Ag{\"{u}}era y Arcas},
  ``Communication-efficient learning of deep networks from decentralized
  data,'' in \emph{International Conference on Artificial Intelligence and
  Statistics (AISTATS)}, 2017.

\bibitem{zhu2019deep}
L.~Zhu, Z.~Liu, and S.~Han, ``Deep leakage from gradients,'' in
  \emph{Conference on Neural Information Processing Systems (NeurIPS)}, 2019.

\bibitem{fredrikson2014privacy}
M.~Fredrikson, E.~Lantz, S.~Jha, S.~Lin, D.~Page, and T.~Ristenpart, ``Privacy
  in pharmacogenetics: An end-to-end case study of personalized warfarin
  dosing,'' in \emph{{USENIX} Security Symposium}, 2014.

\bibitem{ganju2018property}
K.~Ganju, Q.~Wang, W.~Yang, C.~A. Gunter, and N.~Borisov, ``Property inference
  attacks on fully connected neural networks using permutation invariant
  representations,'' in \emph{{ACM} Conference on Computer and Communications
  Security (CCS)}, 2018.

\bibitem{suri2021formalizing}
A.~Suri and D.~Evans, ``Formalizing and estimating distribution inference
  risks,'' \emph{arXiv:2109.06024}, 2021.

\bibitem{DBLP:conf/sp/NasrSTPC21}
M.~Nasr, S.~Song, A.~Thakurta, N.~Papernot, and N.~Carlini, ``Adversary
  instantiation: Lower bounds for differentially private machine learning,'' in
  \emph{{IEEE} Symposium on Security and Privacy (SP)}, 2021.

\bibitem{DBLP:conf/ccs/AbadiCGMMT016}
M.~Abadi, A.~Chu, I.~J. Goodfellow, H.~B. McMahan, I.~Mironov, K.~Talwar, and
  L.~Zhang, ``Deep learning with differential privacy,'' in \emph{{ACM}
  Conference on Computer and Communications Security (CCS)}, 2016.

\bibitem{fredrikson2015model}
M.~Fredrikson, S.~Jha, and T.~Ristenpart, ``Model inversion attacks that
  exploit confidence information and basic countermeasures,'' in \emph{{ACM}
  Conference on Computer and Communications Security (CCS)}, 2015.

\bibitem{DBLP:conf/csfw/YeomGFJ18}
S.~Yeom, I.~Giacomelli, M.~Fredrikson, and S.~Jha, ``Privacy risk in machine
  learning: Analyzing the connection to overfitting,'' in \emph{{IEEE} Computer
  Security Foundations Symposium (CSF)}, 2018.

\bibitem{DBLP:conf/cvpr/ZhangJP0LS20}
Y.~Zhang, R.~Jia, H.~Pei, W.~Wang, B.~Li, and D.~Song, ``The secret revealer:
  Generative model-inversion attacks against deep neural networks,'' in
  \emph{IEEE Conference on Computer Vision and Pattern Recognition (CVPR)},
  2020.

\bibitem{salem2018ml}
A.~Salem, Y.~Zhang, M.~Humbert, P.~Berrang, M.~Fritz, and M.~Backes,
  ``{ML-Leaks}: Model and data independent membership inference attacks and
  defenses on machine learning models,'' in \emph{Network and Distributed
  System Security Symposium (NDSS)}, 2019.

\bibitem{DBLP:conf/sp/NasrSH19}
M.~Nasr, R.~Shokri, and A.~Houmansadr, ``Comprehensive privacy analysis of deep
  learning: Passive and active white-box inference attacks against centralized
  and federated learning,'' in \emph{{IEEE} Symposium on Security and Privacy
  (SP)}, 2019.

\bibitem{radford2019language}
A.~Radford, J.~Wu, R.~Child, D.~Luan, D.~Amodei, and I.~Sutskever, ``Language
  models are unsupervised multitask learners,'' 2019.

\bibitem{wang2019beyond}
Z.~Wang, M.~Song, Z.~Zhang, Y.~Song, Q.~Wang, and H.~Qi, ``Beyond inferring
  class representatives: User-level privacy leakage from federated learning,''
  in \emph{{IEEE} Conference on Computer Communications (INFOCOM)}, 2019.

\bibitem{geiping2020inverting}
J.~Geiping, H.~Bauermeister, H.~Dr{\"{o}}ge, and M.~Moeller, ``Inverting
  gradients - how easy is it to break privacy in federated learning?'' in
  \emph{Conference on Neural Information Processing Systems (NeurIPS)}, 2020.

\bibitem{wainakh2021user}
A.~Wainakh, F.~Ventola, T.~M{\"u}{\ss}ig, J.~Keim, C.~G. Cordero, E.~Zimmer,
  T.~Grube, K.~Kersting, and M.~M{\"u}hlh{\"a}user, ``User label leakage from
  gradients in federated learning,'' \emph{arXiv:2105.09369}, 2021.

\bibitem{BeguelinWTRPOKB20}
S.~Z. B{\'{e}}guelin, L.~Wutschitz, S.~Tople, V.~R{\"{u}}hle, A.~Paverd,
  O.~Ohrimenko, B.~K{\"{o}}pf, and M.~Brockschmidt, ``Analyzing information
  leakage of updates to natural language models,'' in \emph{{ACM} Conference on
  Computer and Communications Security (CCS)}, 2020.

\bibitem{Salem20}
A.~Salem, A.~Bhattacharya, M.~Backes, M.~Fritz, and Y.~Zhang, ``Updates-leak:
  Data set inference and reconstruction attacks in online learning,'' in
  \emph{{USENIX} Security Symposium}, 2020.

\bibitem{mccullagh2019generalized}
P.~McCullagh and J.~A. Nelder, \emph{Generalized linear models}.\hskip 1em plus
  0.5em minus 0.4em\relax Routledge, 2019.

\bibitem{wedderburn1976existence}
R.~W. Wedderburn, ``On the existence and uniqueness of the maximum likelihood
  estimates for certain generalized linear models,'' \emph{Biometrika}, 1976.

\bibitem{zhang2018unreasonable}
R.~Zhang, P.~Isola, A.~A. Efros, E.~Shechtman, and O.~Wang, ``The unreasonable
  effectiveness of deep features as a perceptual metric,'' in \emph{IEEE
  Conference on Computer Vision and Pattern Recognition (CVPR)}, 2018.

\bibitem{isola2017image}
P.~Isola, J.-Y. Zhu, T.~Zhou, and A.~A. Efros, ``Image-to-image translation
  with conditional adversarial networks,'' in \emph{IEEE Conference on Computer
  Vision and Pattern Recognition (CVPR)}, 2017.

\bibitem{mao2017least}
X.~Mao, Q.~Li, H.~Xie, R.~Y. Lau, Z.~Wang, and S.~Paul~Smolley, ``Least squares
  generative adversarial networks,'' in \emph{{IEEE} International Conference
  on Computer Vision (ICCV)}, 2017.

\bibitem{726791}
Y.~Lecun, L.~Bottou, Y.~Bengio, and P.~Haffner, ``Gradient-based learning
  applied to document recognition,'' \emph{Proceedings of the IEEE}, 1998.

\bibitem{ZagoruykoK16}
S.~Zagoruyko and N.~Komodakis, ``Wide residual networks,'' in \emph{British
  Machine Vision Conference (BMVC)}, 2016.

\bibitem{DBLP:journals/corr/abs-2106-03004}
S.~Fort, J.~Ren, and B.~Lakshminarayanan, ``Exploring the limits of
  out-of-distribution detection,'' \emph{arXiv:2106.03004}, 2021.

\bibitem{DBLP:conf/nips/Li0TSG1}
H.~Li, Z.~Xu, G.~Taylor, C.~Studer, and T.~Goldstein, ``Visualizing the loss
  landscape of neural nets,'' in \emph{Conference on Neural Information
  Processing Systems (NeurIPS)}, 2018.

\bibitem{haiku2020github}
\BIBentryALTinterwordspacing
T.~Hennigan, T.~Cai, T.~Norman, and I.~Babuschkin, ``{H}aiku: {S}onnet for
  {JAX},'' 2020. [Online]. Available: \url{http://github.com/deepmind/dm-haiku}
\BIBentrySTDinterwordspacing

\bibitem{jagielski2020auditing}
M.~Jagielski, J.~Ullman, and A.~Oprea, ``Auditing differentially private
  machine learning: How private is private sgd?'' \emph{Advances in Neural
  Information Processing Systems}, 2020.

\bibitem{DBLP:conf/eurocrypt/DworkKMMN06}
C.~Dwork, K.~Kenthapadi, F.~McSherry, I.~Mironov, and M.~Naor, ``Our data,
  ourselves: Privacy via distributed noise generation,'' in \emph{International
  Conference on the Theory and Applications of Cryptographic Techniques
  (EUROCRYPT)}, 2006.

\bibitem{DBLP:conf/csfw/Mironov17}
I.~Mironov, ``R{\'{e}}nyi differential privacy,'' in \emph{{IEEE} Computer
  Security Foundations Symposium (CSF)}, 2017.

\bibitem{DBLP:conf/tcc/BunS16}
M.~Bun and T.~Steinke, ``Concentrated differential privacy: Simplifications,
  extensions, and lower bounds,'' in \emph{Theory of Cryptography Conference
  (TCC)}, 2016.

\bibitem{DBLP:journals/corr/abs-1812-00984}
A.~Bhowmick, J.~C. Duchi, J.~Freudiger, G.~Kapoor, and R.~Rogers, ``Protection
  against reconstruction and its applications in private federated learning,''
  \emph{arXiv:1812.00984}, 2018.

\bibitem{DBLP:conf/iclr/PapernotSMRTE18}
N.~Papernot, S.~Song, I.~Mironov, A.~Raghunathan, K.~Talwar, and
  {\'{U}}.~Erlingsson, ``Scalable private learning with {PATE},'' in
  \emph{International Conference on Learning Representations (ICLR)}, 2018.

\bibitem{DBLP:journals/corr/abs-1908-10530}
I.~Mironov, K.~Talwar, and L.~Zhang, ``R{\'{e}}nyi differential privacy of the
  sampled gaussian mechanism,'' \emph{arXiv:1908.10530}, 2019.

\bibitem{DBLP:conf/pods/BlumDMN05}
A.~Blum, C.~Dwork, F.~McSherry, and K.~Nissim, ``Practical privacy: the {SuLQ}
  framework,'' in \emph{{ACM} Symposium on Principles of Database Systems
  (PODS)}, 2005.

\bibitem{mcsherry_2021}
\BIBentryALTinterwordspacing
F.~McSherry, ``{I suspect the "Discovery" had a different feel for the various
  involved people. I personally spent a lot of time trying to remove explicit
  references to adversaries and assumptions about them.}'' Jan 2021. [Online].
  Available: \url{https://twitter.com/frankmcsherry/status/1354789417727234049}
\BIBentrySTDinterwordspacing

\bibitem{DBLP:journals/tods/KiferM14}
D.~Kifer and A.~Machanavajjhala, ``Pufferfish: {A} framework for mathematical
  privacy definitions,'' \emph{{ACM} Trans. Database Syst.}, 2014.

\bibitem{DBLP:conf/innovations/GhoshK17}
A.~Ghosh and R.~Kleinberg, ``Inferential privacy guarantees for differentially
  private mechanisms,'' in \emph{Innovations in Theoretical Computer Science
  Conference (ITCS)}, 2017.

\bibitem{DBLP:journals/jpc/Kasiviswanathan14}
S.~P. Kasiviswanathan and A.~D. Smith, ``On the 'semantics' of differential
  privacy: {A} bayesian formulation,'' \emph{J. Priv. Confidentiality}, 2014.

\bibitem{DBLP:conf/cikm/Duan09}
Y.~Duan, ``Privacy without noise,'' in \emph{{ACM} Conference on Information
  and Knowledge Management (CIKM)}, 2009.

\bibitem{DBLP:conf/asiacrypt/BhaskarBGLT11}
R.~Bhaskar, A.~Bhowmick, V.~Goyal, S.~Laxman, and A.~Thakurta, ``Noiseless
  database privacy,'' in \emph{International Conference on the Theory and
  Application of Cryptology and Information Security (ASIACRYPT)}, 2011.

\bibitem{DBLP:conf/focs/BassilyGKS13}
R.~Bassily, A.~Groce, J.~Katz, and A.~D. Smith, ``Coupled-worlds privacy:
  Exploiting adversarial uncertainty in statistical data privacy,'' in
  \emph{{IEEE} Symposium on Foundations of Computer Science (FOCS)}, 2013.

\bibitem{DBLP:journals/corr/abs-1905-00650}
D.~Desfontaines, E.~Mohammadi, E.~Krahmer, and D.~Basin, ``Differential privacy
  with partial knowledge,'' \emph{arXiv:1905.00650}, 2019.

\bibitem{DBLP:conf/wpes/ShokriFJH09}
R.~Shokri, J.~Freudiger, M.~Jadliwala, and J.~Hubaux, ``A distortion-based
  metric for location privacy,'' in \emph{{ACM} Workshop on Privacy in the
  Electronic Society (WPES)}, 2009.

\bibitem{DBLP:conf/sp/ShokriTBH11}
R.~Shokri, G.~Theodorakopoulos, J.~L. Boudec, and J.~Hubaux, ``Quantifying
  location privacy,'' in \emph{{IEEE} Symposium on Security and Privacy (SP)},
  2011.

\bibitem{wasserman2010statistical}
L.~Wasserman and S.~Zhou, ``A statistical framework for differential privacy,''
  \emph{Journal of the American Statistical Association}, 2010.

\bibitem{DBLP:conf/icml/KairouzOV15}
P.~Kairouz, S.~Oh, and P.~Viswanath, ``The composition theorem for differential
  privacy,'' in \emph{International Conference on Machine Learning (ICML)},
  2015.

\bibitem{DBLP:conf/aistats/BalleBGHS20}
B.~Balle, G.~Barthe, M.~Gaboardi, J.~Hsu, and T.~Sato, ``Hypothesis testing
  interpretations and renyi differential privacy,'' in \emph{International
  Conference on Artificial Intelligence and Statistics (AISTATS)}, 2020.

\bibitem{DBLP:conf/pods/DinurN03}
I.~Dinur and K.~Nissim, ``Revealing information while preserving privacy,'' in
  \emph{{ACM} Symposium on Principles of Database Systems (PODS)}, 2003.

\bibitem{DBLP:conf/stoc/DworkMT07}
C.~Dwork, F.~McSherry, and K.~Talwar, ``The price of privacy and the limits of
  {LP} decoding,'' in \emph{{ACM} Symposium on Theory of Computing (STOC)},
  2007.

\bibitem{dwork2017exposed}
C.~Dwork, A.~Smith, T.~Steinke, and J.~Ullman, ``Exposed! {A} survey of attacks
  on private data,'' \emph{Annual Review of Statistics and Its Application},
  2017.

\bibitem{DBLP:journals/jpc/CohenN20}
A.~Cohen and K.~Nissim, ``Linear program reconstruction in practice,'' \emph{J.
  Priv. Confidentiality}, 2020.

\bibitem{dpblogreconstruction}
\BIBentryALTinterwordspacing
A.~Cohen, S.~Nikolov, Z.~Schutzman, and J.~Ullman, ``Reconstruction attacks in
  practice,'' 2020. [Online]. Available:
  \url{https://differentialprivacy.org/diffix-attack/}
\BIBentrySTDinterwordspacing

\bibitem{DBLP:conf/fossacs/Smith09}
G.~Smith, ``On the foundations of quantitative information flow,'' in
  \emph{International Conference on Foundations of Software Science and
  Computational Structures (FOSSACS)}, 2009.

\bibitem{alvim2020science}
M.~S. Alvim, K.~Chatzikokolakis, A.~McIver, C.~Morgan, C.~Palamidessi, and
  G.~Smith, \emph{The Science of Quantitative Information Flow}.\hskip 1em plus
  0.5em minus 0.4em\relax Springer, 2020.

\bibitem{DBLP:conf/birthday/ElSalamounyCP14}
E.~ElSalamouny, K.~Chatzikokolakis, and C.~Palamidessi, ``Generalized
  differential privacy: Regions of priors that admit robust optimal
  mechanisms,'' in \emph{Horizons of the Mind. {A} Tribute to Prakash
  Panangaden - Essays Dedicated to Prakash Panangaden on the Occasion of His
  60th Birthday}, 2014.

\bibitem{shaked2007stochastic}
M.~Shaked and J.~G. Shanthikumar, \emph{Stochastic orders}.\hskip 1em plus
  0.5em minus 0.4em\relax Springer Science \& Business Media, 2007.

\bibitem{DBLP:journals/rsa/DasguptaG03}
S.~Dasgupta and A.~Gupta, ``An elementary proof of a theorem of johnson and
  lindenstrauss,'' \emph{Random Struct. Algorithms}, 2003.

\bibitem{johnson_2020}
\BIBentryALTinterwordspacing
M.~Johnson, ``{add gpu determinism note},'' Nov 2020. [Online]. Available:
  \url{https://github.com/google/jax/pull/4824}
\BIBentrySTDinterwordspacing

\bibitem{DBLP:journals/corr/abs-2008-11193}
V.~Feldman and T.~Zrnic, ``Individual privacy accounting via a renyi filter,''
  \emph{arXiv:2008.11193}, 2020.

\bibitem{jaxact}
``{JAX} activations,'' \url{https://jax.readthedocs.io/en/latest/jax.nn.html},
  accessed: 2022-03-25.

\bibitem{song_2020}
\BIBentryALTinterwordspacing
S.~Song and D.~Marn, ``{Introducing a New Privacy Testing Library in
  TensorFlow.}'' June 2020. [Online]. Available:
  \url{https://blog.tensorflow.org/2020/06/introducing-new-privacy-testing-library.html}
\BIBentrySTDinterwordspacing

\end{thebibliography}

\appendix[Proofs]

\main{
We provide proof sketches for the main theoretical results of the paper. Full proofs can be found on the arXiv version \cite{DBLP:journals/corr/abs-2201-04845}.
}

\main{
\begin{proof}[Proof sketch of Theorem~\ref{thm:convex-attack}]
For a GLM model $\model$ trained to convergence, \eqref{eq:convex-attack-strategy} takes the form
\begin{equation*}
x(g^{-1}(\ip{x}{\model})-y) = - \Xfixed^\top(g^{-1}(\Xfixed\model)-\Yfixed) - \lambda\model \enspace.
\end{equation*}
When the model contains an intercept parameter, this yields $d$ equations with $d$ unknowns ($x_2, \ldots, x_d, y$) because $x_1 = 1$.
From the equation corresponding to this coordinate we obtain $g^{-1}(\ip{x}{\theta})-y = \Xfixed_1^\top B + \lambda \model_1$, which we can plug in the rest of equations to obtain the desired expression for $x$.
Once we have $x$ we plug it back into the first equation to recover $y$.
\end{proof}
}

\arxiv{
\todo{Shorten for submission.}

\begin{proof}[Proof of Theorem~\ref{thm:convex-attack}]
A GLM model $\model$ is trained as the solution to:
\begin{align}\label{eqn:glm_optimality}
    -\sum_{i=1}^n x_i (g^{-1}(\ip{x_i}{\model}) - y_i) + \lambda \model = 0 \enspace.
\end{align}

We solve for target point $x$ to obtain a reconstruction attack:
\begin{equation}
\label{eq:glm-reconstruction}
x(g(x\theta)-y) = - \Xfixed^\top(g(\Xfixed\theta)-\Yfixed) - \lambda\theta \,.
\end{equation}
$(\Xfixed, \Yfixed) = \Dfixed$ are resp. the
fixed objects and labels, which are known to the
attacker.
We indicate with $X_j$ the $j$-th column of the data points.
This attack assumes the model is fitted with the intercept
coefficient, i.e. $X_1 = (1, ..., 1)$.
\Cref{eq:glm-reconstruction} defines a system of $\dimpoint$ equations;
consider the $j$-th of them:

$$g(x\model) = y - \frac{\Xfixed_j^\top(g(\Xfixed\model)-\Xfixed) + \lambda\model_j}{x_j}$$

Now observe that $g(x\model)$ can be
determined by the attacker
for the intercept column, $j = 1$; (note that also $x_1 = 1$).

By plugging this back into \Cref{eq:glm-reconstruction} we get the desired form.
The target label $y$ is reconstructed similarly,
by plugging the solution for $x$ into the same system.
\end{proof}
}

\arxiv{
\begin{theorem}[Alternative attack to linear regression without intercept]
\label{thm:lr-reconstruction-no-intercept}
Consider a linear regression model $\model$ trained on dataset
$(X, Y) = (\Xfixed \cup \{x\}, \Yfixed \cup \{y\})$;
that is, $\model$ is the unique minimizer of $\frac{1}{2}\| X\model - Y \|^2$.
(The model may or may not have been fitted with an intercept term.)
Then $x$ can be retrieved given $\model$, $\Xfixed$, $\Yfixed$, and $y$, as follows:
$$\hat{x} = \Xfixed^\top (\Xfixed \model - \Yfixed) \cdot \frac{y \pm \sqrt{y^2 - 4 (\Xfixed \model - \Yfixed)^\top \Xfixed \model}}{2 (\Xfixed \model - \Yfixed)^\top \Xfixed \model} \,.$$
Note that this expression returns two candidate guesses for $x$.
\end{theorem}

\begin{proof}
Write the objective as
    $$
    \frac{1}{2}\| X \model - Y \|^2
    =
    \frac{1}{2}\| \Xfixed \model - \Yfixed \|^2 + \frac{1}{2} (x^\top \model - y)^2 \,.
    $$

Since $\model$ is the unique minimizer we can take the gradient of the objective above and set it to zero to find the system of equations
    $$
    \Xfixed^\top (\Xfixed \model - \Yfixed) + x (x^\top \theta - y) = 0
    $$
Reorganizing the terms we have
$x (x^\top \model - y) = - \Xfixed^\top (\Xfixed \model - \Yfixed)$.
Since $x^\top \model - y$ is a scalar we see that for $x$ to satisfy this equation it must be a multiple of the vector $\Xfixed^\top (\Xfixed \model - \Yfixed)$;
(here we assume that $\Xfixed \model - \Yfixed \neq 0$). \gnote{Analyze the corner case} So we only need to consider solutions of the form $x = \alpha \Xfixed^\top (\Xfixed \model - \Yfixed)$ for some $\alpha \in \R$.

Plugging this expression for $x$ in the system we obtain an equation for $\alpha$:
    $$
    \Xfixed^\top (\Xfixed \model - \Yfixed) (\alpha^2 (\Xfixed \model - \Yfixed)^\top \Xfixed \model - \alpha y + 1) = 0
    $$
    so
    $$
    \alpha^2 (\Xfixed \model - \Yfixed)^\top \Xfixed \model - \alpha y + 1 = 0
    $$
Solving for $\alpha$ we get the desired expression.
\end{proof}
}

\main{
\begin{proof}[Proof sketch of Theorem~\ref{thm:dp2rero}]
Fix $\reconstruct : \Theta \to \Zset$ and $\Dfixed \in \Zset^{n-1}$,
and let $Z \sim \pi$, $D_Z = \Dfixed \cup \{Z\}$ and $\theta \sim M(D_Z)$.
We write $p_M(\theta|z) = \Pr[M(D_z) = \theta]$ for the output density of $M$ on input $D_z$.
First we take an arbitrary $z_0 \in \Zset$ and show the probability $\Pr[\lattackloss(Z, \reconstruct(\theta)) \leq \eta]$ equals
\begin{align*}
\int_{\Theta} \left(\int_{\Zset}  \onesv[\lattackloss(z, \reconstruct(\theta)) \leq \eta]  \frac{p_M(\theta | z)}{p_M(\theta | z_0)} \pi(dz) \right) p_M(d\theta | z_0)  \enspace.
\end{align*}
Next we take $\alpha' = \frac{\alpha}{\alpha - 1}$ and through a standard application of H{\"o}lder's inequality bound the inner integral above by:
\begin{align*}
\kappa^{1/\alpha'}
    \cdot
    \left(\int_{\Zset} \left(\frac{p_M(\theta | z)}{p_M(\theta | z_0)}\right)^\alpha \pi(dz) \right)^{1/\alpha} \enspace.
\end{align*}
Plugging this bound into the expression for $\Pr[\lattackloss(Z, \reconstruct(\theta)) \leq \eta]$ and re-arranging terms, we use Jensen's inequality and the RDP assumption on $M$ to obtain:
\begin{align*}
    &
    \left(\frac{\Pr[\lattackloss(Z, \reconstruct(\theta)) \leq \eta]}{\kappa_{\pi}(\eta)^{1/\alpha'}}\right)^\alpha \leq
    \\
&\leq
    \int_{\Zset} \left(\int_{\Theta}  \left(\frac{p_M(\theta | z)}{p_M(\theta | z_0)}\right)^\alpha p_M(d\theta | z_0) \right) \pi(dz)
    \\
    &\leq
    \sup_z \int_{\Theta} \left(\frac{p_M(\theta | z)}{p_M(\theta | z_0)}\right)^\alpha p_M(d\theta | z_0)
    \\
    &\leq
    e^{(\alpha - 1) \epsilon}
    \enspace.
    \qedhere
\end{align*}
\end{proof}
}
\arxiv{
\dprero*
\begin{proof}[Proof of Theorem~\ref{thm:dp2rero}]
Fix arbitrary $\reconstruct : \Theta \to \Zset$, $\Dfixed \in \Zset^{n-1}$ and $z_0 \in \Zset$.
Let $Z \sim \pi$, $D_Z = \Dfixed \cup \{Z\}$ and $\theta \sim M(D_Z)$.
We write $p_M(\theta|z) = \Pr[M(D_z) = \theta]$ to denote the output density of $M$ on input $D_z$.
Then the probability $\Pr[\lattackloss(Z, \reconstruct(\theta)) \leq \eta]$ equals
\begin{align*}
    &\int_{\Zset} \int_{\Theta} \onesv[\lattackloss(z, \reconstruct(\theta)) \leq \eta] p_M(d\theta | z) \pi(dz)
    \\
    &=
    \int_{\Zset} \int_{\Theta} \onesv[\lattackloss(z, \reconstruct(\theta)) \leq \eta] p_M(d\theta | z_0) \frac{p_M(\theta | z)}{p_M(\theta | z_0)} \pi(dz)
    \\
    &=
    \int_{\Theta} \left(\int_{\Zset}  \onesv[\lattackloss(z, \reconstruct(\theta)) \leq \eta]  \frac{p_M(\theta | z)}{p_M(\theta | z_0)} \pi(dz) \right) p_M(d\theta | z_0)  \enspace.
\end{align*}
Taking $\alpha' = \frac{\alpha}{\alpha - 1}$ and applying H{\"o}lder's inequality to the inner integral we get:
\begin{align*}
    &\int_{\Zset} \onesv[\lattackloss(z, \reconstruct(\theta)) \leq \eta]  \frac{p_M(\theta | z)}{p_M(\theta | z_0)} \pi(dz)
    \\
    &\leq
    \left( \int_{\Zset} \onesv[\lattackloss(z, \reconstruct(\theta)) \leq \eta] \pi(dz) \right)^{1/\alpha'}
    \\
    &\;\; \times
    \left(\int_{\Zset} \left(\frac{p_M(\theta | z)}{p_M(\theta | z_0)}\right)^\alpha \pi(dz) \right)^{1/\alpha}
    \\
    &\leq
    \kappa^{1/\alpha'}
    \cdot
    \left(\int_{\Zset} \left(\frac{p_M(\theta | z)}{p_M(\theta | z_0)}\right)^\alpha \pi(dz) \right)^{1/\alpha}
    \enspace.
\end{align*}
After plugging the bound above into the expression for $\Pr[\lattackloss(Z, \reconstruct(\theta)) \leq \eta]$ and re-arranging terms, we use Jensen's inequality and the RDP assumption on $M$ to obtain:
\begin{align*}
    &
    \left(\frac{\Pr[\lattackloss(Z, \reconstruct(\theta)) \leq \eta]}{\kappa_{\pi}(\eta)^{1/\alpha'}}\right)^\alpha \leq
    \\
    &\leq
    \left(\int_{\Theta}  \left(\int_{\Zset} \left(\frac{p_M(\theta | z)}{p_M(\theta | z_0)}\right)^\alpha \pi(dz) \right)^{1/\alpha} p_M(d\theta | z_0) \right)^{\alpha}
    \\
    &\leq
    \int_{\Zset} \left(\int_{\Theta}  \left(\frac{p_M(\theta | z)}{p_M(\theta | z_0)}\right)^\alpha p_M(d\theta | z_0) \right) \pi(dz)
    \\
    &\leq
    \sup_z \int_{\Theta} \left(\frac{p_M(\theta | z)}{p_M(\theta | z_0)}\right)^\alpha p_M(d\theta | z_0)
    \\
    &\leq
    e^{(\alpha - 1) \epsilon}
    \enspace.
\end{align*}
The result follows from re-arranging this inequality.
\end{proof}
}

\arxiv{
\rerozcdp*

\begin{proof}[Proof of Corollary~\ref{cor:rero-zcdp}]
Theorem~\ref{thm:dp2rero} yields the bound $\gamma = \left(\kappa \cdot e^{\alpha \rho} \right)^{\frac{\alpha - 1}{\alpha}}$ for any $\alpha > 1$. This is minimized by taking $\alpha = \sqrt{\frac{\log(1/\kappa)}{\rho}}$, which is greater than $1$ by assumption.
Plugging this value of $\alpha$ into $\gamma$ and re-organizing the terms completes the proof.
\end{proof}
}

\main{
\begin{proof}[Proof sketch of Theorem~\ref{thm:rero2dp}]
Fix arbitrary $\Dfixed \in \Zset^{n-1}$, $z, z' \in \Zset$, $z \neq z'$, and $E \subseteq \Theta$, and let $\pi = \pi_{p, z, z'}$.
Define the reconstruction mapping $\reconstruct_E$ mapping $\theta$ to $z$ is $\theta \in E$ and to $z'$ otherwise.
By the ReRo assumptions on $M$ we have $\Pr_{Z \sim \pi, \theta \sim M(D_Z)}[\reconstruct_E(\theta) = Z] \leq \gamma$.
On the other hand, by definition of $\pi$ and $\reconstruct_E$, $\Pr_{Z \sim \pi, \theta \sim M(D_Z)}[\reconstruct_E(\theta) = Z]$ equals
\begin{align*}
\frac{\Pr[M(D_z) \in E] - e^{\epsilon} \Pr[M(D_{z'}) \in E] + e^{\epsilon}}{e^{\epsilon} + 1}
    \enspace.
\end{align*}
Upper bounding by $\gamma$ and re-arranging completes the proof.
\end{proof}
}

\arxiv{
\rerotodp*
\begin{proof}[Proof of Theorem~\ref{thm:rero2dp}]
Fix arbitrary $\Dfixed \in \Zset^{n-1}$, $z, z' \in \Zset$, $z \neq z'$, and $E \subseteq \Theta$.
Define the reconstruction mapping $\reconstruct_E$ given by
\begin{align*}
    \reconstruct_E(\theta)
    =
    \begin{cases}
    z & \text{if $\theta \in E$} \enspace,\\
    z' & \text{if $\theta \notin E$} \enspace.
    \end{cases}
\end{align*}
By the ReRo assumptions on $M$ we have
\begin{align*}
    \Pr_{Z \sim \pi_{p, z, z'}, \theta \sim M(D_Z)}[\reconstruct_E(\theta) = Z] \leq \gamma \enspace.
\end{align*}
On the other hand, by definition of $\pi_{p,z,z'}$ and $\reconstruct_E$ we have
\begin{align*}
    &\Pr_{Z \sim \pi_{p, z, z'}, \theta \sim M(D_Z)}[\reconstruct_E(\theta) = Z] =
    \\
    &=
    \frac{1}{e^{\epsilon} + 1} \Pr[M(D_z) \in E] + \frac{e^{\epsilon}}{e^{\epsilon} + 1} \Pr[M(D_{z'}) \notin E]
    \\
    &=
    \frac{\Pr[M(D_z) \in E] - e^{\epsilon} \Pr[M(D_{z'}) \in E] + e^{\epsilon}}{e^{\epsilon} + 1}
    \enspace.
\end{align*}
Thus, we get
\begin{align*}
    &\Pr[M(D_z) \in E] - e^{\epsilon} \Pr[M(D_{z'}) \in E]
    \\
    &\leq (e^{\epsilon} + 1) \gamma - e^{\epsilon}
    \leq
    \max\{0, (e^{\epsilon} + 1) \gamma - e^{\epsilon}\}
    \enspace.
    \qedhere
\end{align*}
\end{proof}
}

\begin{proof}[Proof of Proposition~\ref{prop:uniprior}]
Let $\pi = \cU(B_1^d(0))$ and write $\Vol(A)$ to denote the Eucliean volume of a set $A \subset \R^d$. By definition of the baseline error, for $\eta \in (0,1)$ we have
\begin{align*}
    \kappa_{\pi,\ell_2}(\eta)
    &=
    \sup_{z_0} \frac{\Vol(B_1^d(0) \cap B_\eta^d(z_0))}{\Vol(B_1^d(0))} = \eta^d = e^{-\Omega(d)} \enspace,
\end{align*}
where the calculation follows by the standard volume formula for $d$-dimensional Euclidean balls.
Plugging this expression in Corollary~\ref{cor:rero-pure-dp} shows that any $\epsilon$-DP mechanism with $\epsilon = o(d)$ provides $(\eta,\gamma)$-ReRo with respect to $\pi$ and $\ell$ with $\gamma = e^{-\Omega(d)}$.
A similar claim follows from Corollary~\ref{cor:rero-zcdp} applied to $\rho$-zCDP mechanisms with $\rho = o(d)$.
\end{proof}

\main{
\begin{proof}[Proof sketch of Proposition~\ref{prop:normalprior}]
Let $Z \sim \cN(0, I)$ and $F_{\eta}(z_0) = \Pr[\norm{Z + z_0}^2 \leq \eta^2]$.
First we show that $\argmax_{z_0} F_{\eta}(z_0) = 0$.
The proof of this intuitive fact relies on extending a $1$-dimensional stochastic domination property of Gaussian random variables \cite[Example 1.A.27]{shaked2007stochastic} to $d$ dimensions using an orthogonal decomposition of $Z$ along the space spanned by $z_0$ and its orthogonal complement.
Then we show that for $\nu = \cN(w,\sigma^2 I)$ this claim implies $\kappa_{\nu,\ell_2}(\eta) = F_{\eta/\sigma}(0)$.
Next we use a tail lower bound for chi-squared random variables \cite[Lemma 2.2]{DBLP:journals/rsa/DasguptaG03} to get
\begin{align*}
    \kappa_{\nu,\ell_2}(\eta)
    &\leq
    e^{\frac{d}{2}\left(1 - \frac{\eta^2}{\sigma^2 d} + \log \frac{\eta^2}{\sigma^2 d}\right)} \enspace.
\end{align*}
In particular, for $\sigma \geq \frac{2 \eta}{\sqrt{d}}$ we get $\kappa_{\nu,\ell_2}(\eta) \leq e^{-\Omega(d)}$.
The remaining of the proof follows the same argument as in Proposition~\ref{prop:uniprior}.
\end{proof}
}

\arxiv{
\normalprior*
\begin{proof}[Proof of Proposition~\ref{prop:normalprior}]
Let $Z \sim \cN(0, I)$ and $F_{\eta}(z_0) = \Pr[\norm{Z + z_0}^2 \leq \eta^2]$. First we claim that $\argmax_{z_0} F_{\eta}(z_0) = 0$. To see this, fix $z_0$ and let $\bar{z}_0 = z_0 / \norm{z_0}$. We can then write the orthogonal decomposition $Z = Z_{\parallel} + Z_{\bot}$ with $Z_{\parallel} = \ip{Z}{\bar{z}_0} \bar{z}_0$, where $Z_{\parallel}$ and $Z_{\bot}$ are independent multivariate Gaussians. By orthogonality we have $\norm{Z + z_0}^2 = \norm{Z_\bot}^2 + (\ip{Z}{\bar{z}_0} + \norm{z_0})^2$, which is a sum of independent chi-squared random variables. Note $\ip{Z}{\bar{z}_0} \sim \cN(0, 1)$, and let $f$ be the density function of $\norm{Z_\bot}^2$. Using that for a standard normal random variable $W \sim \cN(0, 1)$ we have $ \Pr[(W+a)^2 \leq t] \leq \Pr[W^2 \leq t]$ for all $a, t \in \R$ (see, e.g., \cite[Example 1.A.27]{shaked2007stochastic}), we get
\begin{align*}
    F_{\eta}(z_0)
    &=
    \int \Pr[(W + \norm{z_0})^2 \leq \eta^2 - t] f(t) dt
    \\
    &\leq
    \int \Pr[W^2 \leq \eta^2 - t] f(t) dt
    =
    F_{\eta}(0) \enspace.
\end{align*}
This proves the claim. For $\nu = \cN(w,\sigma^2 I)$ this implies
\begin{align*}
    \kappa_{\nu,\ell_2}(\eta)
    &=
    \sup_{z_0} \Pr_{Z \sim \nu}[\norm{Z - z_0} \leq \eta]
    \\
    &=
    \sup_{z_0} \Pr_{Z \sim \cN(0,I)}[\norm{Z + z_0}^2 \leq \eta^2/\sigma^2]
    =
    F_{\eta/\sigma}(0)
    \enspace.
\end{align*}
Therefore, using a tail lower bound for chi-squared random variables \cite[Lemma 2.2]{DBLP:journals/rsa/DasguptaG03} we get
\begin{align*}
    \kappa_{\nu,\ell_2}(\eta)
    &\leq
    e^{\frac{d}{2}\left(1 - \frac{\eta^2}{\sigma^2 d} + \log \frac{\eta^2}{\sigma^2 d}\right)} \enspace.
\end{align*}
In particular, for $\sigma \geq \frac{2 \eta}{\sqrt{d}}$ we get $\kappa_{\nu,\ell_2}(\eta) \leq e^{-\Omega(d)}$ -- this follows from a simple calculation using that $5/8 < \log 2$.
The remaining of the proof follows the same pattern as in Proposition~\ref{prop:uniprior}.
\end{proof}
}

\appendix[Additional Experimental Results]

\main{
We provide additional experiment results here. The interested reader can find a more expansive set of findings in \cite{DBLP:journals/corr/abs-2201-04845}, where we discuss in more detail the role of the released model hyperparameters, size, initialization, and gradient norm of the target on the reconstruction MSE.
}

\subsection{Randomness from Released Model Initialization}
\label{app: init_randomness}

\arxiv{
\begin{table}[H]
\centering
\caption{Reconstruction metrics with and without random released model initialization for different released model learning rate and momentum hyperparameters on MNIST. See \cref{ssec: metrics} for a description of each metric.}
\label{tab: mnist_random_init}
\resizebox{0.5\textwidth}{!}{
\begin{tabular}{cCCCCC}
\toprule
     
Randomly initialize     & \text{Released model}   & \text{Released model}  & \multirow{2}{*}{\text{LPIPS}}            & \multirow{2}{*}{\text{MSE}}                      & \multirow{2}{*}{\text{KL}}                \\
released models &  \text{learning rate}  &  \text{momentum} &           &                      &                   \\
\cmidrule{2-6}
\multirow{4}{*}{\cmark}  & 0.01 & 0.0        & 0.3342   & 0.0693     & 5.2477                              \\
                       & 0.01 & 0.9      & 0.3326    & 0.0691 & 5.2218                                                      \\
                       & 0.20  & 0.0       & 0.3326 & 0.0692    & 5.3599                                                   \\
                       & 0.20  & 0.9      & 0.3226  & 0.0695    & 5.5261                                                 \\
\cmidrule{2-6}
\multirow{4}{*}{\xmark} & 0.01 & 0.0       & 0.0197  & 0.0041  & 0.0140                                     \\
                       & 0.01 & 0.9      & 0.0225 & 0.0049 & 0.0179                                                    \\
                       & 0.20 & 0.0        & 0.0286  & 0.0063 & 0.0357                                                \\
                       & 0.20  & 0.9      & 0.0382  & 0.0089   & 0.0414                                          \\
\bottomrule
\end{tabular}
}
\end{table}
}

\begin{figure}[H]
\captionsetup{width=1.\linewidth}
  \centering
    \includegraphics[width=0.8\linewidth]{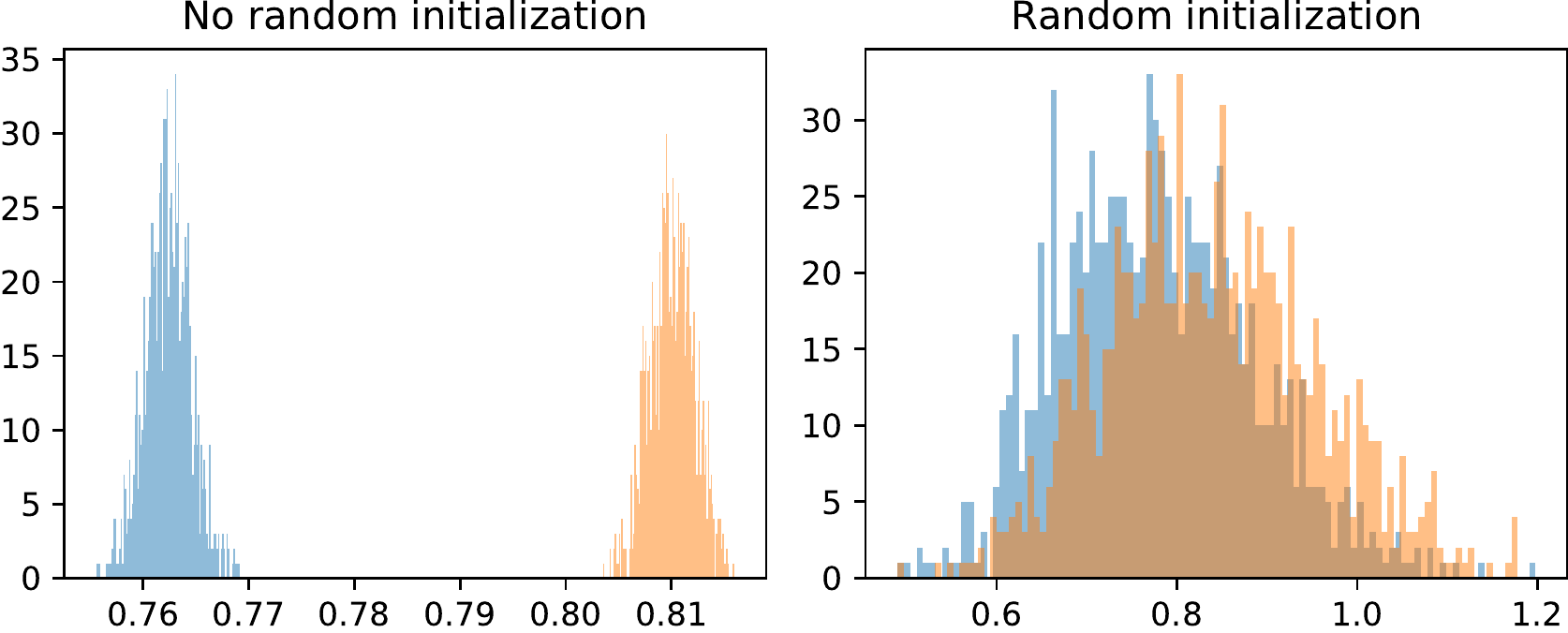}
\caption{On the MNIST dataset, given a target point $z$ we train $1K$ released models with (blue) and without (orange) this point included in two settings: when each model is initialized with a new random seed, and when each model has the same initialization. We plot the distribution of losses on this target point in these two settings. Clearly, when there is no model randomization the distributions are perfectly separable and so membership is easy to infer, while in the random setting, the distributions nearly perfectly overlap implying membership may be more difficult.
}
\label{fig:membership_diff}
\end{figure}

\arxiv{
In \Cref{tab: mnist_random_init}, we show results both with and without randomization from initialization for different released model learning rates on the MNIST dataset. 
The average MSE with and without random initialization was 0.0089 and 0.0695, respectively.
The choice of learning rate negatively impacts the attack in settings where initial parameters are known, but not when initial parameters are unknown, since the attack fails for any choice.
}

\main{
As observed in \Cref{ssec:reconctruction-factors}, the attack will fail when the adversary does not have knowledge of the initial parameters of the released model and so must instantiate each \attackin{} used to train the attack with a new seed that controls the selection of initial parameters.
}
\arxiv{
The attack fails when the adversary does not have knowledge of the initial parameters of the released model and so must instantiate each \attackin{} used to train the attack with a new seed that controls the selection of initial parameters.
}
We provide evidence that it may not be possible to perform a reconstruction attack in this setting by appealing to a simpler task of inferring membership, and demonstrating this problem is also difficult without knowledge of the initial parameters.
We instantiate an informed MIA as described in \Cref{sec:threat-model} on the MNIST dataset.
Specifically, given a target point $z$ we train $1K$ released models with and without this point included (but with the same fixed set) in two settings: when each model is initialized with a new seed (differing initial parameters), and when each model is initialized with the same seed (identical initial parameters).
In \Cref{fig:membership_diff}, we plot the distribution of losses on this target point in these two settings. 
Clearly, when there is no initial parameter randomization the distributions are perfectly separable and so membership is easy to infer, while in the random setting, the distributions nearly perfectly overlap implying membership may be more difficult, if not impossible.
Note that if released model training was fully deterministic, the distribution of losses on the target point in the setting with no random initialization would collapse to a point distribution.
However, all our models are trained with JAX on GPUs that compile with non-deterministic reductions, introducing a small source of randomness~\cite{johnson_2020}.

\subsection{Transfer Learning from a Reconstructor Network Trained on a Different Fixed Set}
\label{app: transfer_learning}

\begin{figure}[H]
  \centering
\begin{subfigure}[t]{.24\textwidth}
\captionsetup{width=0.75\textwidth}
\centering
    \includegraphics[width=0.99\linewidth]{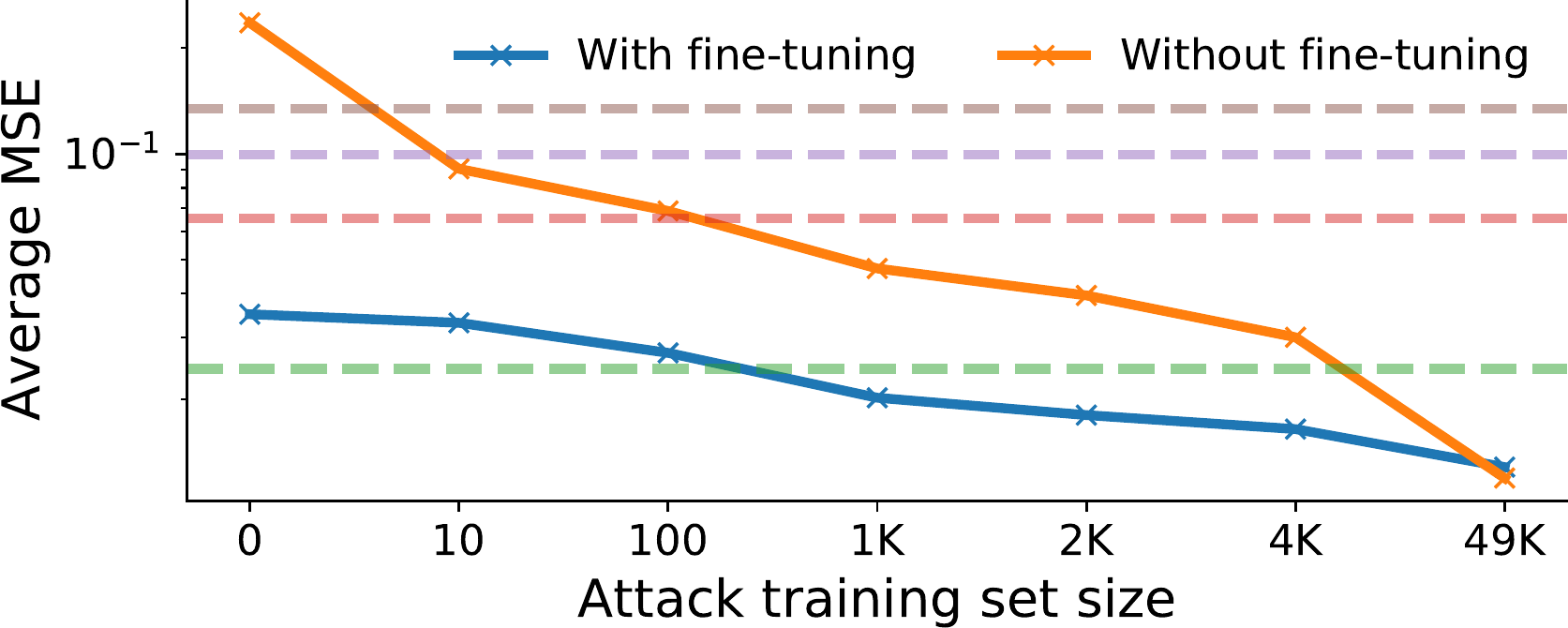}
\caption{MNIST, $|\Dfixed| = |\Dfixed'| = 10K$, leaving a maximum $49K$ shadow models.}
\label{fig:mnist_finetune_attack}
\end{subfigure}\begin{subfigure}[t]{.24\textwidth}
\captionsetup{width=0.75\textwidth}
\centering
    \includegraphics[width=0.99\linewidth]{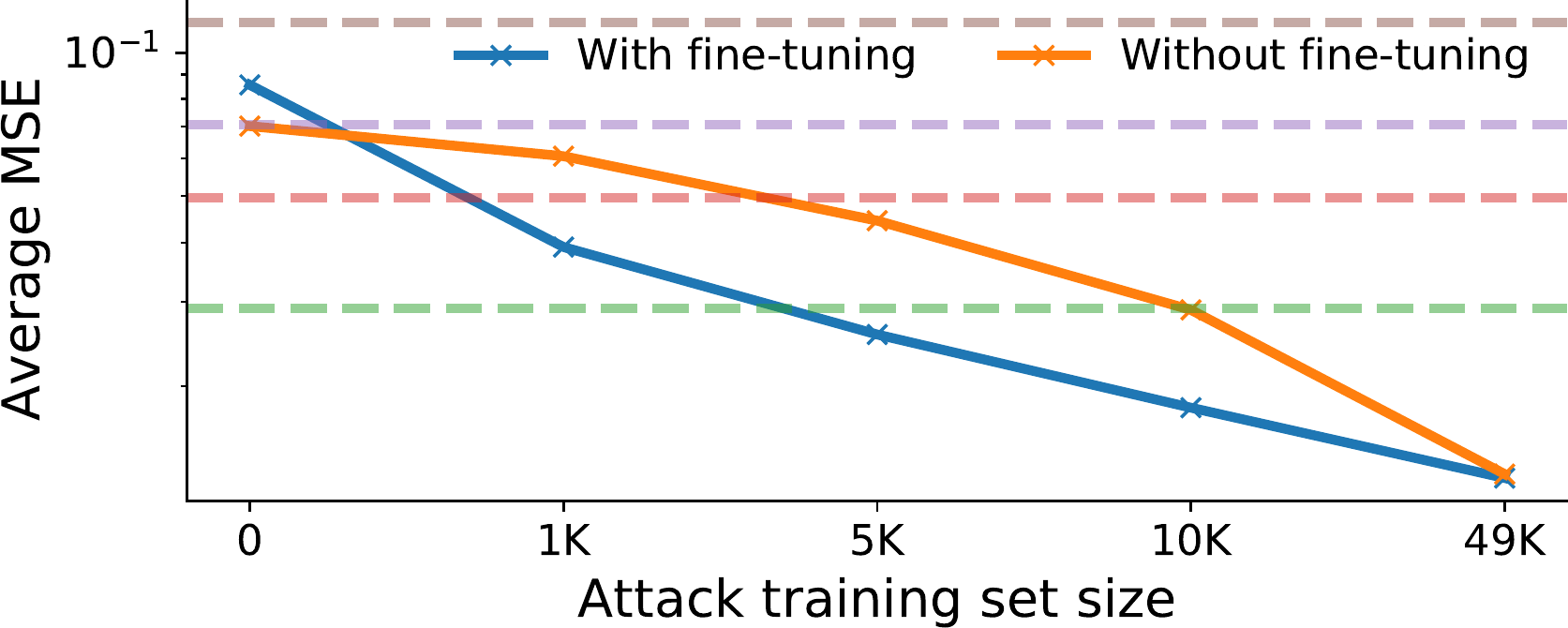}
    \caption{CIFAR-10, $|\Dfixed| = |\Dfixed'| = 5K$, leaving a maximum $49K$ shadow models.}
    \label{fig:cifar10_finetune_attack}
\end{subfigure}\caption{Fine-tuning the reconstruction network for a new target. The reconstruction network is initially trained to attack a released model trained with fixed dataset $\Dfixed$, and then fine-tuned for a new released model trained with fixed dataset $\Dfixed'$. Interestingly, the reconstructor network can do zero-shot learning on MNIST images, despite being trained on entirely separate data (i.e. $\Dfixed' \cap \Dfixed = \emptyset$).}
\label{fig:finetune_attack}
\end{figure}

Given a reconstructor network, $\attackmodel$, trained to attack released models of the form $\theta = \train_{\Dfixed}(z)$, can the adversary amortize the cost training a new $\attackmodel'$ that aims to attack a released model $\theta' = \train_{\Dfixed'}(z)$, where $\Dfixed' \cap \Dfixed = \emptyset$?
On both MNIST and CIFAR-10, in \Cref{fig:finetune_attack} we show that fine-tuning the reconstructor $\phi$ on only a small number of shadow models can reach comparative performance to a reconstructor trained from scratch on substantially more data.

\subsection{Adversary Knowledge of Starting Point: Initialization vs Near Convergence}
\label{app: pretrain_vs_init}

\begin{figure}[H]
  \centering
    \includegraphics[width=0.7\linewidth]{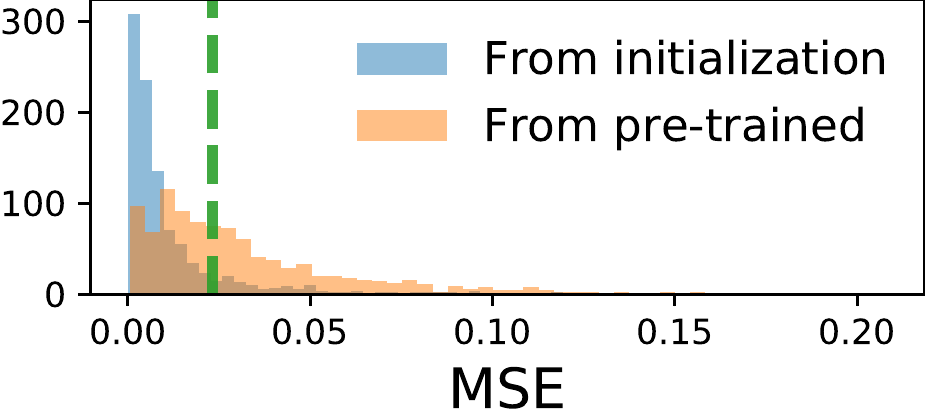}
\caption{A histogram of MSE for $1K$ released model targets for an adversary that observes the initial parameters compared to first observing a pre-trained released model near convergence. We also give the NN oracle for reference.} 
\label{fig:pretrain_vs_init}
\end{figure}

By default we assume the adversary knows the initial released model parameters, motivated by scenarios where the random seed used to generate initial parameters is made public or is leaked.
Another motivating example is that of federated learning, where an adversary participates in the learning protocol. 
However, in such a setting, it is not guaranteed the adversary will observe a model at it's initial state. 
If the adversary is only included in the protocol after a sufficient number of time steps, the state at which they first observe released model parameters may be close to convergence.
Here, we measure how reconstructions are affected by this subtle assumption.
We pre-train a released model on $10K$ MNIST images (this model already achieves $>92\%$ MNIST test set accuracy), and then following the experimental set-up reported in \Cref{ssec: exp_setup} on the remaining MNIST data, and compare to a released model in the standard setting where no pre-training occurs.
\Cref{fig:pretrain_vs_init} shows the MSE for each 1K released model target in both settings.
Clearly there is a difference in reconstruction fidelity that depends on the step at which the adversary first observes the released model parameters.
A model that has nearly converged may be less dependent / not memorize it's newly seen training data, making reconstructions more challenging.

\subsection{Visualization of Easy and Hard CIFAR-10 Reconstructions}
\label{app: viz}

\begin{figure}[H]
\captionsetup{width=0.5\textwidth}
  \centering
\begin{subfigure}[t]{.24\textwidth}
\centering
    \includegraphics[width=0.99\linewidth]{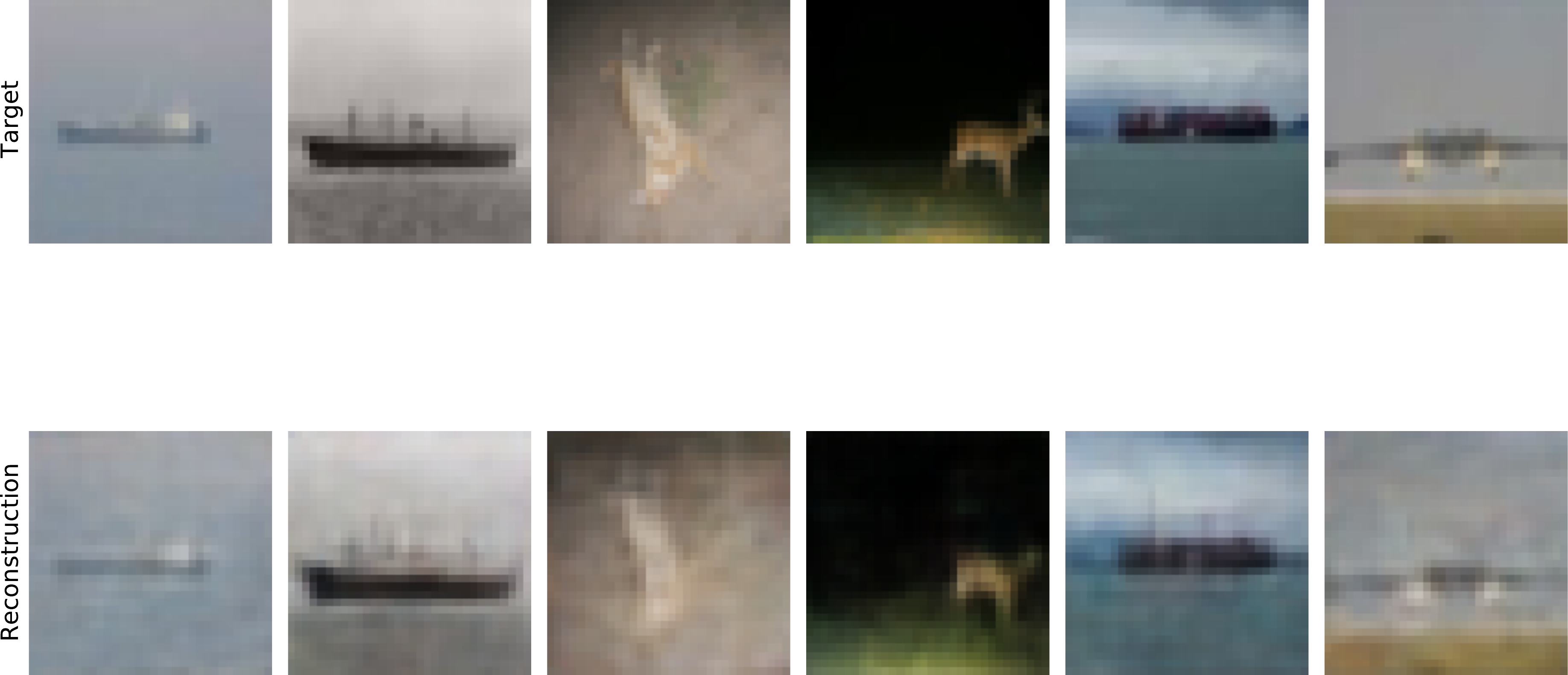}
    \caption{Smallest MSE}
    \label{fig:easy_viz}
\end{subfigure}\begin{subfigure}[t]{.24\textwidth}
\centering
    \includegraphics[width=0.99\linewidth]{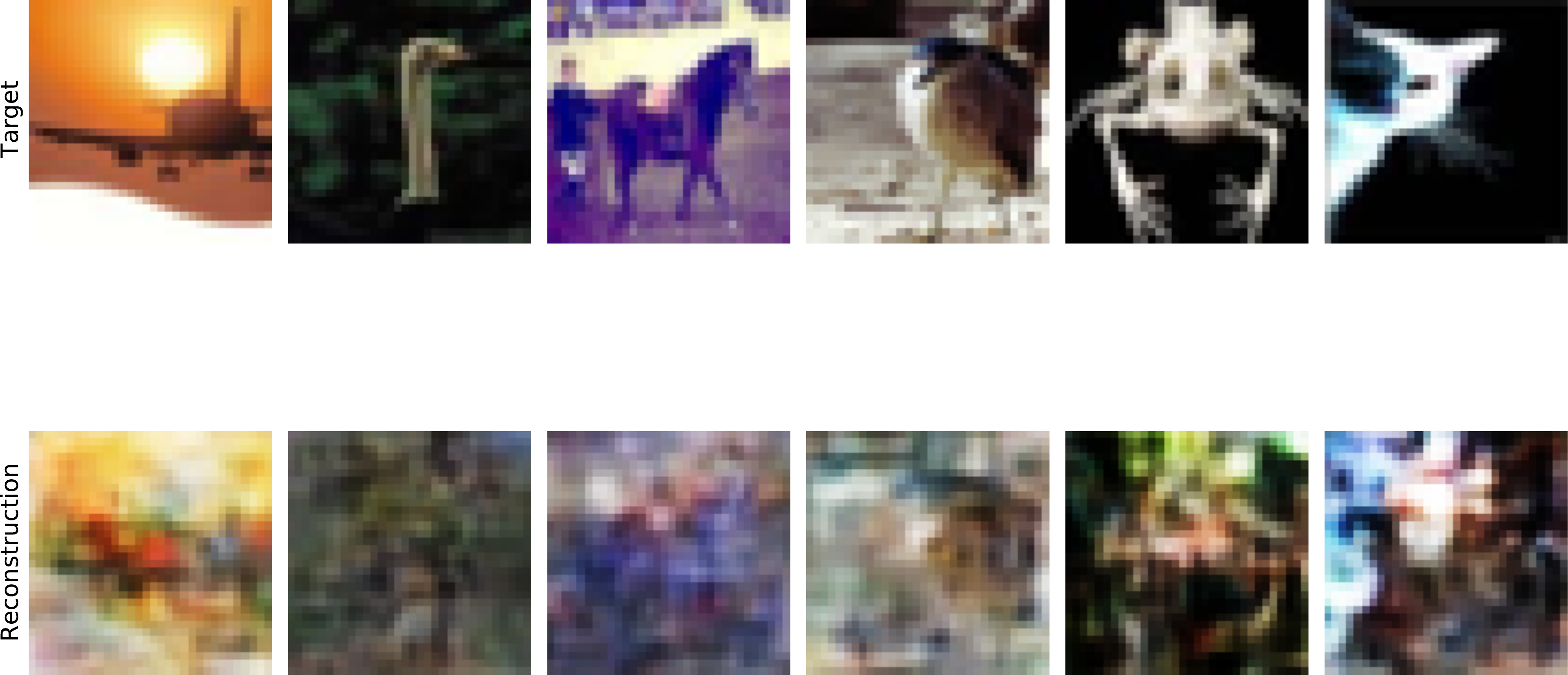}
    \caption{Largest MSE}
    \label{fig:hard_viz}
\end{subfigure} \caption{Example of the six smallest and size largest MSE reconstructions for CIFAR-10.}
\label{fig:easy_hard_viz}
\end{figure}

In \Cref{fig:easy_hard_viz} we show the six reconstructed CIFAR-10 examples with smallest MSE and six examples with largest MSE out the $1K$ targets used for evaluation. 
The easiest targets to reconstruct correspond to structurally simple images with a constant background, while the most difficult often have complex background and color schemes.

\subsection{Reconstructing Against a Released Model Trained with DP on CIFAR-10}
\label{app: cifar10_dp}

\begin{figure}[H]
\captionsetup{width=0.5\textwidth}
\centering
  \centering
  \includegraphics[width=0.76\linewidth]{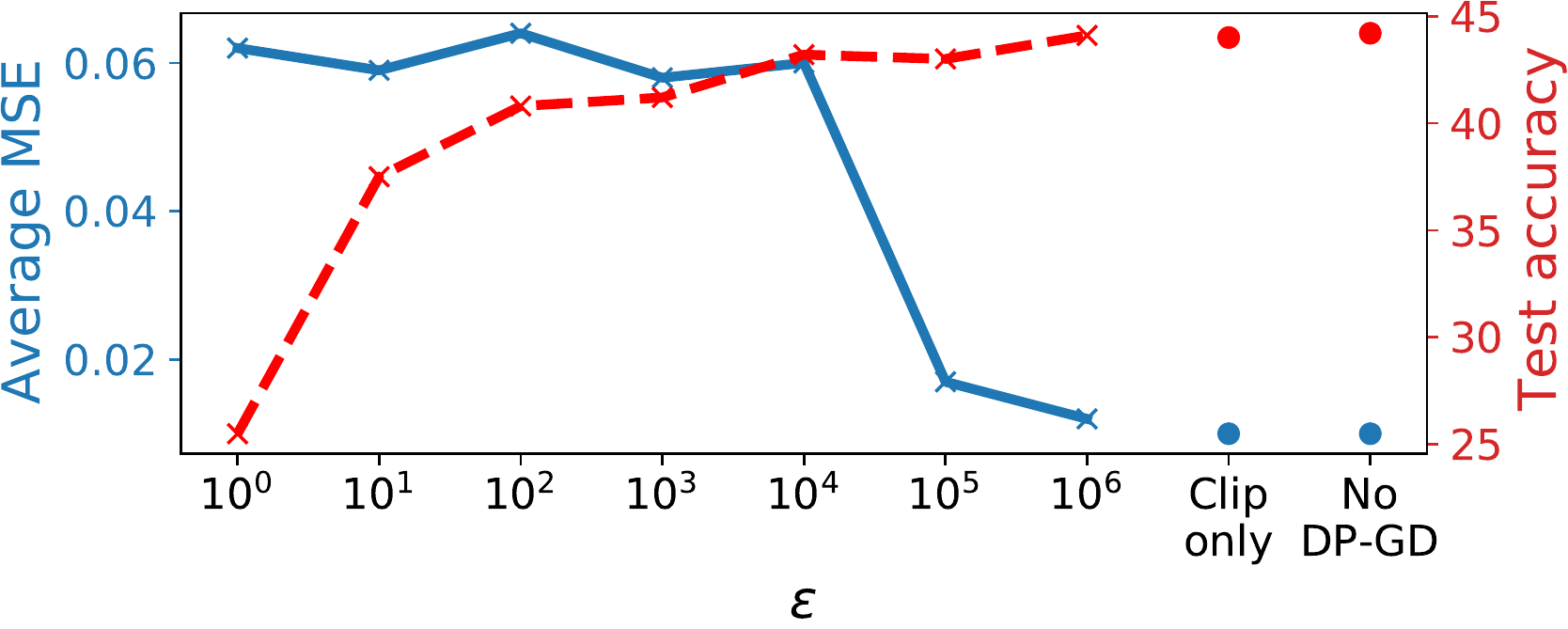}
  \caption{Average MSE of reconstructions and test accuracy of released model using $(\epsilon,\delta)$-DP on the CIFAR-10 dataset.}
  \label{fig:dpsgd_cifar10}
\end{figure}

We perform analogous DP experiments as in \Cref{ssec: mnist_dp_exp} for CIFAR-10.
Gradients are clipped to have a maximum $\ell_2$ norm of 10, and Gaussian noise is added such that the model is $(\epsilon, \delta=10^{-5})$-DP.
In \Cref{fig:dpsgd_cifar10} we see that again, a large $\epsilon$ in $(\epsilon, \delta)$-DP successfully mitigates against reconstruction attacks while preserving test accuracy in comparison to non-DP training.

\arxiv{
\subsection{Size of Released Model}
\label{app: released_model_size}

\begin{figure}[H]
\captionsetup{width=0.5\textwidth}
  \centering
    \includegraphics[width=0.75\linewidth]{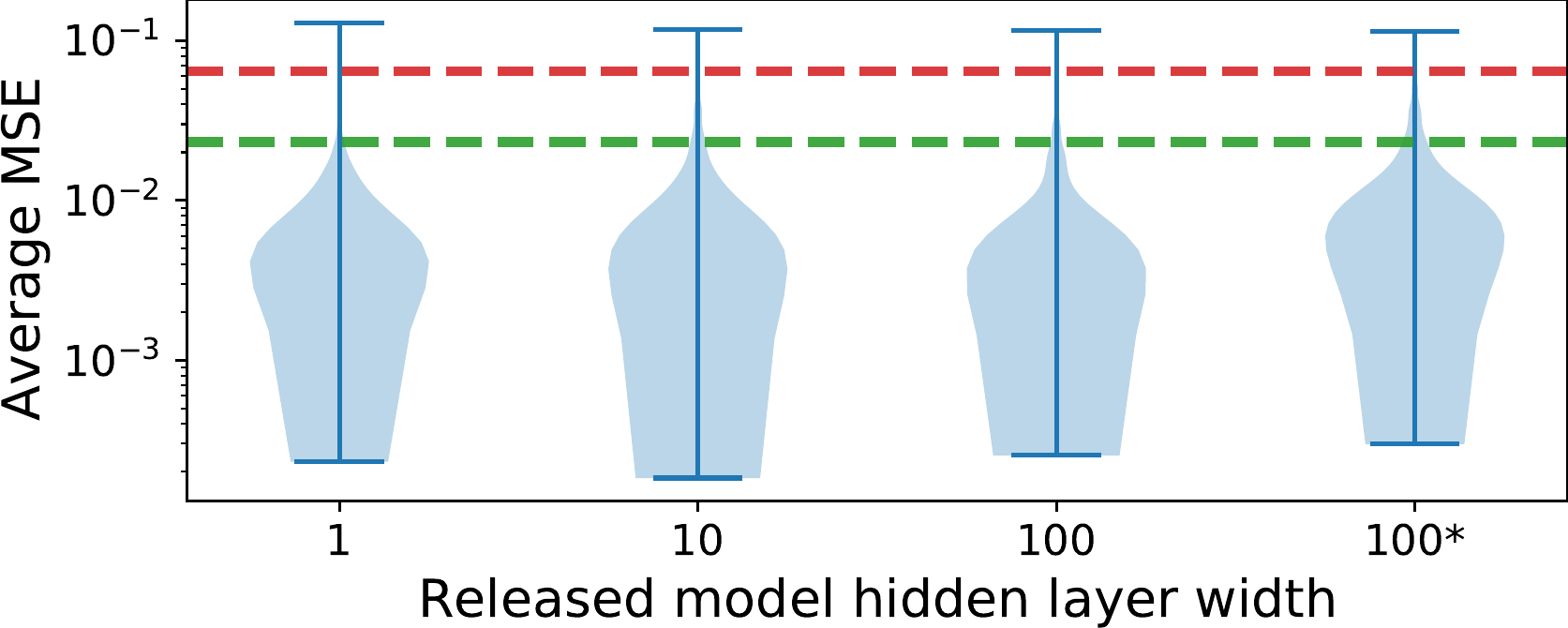}
\caption{Average MSE for different widths of a 1-hidden layer released model on MNIST. We denote the attack that only uses the second layer of the released model by *.}
\label{fig:mnist_mlp_arch}
\end{figure}

\begin{table}[H]
\centering
\caption{How the size of the released model impacts the reconstruction attack on MNIST. We denote the case where the attack only uses the second layer of the released model by *, resulting in a significant decrease in attack input dimensionality.}
\label{tab: mnist_mlp_arch_tab}
\resizebox{0.5\textwidth}{!}{
\begin{tabular}{lCCCCC}
\toprule
\multirow{2}{*}{\text{Hidden layer width}}                                                          & \text{Number of}  & \text{Dimensionality of input} & \multirow{2}{*}{\text{LPIPS}}          & \multirow{2}{*}{\text{MSE}}                    & \multirow{2}{*}{\text{KL}}           \\
 &    \text{trainable parameters}      & \text{to reconstructor network} &                &      &              \\
\midrule
$1$                           & 805 & 805        & 0.0325  & 0.0079 & 0.0448                                                \\
$10$                          & 7960 & 7960       & 0.0392    & 0.0089 & 0.0516                                                                                                         \\
$100$                         & 79420 & 79420     & 0.0338 & 0.0079 & 0.0258                                                                              \\
$100$*  & 79420 & 1010      & 0.0477  & 0.0124 & 0.0840                                                                             \\                          
\bottomrule
\end{tabular}
}
\end{table}

Here we expand on our investigation around the interplay between size of the released model and reconstruction success by varying the width of the released model hidden layer between 1, 10, and 100 for the MNIST dataset.
\Cref{fig:mnist_mlp_arch} and \Cref{tab: mnist_mlp_arch_tab}
show that the width of the hidden layer does
not significantly affect reconstruction, as all have an average MSE on the $1K$ targets far below the NN oracle of 0.0232.
We also investigate the case in which the attacker only
trains and evaluates the attack using the second layer of the \attackin{} (and released model) for the MNIST dataset.
Because the architecture of the 1-hidden layer released model has 7850 parameters in the first layer and 110 in the second, by only using the second layer as inputs to the attack, we reduce the dimensionality of attack inputs by 98\%.
Using only the second layer of released model for the input to the reconstructor network marginally increases average MSE, but substantially reduces the dimensionality of inputs, thereby improving efficiency during training the attack.

\subsection{Variance over Different Initializations}
\label{app: result_variance}

\begin{figure}[H]
\captionsetup{width=0.5\textwidth}
  \centering
\begin{subfigure}[t]{.24\textwidth}
\captionsetup{width=0.95\textwidth}
\centering
    \includegraphics[width=0.99\linewidth]{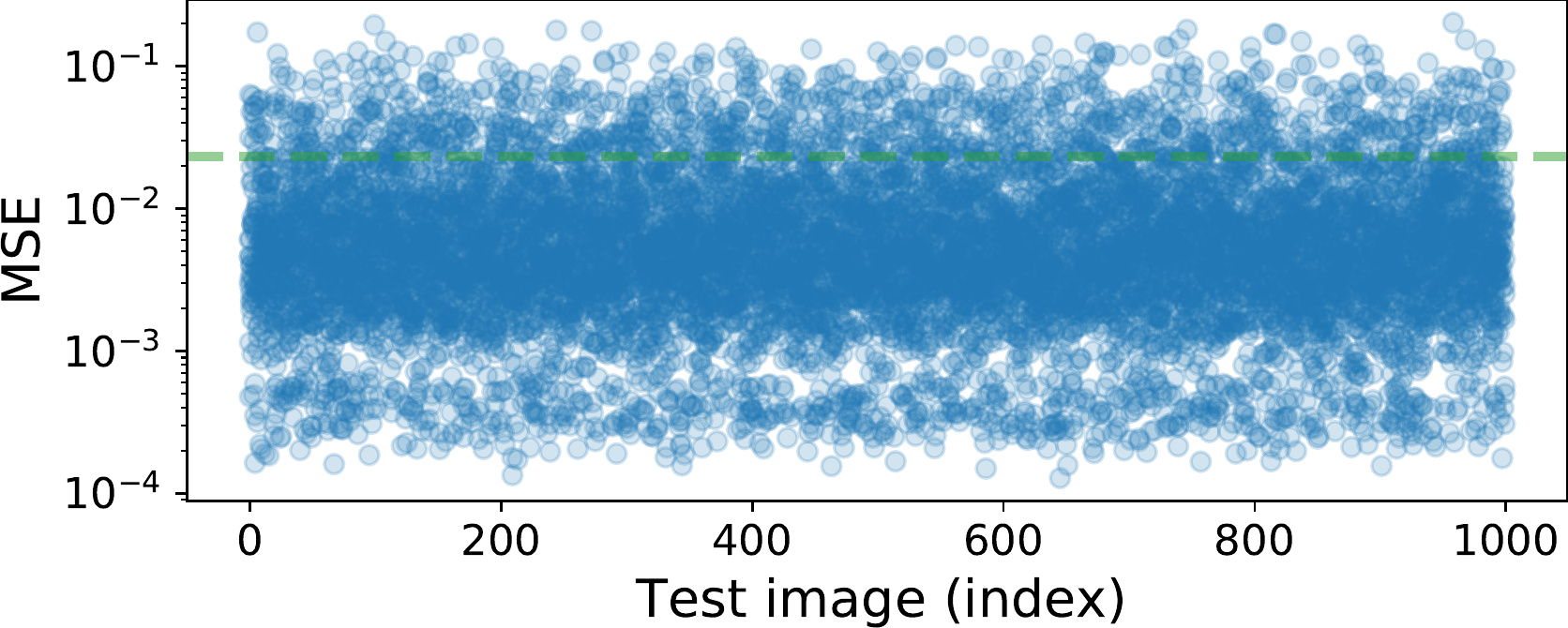}
    \caption{Variance in reconstructions over ten different experiment runs (ten different initial parameter configurations).}
    \label{fig:mnist_varying_seeds_width_10}
\end{subfigure}\begin{subfigure}[t]{.24\textwidth}
\captionsetup{width=0.95\textwidth}
\centering
    \includegraphics[width=0.99\linewidth]{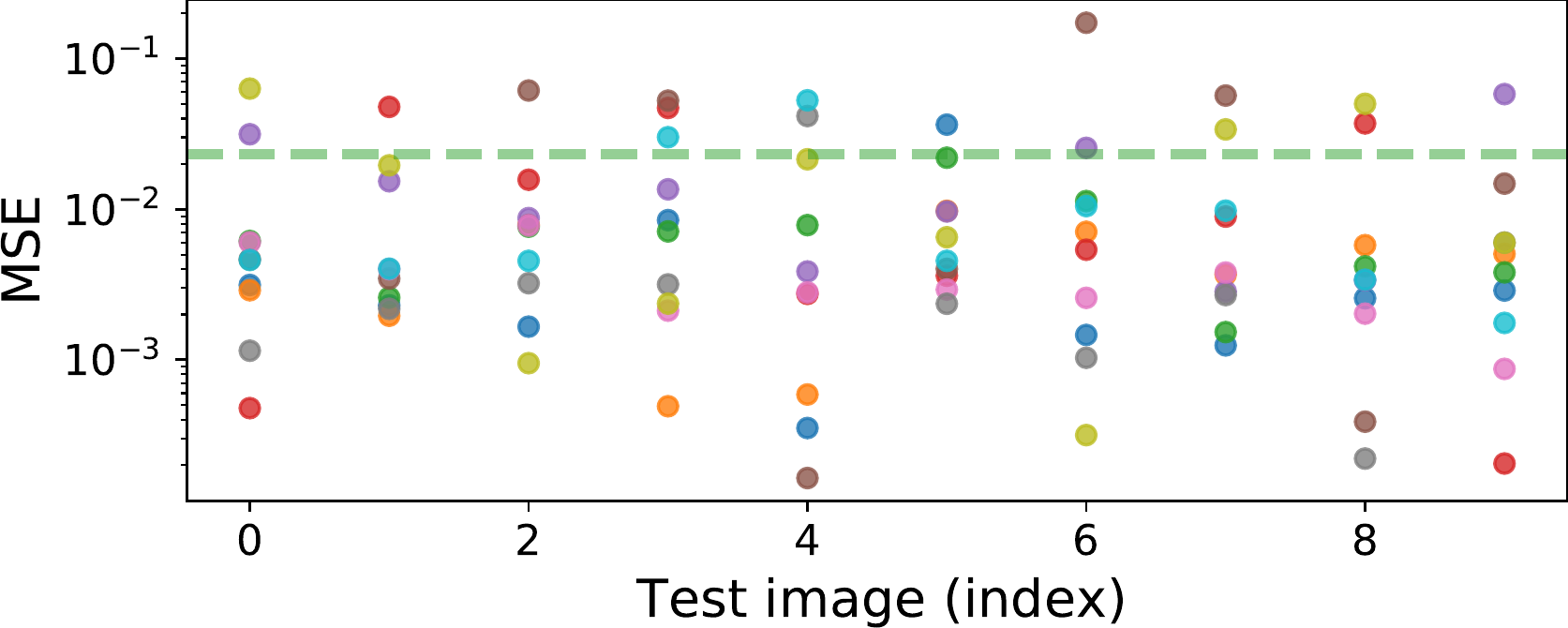}
    \caption{Each seed is a different colour (shown for targets out of 1K).}
    \label{fig:mnist_varying_seeds_colored}
\end{subfigure}\caption{Variance in MSE across ten different experimental runs each with a different initial parameter configuration for the release model, for the MNIST dataset.}
\label{fig:mnist_varying_seeds}
\end{figure}

We repeat our default reconstruction attack on MNIST ten times, where at each repetition we will sample a new seed controlling initial parameters over that experiment run.
We the measure how consistent our experimental results are for different choices of seeds. 
From \Cref{fig:mnist_varying_seeds} we can see that while there is variance in results, almost all reconstructions lie close or below the NN oracle distance in \Cref{fig:mnist_oracle_hist}. 
One may wonder if certain seeds are more amenable to attacks than others? 
That is, are there configurations of initial released model parameters that result in better reconstructions across all target points in comparison to other initializations?
We show this is not the case in \Cref{fig:mnist_varying_seeds_colored} by marking each seed in a different color and showing the MSE for ten test targets over the ten different seeds; this highlights that the ordering of seeds with respect to MSE changes for different images.

\subsection{Effect of Batch Size, Learning Rate, and Fixed Set Size on MNIST Reconstructions}
\label{app: mnist_batch_size_lr_epoch_fixed_size}

\begin{figure}[H]
\captionsetup{width=0.5\textwidth}
  \centering
\begin{subfigure}[t]{.16\textwidth}
\centering
    \includegraphics[width=0.99\linewidth]{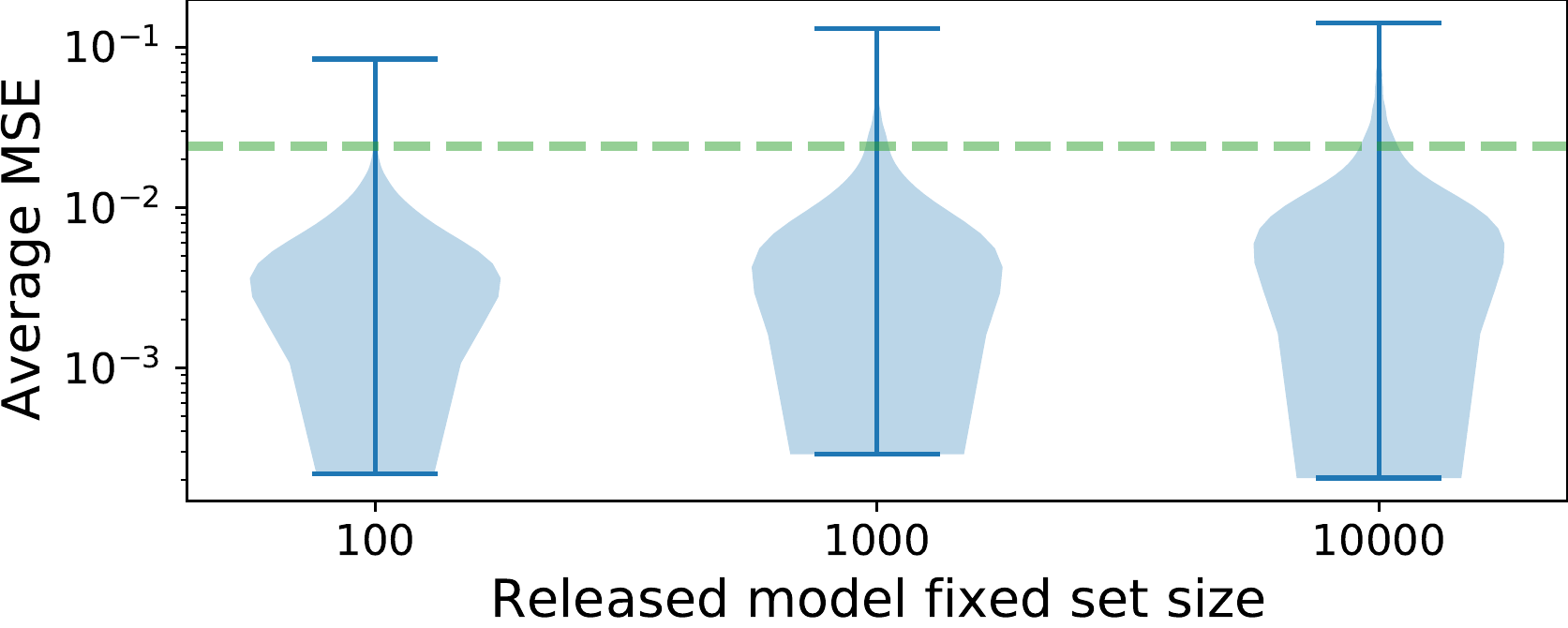}
\label{fig:}
\end{subfigure}\begin{subfigure}[t]{.16\textwidth}
\centering
    \includegraphics[width=0.99\linewidth]{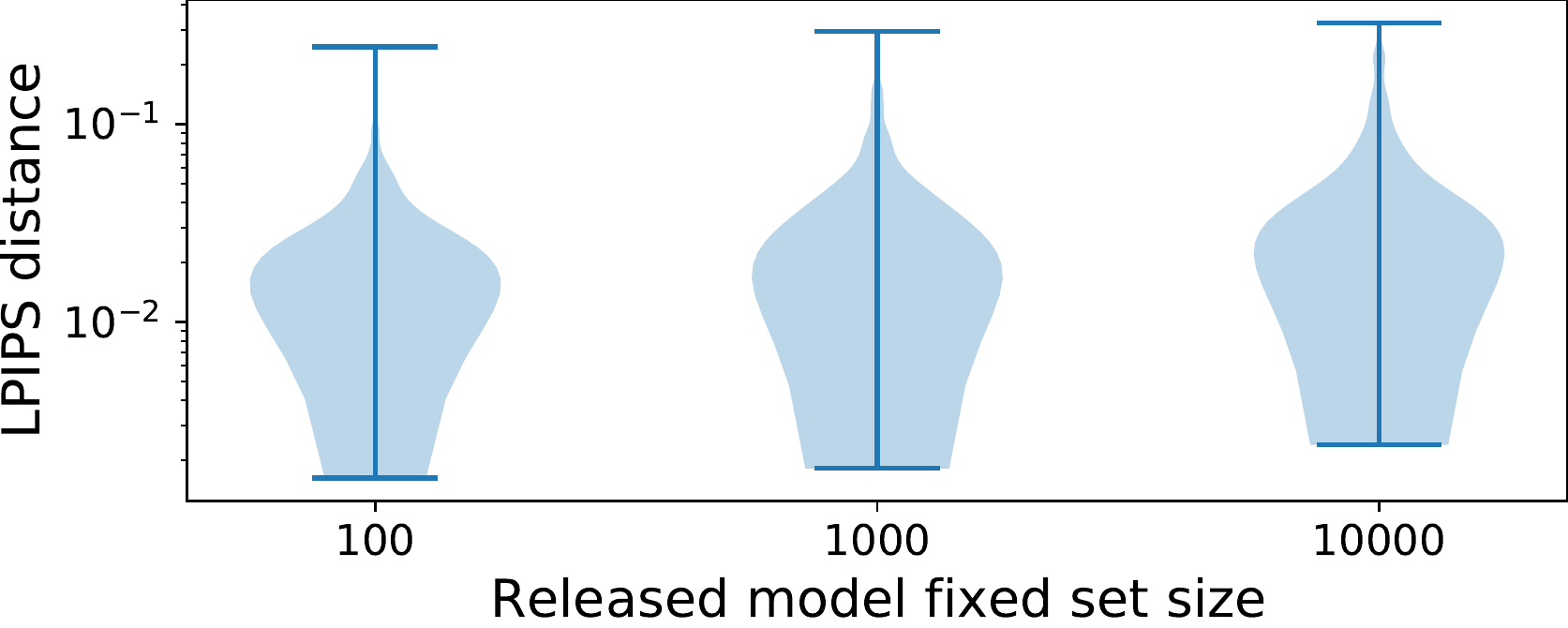}
\label{fig:}
\end{subfigure}\begin{subfigure}[t]{.16\textwidth}
\centering
    \includegraphics[width=0.99\linewidth]{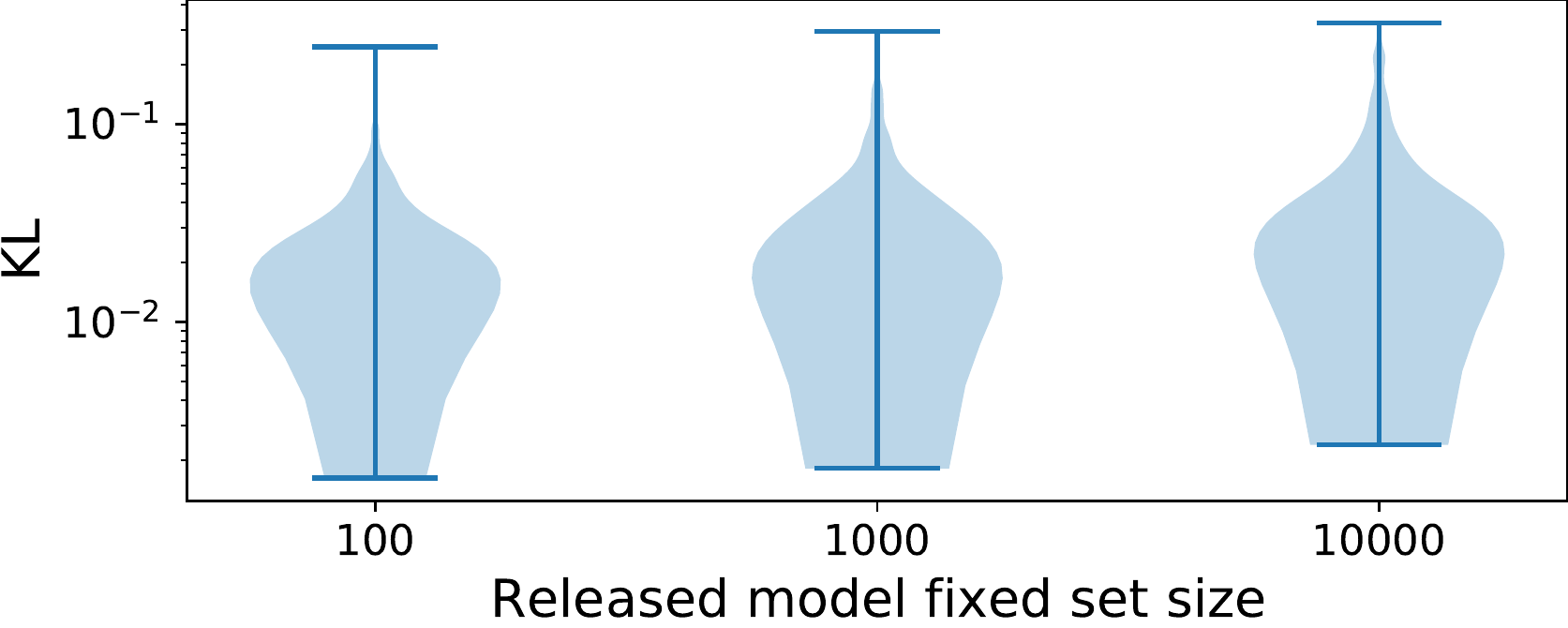}
\label{fig:}
\end{subfigure}\caption{Reconstruction metrics (on MNIST) for different fixed set sizes, $\Dfixed$. We observe a slight increase in in reconstruction error across all metrics when $|\Dfixed|$ grows.}
\label{fig:mnist_fixed_set_size}
\end{figure}

We summarize other aspects of the training procedure of the released model that impact the quality of reconstructions on the MNIST dataset.
Firstly, larger learning rates when training the released model can negatively influence the quality of reconstructions in SGD (i.e.\ using mini-batches), but in full-batch gradient descent we observed similar a MSE across different learning rates (cf.\ \Cref{fig:mnist_batch_size_appendix}).
Secondly, in \Cref{fig:mnist_fixed_set_size}, we measure the influence of the size of the fixed set,
by varying it between 100, 1$K$, and 10$K$ (which is the default in our MNIST experiments).
Interestingly, we only observe
a small decrease in MSE as the fixed set size decreases.

\subsection{Correlation between Metrics on CIFAR-10}
\label{app: cifar10_correlation}

\begin{figure}[H]
\captionsetup{width=0.5\textwidth}
  \centering
\begin{subfigure}[t]{.24\textwidth}
\centering
    \includegraphics[width=0.99\linewidth]{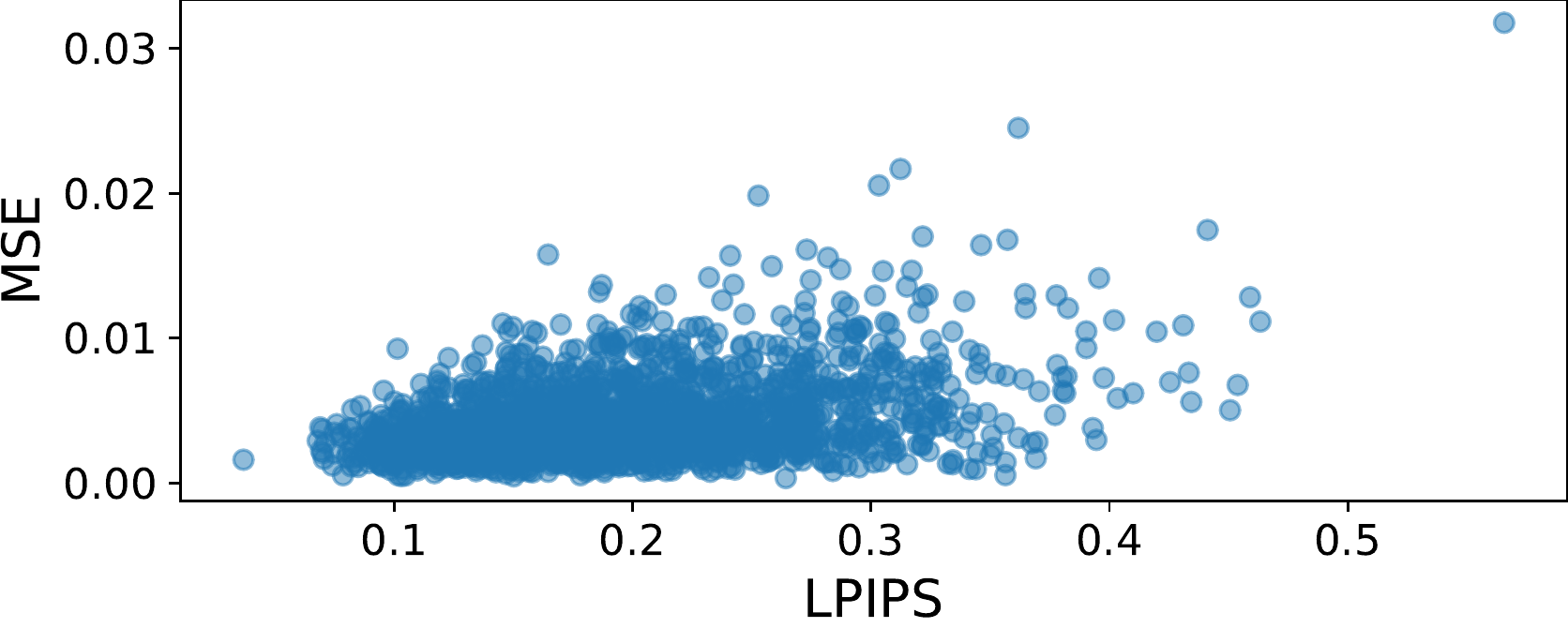}
\label{fig:}
\end{subfigure}\begin{subfigure}[t]{.24\textwidth}
\centering
    \includegraphics[width=0.99\linewidth]{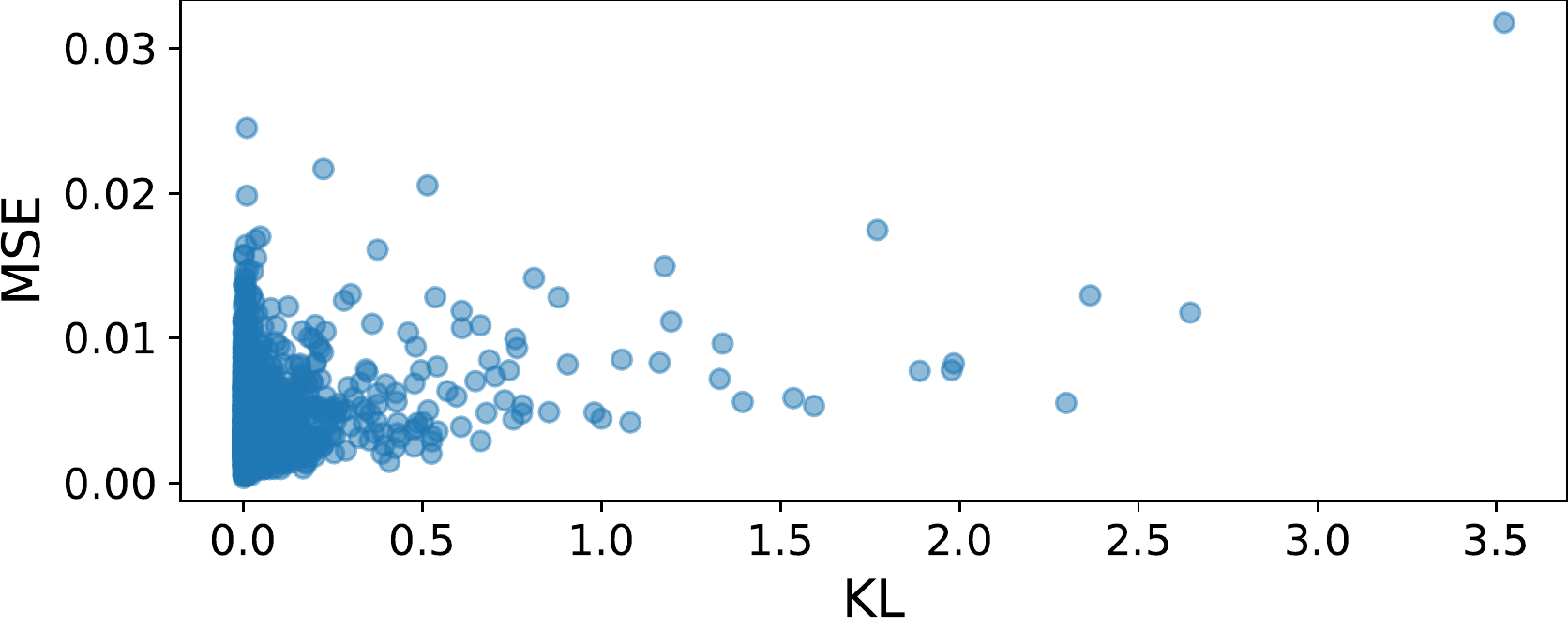}
\label{fig:}
\end{subfigure}\caption{Correlation between different metrics over 1$K$ released model targets for CIFAR-10.}
\label{fig:cifar10_correlation}
\end{figure}

We measure the correlation between MSE, LPIPS, and KL on the 1K test targets for CIFAR-10 in \Cref{fig:cifar10_correlation}, and observe similar relationships between the metrics as in \Cref{fig:mnist_attack_train_size_l2_lpips} and \Cref{fig:mnist_attack_train_size_l2_kl}.

\subsection{Expanded Investigation into Factors that Affect CIFAR-10 Reconstructions}
\label{app: cifar10_factors}

In \Cref{tab: cifar10_master_table}, we investigate reconstructions under different released model optimizers, fixed set sizes and number of training epochs.
We find that the dominant factor in reconstruction fidelity is the number of training epochs of the released model. 
A smaller number of training epochs and smaller fixed set size improves the quality of reconstructions.
Note also that the ability to reconstruct does not seem to be correlated with overfitting or (standard) membership inference success, and that the attack succeeds for different choices of optimizer. 

Due to the size of the released model, and other restrictions such as full-batch training with no regularizers, the test accuracy of the released model for which we can successfully perform attacks is approximately 35-50\%. 
Improving the efficiency of the attack such that it can scale to larger released models is a challenge to address in future work.

\subsection{Correlation between Gradient Norm and Ease of Reconstruction}
\label{app: grad_norm_exp}

\begin{figure}[H]
\captionsetup{width=0.99\linewidth}
\centering
  \centering
  \includegraphics[width=0.75\linewidth]{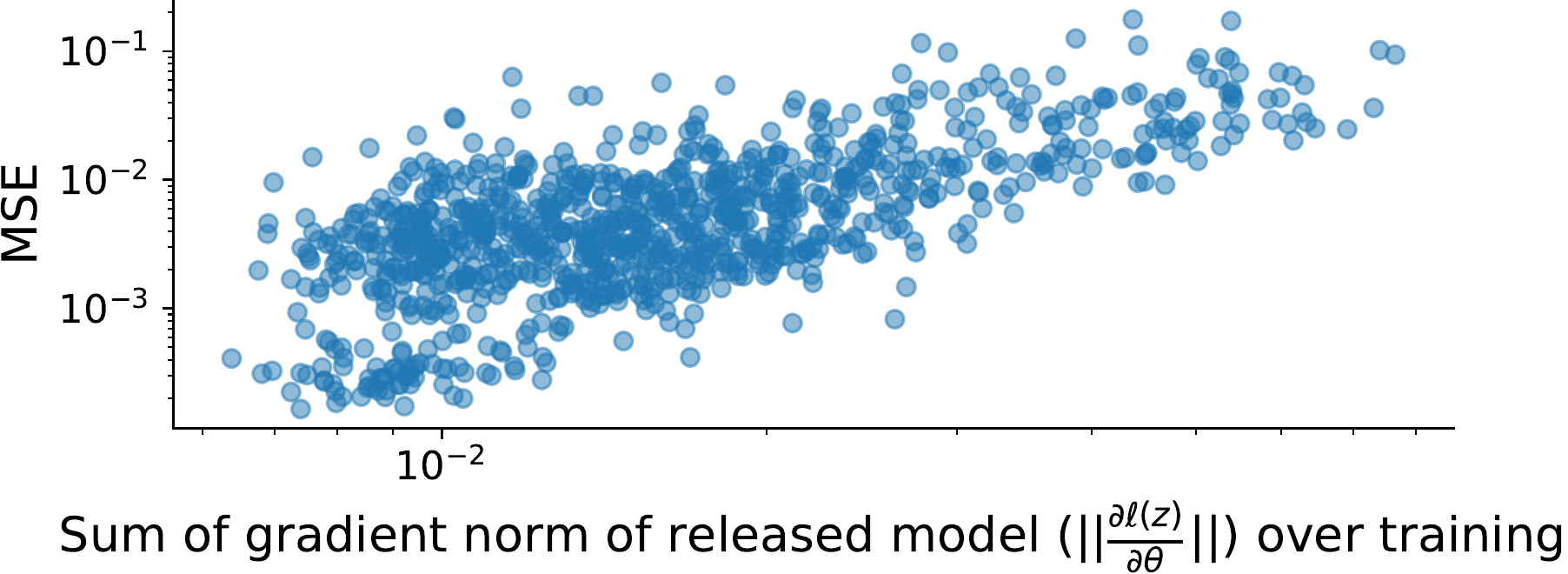}
  \caption{For each target point $z$ in the $1K$ released model target set, we plot MSE against gradient norm $\norm{\frac{\partial \ell(z)}{\partial \theta}}$ for MNIST.}
  \label{fig:l2_grad_norm}
\end{figure}

Following on from \Cref{app: pretrain_vs_init}, we investigate the relationship between reconstruction and the gradient norm of loss with respected to released model parameters computed on the target point through training.
Recent work on training data memorization~\cite{DBLP:conf/stoc/Feldman20} and individual privacy accounting in differential privacy \cite{DBLP:journals/corr/abs-2008-11193} have used the gradient norm of a model with respect to the loss induced by a training point as a measure of memorization or privacy leakage. 
In \Cref{fig:l2_grad_norm}, we evaluate the MSE between target and reconstructions for each 1K target point on MNIST, and also plot the sum of gradient norms over training. 
The two quantities are weakly correlated with one another, however one may expect that the two would be inversely correlated if examples that are strongly memorized are easier to reconstruct -- targets with a larger gradient norm throughout training are outliers that the released model must necessarily memorize to perform well on (\cite{DBLP:conf/stoc/Feldman20}).
We conjecture that the effect we are observing stems from ``outliers'' that are harder to reconstruct not because the released model memorizes them less, but because they’re also outliers for the RecoNN, and therefore target where the reconstructor also performs poorly.

}

\subsection{Fine-Grained Analysis of CIFAR-10 Reconstructions over Released Model Training Epochs}
\label{app: finegrained_cifar10}

\begin{figure}[H]
  \centering
    \includegraphics[width=0.75\linewidth]{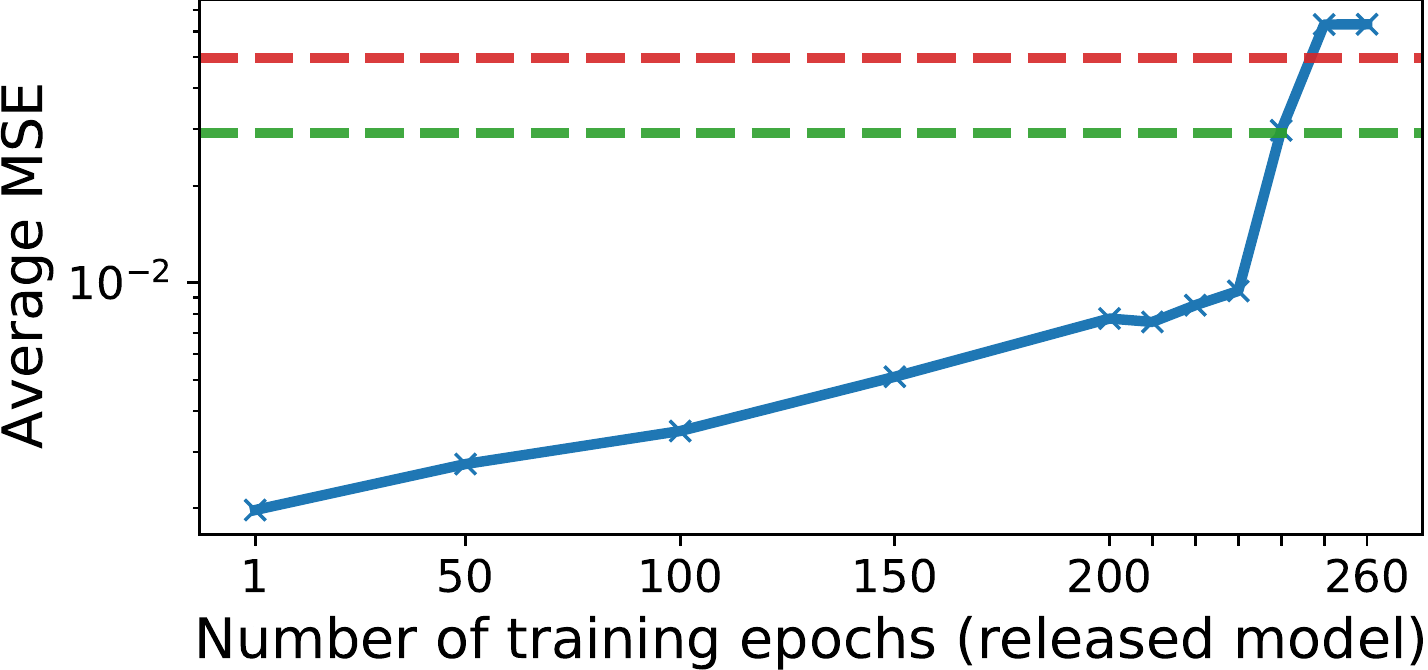}
\caption{How average MSE increases with the number of training epochs of the released model for CIFAR-10.} 
\label{fig:finegrained_epochs_cifar10}
\end{figure}

\arxiv{As we observed in \Cref{app: cifar10_factors}, reconstructing CIFAR-10 images is sensitive to the number of training epochs of the released model.}
\main{Reconstructing CIFAR-10 images is sensitive to the number of training epochs of the released model.}
We perform a fine-grained analysis to inspect at what epoch the attack becomes unsuccessful.
This can be seen in \Cref{fig:finegrained_epochs_cifar10}, where we plot average MSE over 1K released model targets as a function of the number of training epochs. 
MSE slowly increases with number of epochs up until approximately 240-250 epochs, at which point we observe that ``reconstructability'' undergoes a phase transition.
Initially, we conjectured this was due to non-determinism from GPU training increasing the variance of shadow model parameters for a larger number of training epochs. 
However, when we implemented shadow model training in a deterministic set-up (using TPUs) we observed no difference in experimental outcomes.
We leave a more in-depth investigation into the relationship between reconstruction success and number of training epochs for future work.

\subsection{ReLU Activations in Released Model}
\label{app: relu_act}

\begin{figure}[H]
\captionsetup{width=0.99\linewidth}
  \centering
    \includegraphics[width=0.75\linewidth]{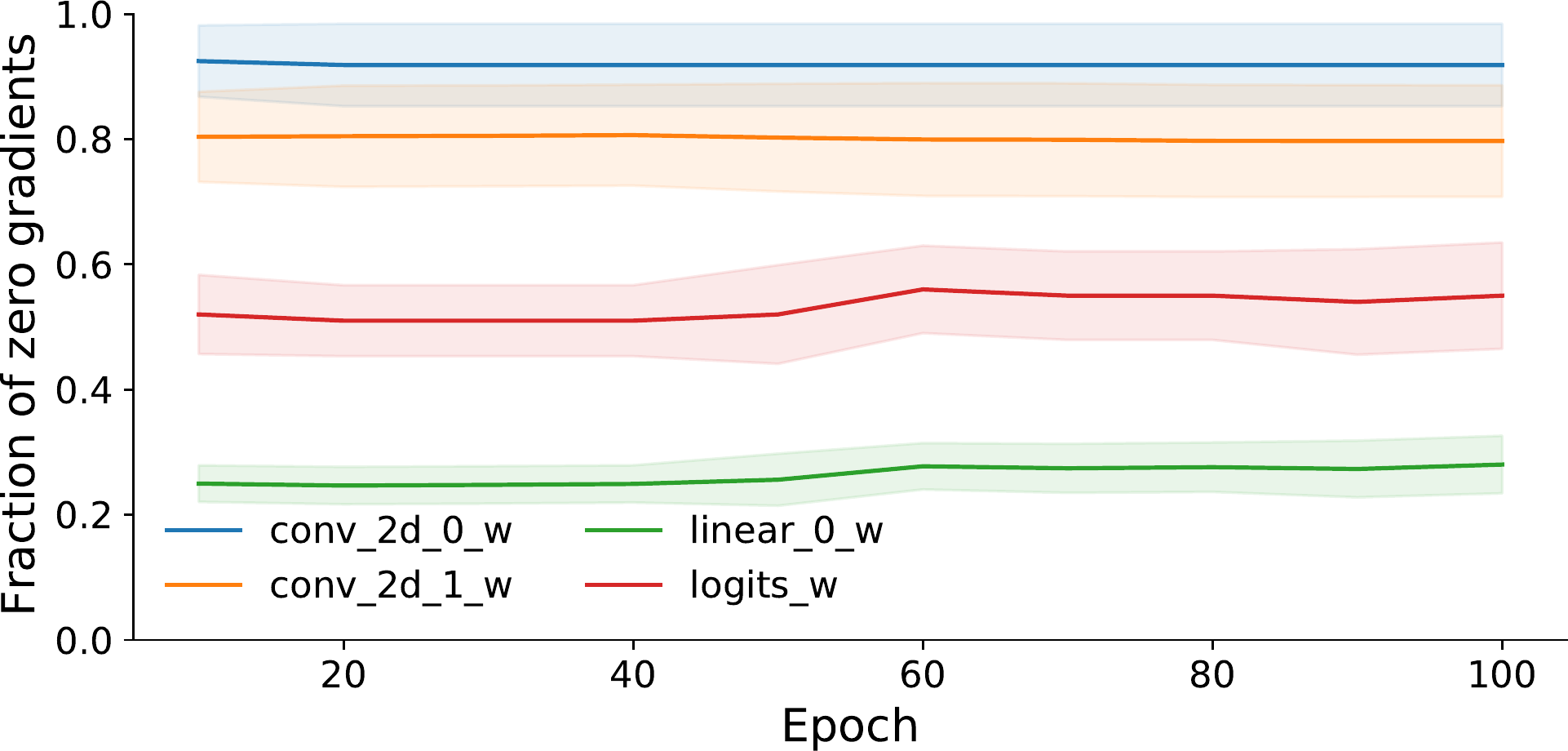}
\caption{Evidence on CIFAR-10 reconstruction task that ReLU activations make reconstruction attacks harder. For the target, $z$, we plot $\frac{\partial \ell(z)}{\partial \theta}$ for each layer in the released model $\theta$, throughout training. A large fraction of these gradients are zero, implying less influence of this additional point on the trained model, in comparison to other activations that have non-zero gradients everywhere.}
\label{fig:relu_act}
\end{figure}

\begin{table}[htp]
\centering
\caption{Comparison of reconstructions for different released model activations on MNIST. Please refer to \cite{jaxact} for a description of each activation function.}
\label{tab: act_comparison}
\begin{tabular}{lC}
\toprule
Activation   & \text{Average MSE over 1K test targets} \\
\midrule
ReLU         & 0.0182                           \\
$\max(-0.5, x)$ & 0.0096                           \\
ELU          & 0.0089                           \\
Sigmoid      & 0.0085                           \\
Softplus     & 0.0083                           \\
Swish        & 0.0091                           \\
Leaky ReLU   & 0.0092                           \\
Tanh         & 0.0086                           \\
CELU         & 0.0077                           \\
SELU         & 0.0083                           \\
GELU         & 0.0088                           \\
Identity     & 0.0085                        
   \\
\bottomrule
\end{tabular}
\end{table}

As we saw in \Cref{ssec:reconctruction-factors}, released models with ReLU activations tend to be harder to attack in comparison to other activation functions with non-zero gradients almost everywhere, and result in poor quality reconstructions (an MSE larger than the NN oracle distance). 
We conjecture that this is caused by a large fraction of parameters receiving zero gradients at each step of training, thereby diminishing the mutual information shared between model parameters and the unknown target training point.
In \Cref{fig:relu_act}, for each layer of the released model, we show the fraction of parameters that received zero gradient when computing the loss of the unknown training point.
Over 80\% of the parameters in the convolutional layers have zero gradients.
Additionally, in \Cref{tab: act_comparison} we compare reconstructions against released models that employ different activation functions, and find that ReLU remains the outlier. 
Note that we also reconstruct against a released model that uses a modified version of ReLU that has zero gradient for $x<-0.5$, and find that allowing a small negative signal is enough to reach parity with reconstruction MSE on smooth activations or activations that contain a non-zero signal almost everywhere.

\begin{table}[H]
\caption{Experimental setup.}
\label{tab:experimental-setup}
\begin{tabular}{@{}clrr@{}}
\toprule
                       &               & \multicolumn{1}{c}{MNIST} & \multicolumn{1}{c}{CIFAR10} \\ \midrule
\multirow{5}{*}{Data}  & Resolution    & $28\times28$ (grayscale)                     & $32\times32$ (RGB)                       \\
                       & Size          & $70K$                       & $60K$                         \\
                       & Fixed size    & $10K$                       & $5K$                         \\
                       & Shadow size   & $59K$                       & $54K$                         \\
                       & Test targets      & $1K$                        & $1K$                          \\ \midrule
\multirow{4}{*}{$\theta, \bar{\theta}$} & Type          & MLP                       & CNN                         \\
                       & Architecture  & 1-hidden layer, width $10$         & \Cref{tab: cifar10_released_model}                    \\
                       & Activations   & ELU                       & ELU                           \\
                       & Parameters    & $8K$                        & $55K$                         \\ \midrule
\multirow{4}{*}{$\phi$}   & Type          & MLP                       & Transposed CNN              \\
                       & Architecture  & 2-hidden layers, width $1K$        & \Cref{tab: cifar10_attack_model}                   \\
                       & Activations   & ReLU                      & ReLU                           \\
                       & Parameters    & $9.7M$                      & $226M$                        \\ \midrule
\multirow{5}{*}{$A$}    & Algorithm     & GD+Momentum                      & GD+Momentum                        \\
                       & Loss          & Cross-entropy     &  Cross-entropy       \\
                       & Learning rate & $0.2$                       & $0.01$                        \\
                       & Momentum      & $0.9$                       & $0.9$                         \\
                       & Epochs        & $100$                       & $100$                         \\ \midrule
\multirow{6}{*}{$R$}     & Algorithm     & RMSProp                   & Adam                        \\
                       & Loss          & MAE+MSE                   & +LPIPS+Discriminator \\
                       & Learning rate & $0.001$                  & $0.0001$                    \\
                       & Weight decay & $0$                  & $0.0001$                    \\
                       & Batch size    & $128$                       & $128$                         \\
                       & Epochs        & $100$                       & $1000$                        \\ \bottomrule
\end{tabular}
\end{table}

\begin{table}[H]
\centering
\caption{CIFAR-10 released model, $\theta$.}
\label{tab: cifar10_released_model}
\begin{tabular}{ll}
\toprule
Layer   & Parameters\\
\midrule
Convolution & $16$ filters of $4\times4$, strides $2$ \\
Convolution & $32$ filters of $4\times4$, strides $1$ \\
Fully connected & $10$ units \\
Softmax & $10$ units \\
\bottomrule
\end{tabular}
\end{table}

\begin{table}[H]
\centering
\caption{CIFAR-10 reconstructor network, $\phi$.}
\label{tab: cifar10_attack_model}
\begin{tabular}{ll}
\toprule
Layer   & Parameters\\
\midrule
Fully connected & $4096$ units \\
Reshape & $64\times64$ \\
Transposed convolution & $32$ filters of $5\times5$, strides $2$ \\ 
Transposed convolution & $3$ filters of $5\times5$, strides $2$ \\ 
\bottomrule
\end{tabular}
\end{table}

\begin{table}[H]
\centering
\caption{CIFAR-10 attack PatchGAN Discriminator model.}
\label{tab: cifar10_attack_discrim_model}
\begin{tabular}{ll}
\toprule
Layer   & Parameters\\
\cmidrule{1-2}
Convolution & $64$ filters of $4\times4$, stride $2$ \\
Convolution & $128$ filters of $4\times4$, stride $2$ \\
Convolution & $256$ filters of $4\times4$, stride $2$ \\
Convolution & $512$ filters of $4\times4$, stride $1$ \\
Convolution & $1$ filter of $4\times4$, stride $1$ \\
\bottomrule
\end{tabular}
\end{table}

\arxiv{
\begin{table*}[]
\centering
\caption{Reconstruction metrics for different released model learning hyperparameters on the CIFAR-10 dataset. We also include a column denoting the membership inference AUC for the released model using \cite{song_2020} and show that vulnerability to reconstruction attacks occurs even when standard membership inference attacks pose little risk. For reference the nearest neighbor oracle distance for the CIFAR-10 dataset is $0.0291$, and so we judge a reconstruction to be successful if it is below this value. The Adam optimizer is set with a learning rate of $0.002$.}
\label{tab: cifar10_master_table}
\resizebox{\textwidth}{!}{
\begin{tabular}{lcCCCCCCc}
\toprule
\multirow{2}{*}{Fixed set size}       & \multirow{2}{*}{Optimizer} & \multirow{2}{*}{Training epochs}                      & \multirow{2}{*}{MSE} & \multirow{2}{*}{LPIPS} & \text{Released model} & \text{Released model} & \multirow{2}{*}{Membership AUC} & \multirow{2}{*}{OOD shadow target}\\
       & & & & & \text{train accuracy} & \text{test accuracy} &  & \\
\cmidrule{1-9}
     \multirow{6}{*}{$1K$}  & \multirow{3}{*}{GD + momentum} & 100               &   0.0041       &  0.1861     & 0.451                         & 0.327                        &        0.61        & \multirow{18}{*}{\xmark}  \\
                          &              & 250                                &  0.0126        & 0.3117       & 0.759                         & 0.324                        &       0.80         &                                          \\
                          &              & 500                               &          0.0170 & 0.3594	       & 0.998                         & 0.328                        &       0.89         &                                        \\
                          \cmidrule{3-8}
                          & \multirow{3}{*}{Adam}             & 100          &          0.0052 & 0.2042	       & 0.952                          & 0.299                        &       0.87         &                                        \\
                          &              & 250                               &          0.0065 & 0.2355	       & 1.000                          & 0.300                        &      0.92          &                                        \\
                          &              & 500                                &          0.0068 & 0.2428	       & 1.000                          & 0.303                        &      0.91          &                                        \\
\cmidrule{2-8}
     \multirow{6}{*}{$5K$}  & \multirow{3}{*}{GD + momentum}             & 100 &          0.0049 & 0.2070       & 0.392                         & 0.388                        &        0.51        &                                         \\
                      &              &                250                &          0.0201 & 0.3863	       & 0.461                         & 0.401                        &     0.60           &                                         \\
         &              & 500                                &          0.0761 & 0.5272	       & 0.610                         & 0.451                        &          0.61      &                                         \\
         \cmidrule{3-8}
                          & \multirow{3}{*}{Adam}             & 100          &          0.0052 & 0.2179       & 0.546                         & 0.395                        &        0.63        &                                         \\
                          &              &   250                             &          0.0062 & 0.2148	       & 0.844                         & 0.418                        &     0.75           &                                         \\
                          &              & 500                                &          0.0143 & 0.3460       & 0.999                         & 0.421                        &       0.87         &                                         \\
\cmidrule{2-8}
     \multirow{6}{*}{$10K$} & \multirow{3}{*}{GD + momentum}             & 100 &          0.0209 & 0.4184	      & 0.401                         & 0.397                        &        0.52        &                                         \\
                          &              &            250                    &          0.0385 & 0.4851	       & 0.420                         & 0.410                        &        0.52        &                                          \\
                          &              & 500                                &          0.0761 & 0.5287	       & 0.514                         & 0.473                        &          0.54      &                                         \\
                          \cmidrule{3-8}
                          & \multirow{3}{*}{Adam}             & 100          &          0.0081 & 0.2501       & 0.477                         & 0.443                        &      0.56          &                                         \\
                          &              & 250                                &          0.0208 & 0.3950       & 0.662                         & 0.470                        &         0.66       &                                         \\
                          &              & 500                               &          0.0357 & 0.4930        & 0.937                           & 0.474	                         & 0.73                &                                         \\
\cmidrule{2-9}
     \multirow{6}{*}{$50K$} & \multirow{3}{*}{GD + momentum}             &  100 &          0.0039 & 0.1929	      & 0.419                          & 0.400                         & 0.55               & \multirow{6}{*}{CIFAR-100}              \\
                          &              & 500                                &          0.0563 & 0.5260       & 0.503                          & 0.501                         &   0.51             &                                         \\
                          &             & 2000                                &          0.0693 & 0.5350       & 0.726                          & 0.633                         & 0.55               &                                         \\
                          \cmidrule{3-8}
                          & \multirow{3}{*}{Adam}             & 100          &          0.0109 & 0.3042	      & 0.422                               &  0.436                            &       0.51         &                                         \\
                          &              &    500                            &          0.0653 & 0.5585	       & 0.546                               &           0.541                   & 0.52                &                                         \\
                          &             & 2000                                &          0.0721 & 0.5727	       &                              0.730 & 0.596	                              &    0.57           &                                         \\
\bottomrule
\end{tabular}
}
\end{table*}
}

\arxiv{
\begin{figure*}[t]
\captionsetup{width=0.99\textwidth}
  \centering
\begin{subfigure}[t]{.24\textwidth}
\centering
    \includegraphics[width=0.99\linewidth]{l2_vs_released_model_batch_size_lr_0.01_momentum_0.0.pdf}
    \caption{Learning rate $0.01$, momentum $0$.}
    \label{fig:l2_lr_001_mom_0}
\end{subfigure}\begin{subfigure}[t]{.24\textwidth}
\centering
    \includegraphics[width=0.99\linewidth]{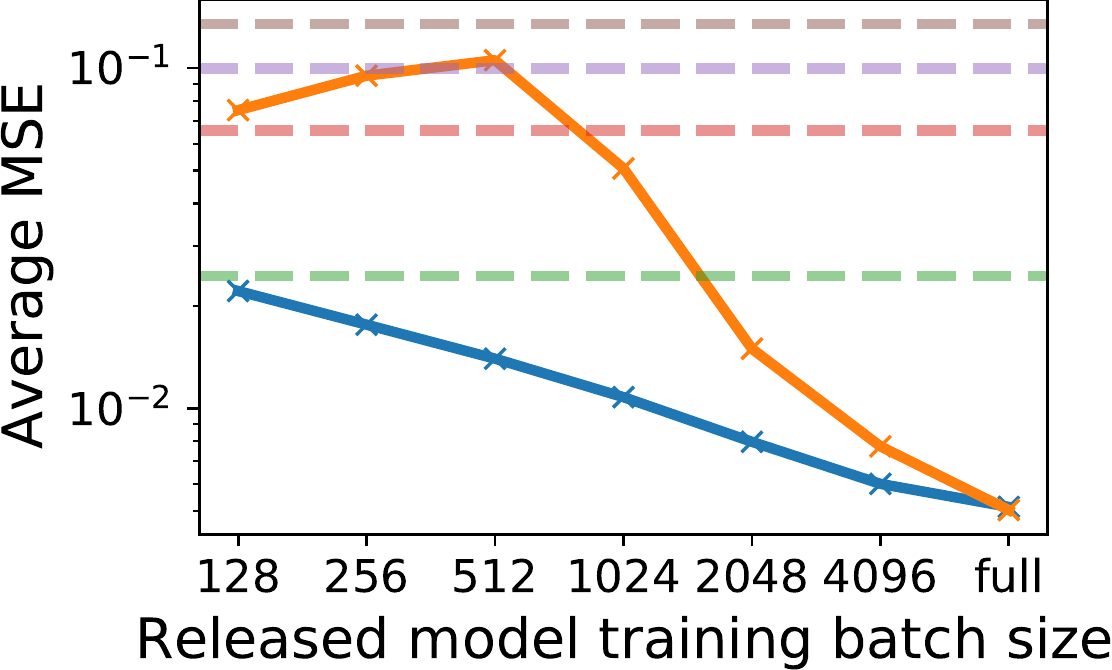}
    \caption{Learning rate $0.01$, momentum $0.9$.}
    \label{fig:l2_lr_001_mom_09}
\end{subfigure}\begin{subfigure}[t]{.24\textwidth}
\centering
    \includegraphics[width=0.99\linewidth]{l2_vs_released_model_batch_size_lr_0.2_momentum_0.0.pdf}
    \caption{Learning rate $0.2$, momentum $0$.}
    \label{fig:l2_lr_02_mom_0}
\end{subfigure}\begin{subfigure}[t]{.24\textwidth}
\centering
    \includegraphics[width=0.99\linewidth]{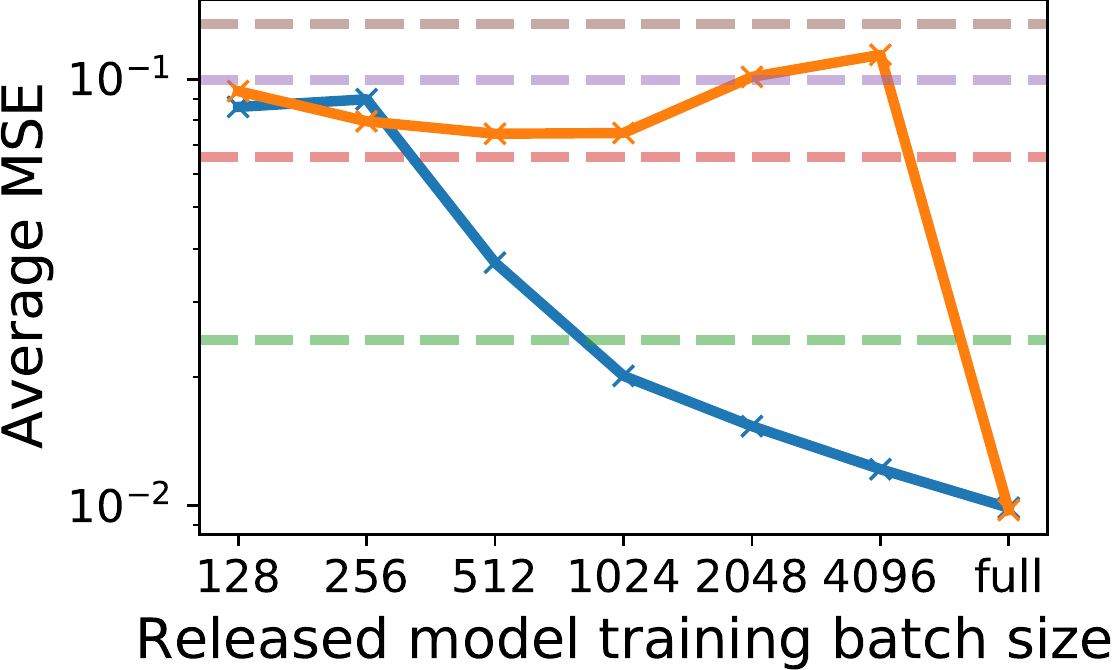}
    \caption{Learning rate $0.2$, momentum $0.9$.}
    \label{fig:l2_lr_02_mom_09}
\end{subfigure}

\begin{subfigure}[t]{.24\textwidth}
\centering
    \includegraphics[width=0.99\linewidth]{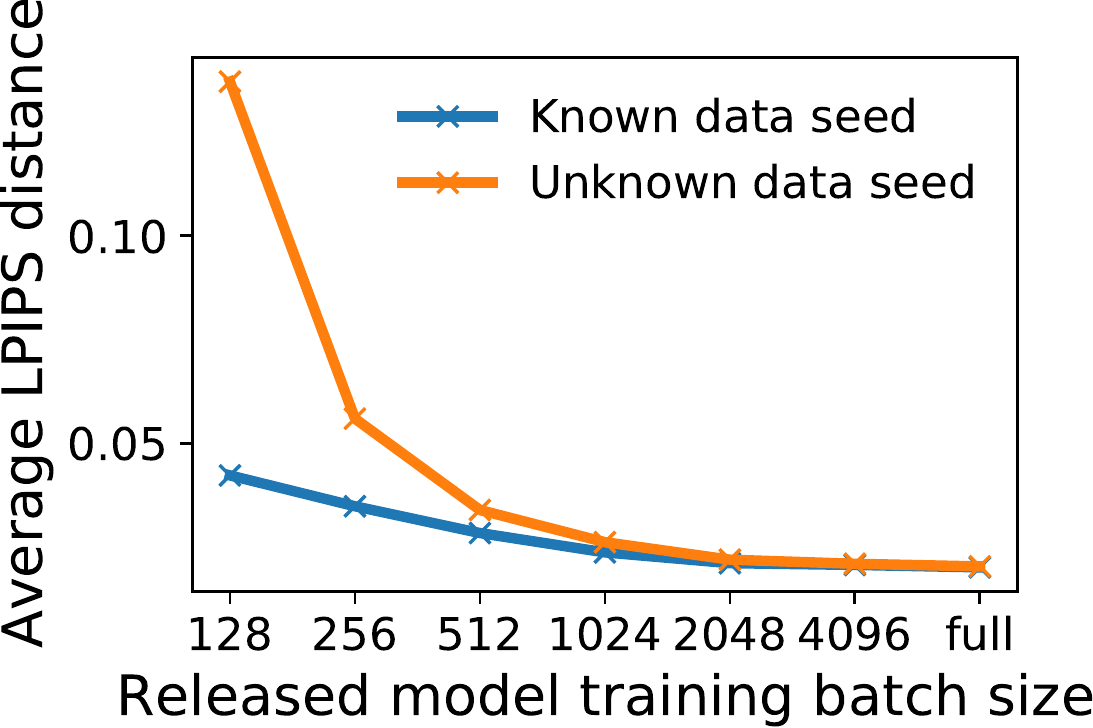}
    \caption{Learning rate $0.01$, momentum $0$.}
    \label{fig:lpips_lr_001_mom_0}
\end{subfigure}\begin{subfigure}[t]{.24\textwidth}
\centering
    \includegraphics[width=0.99\linewidth]{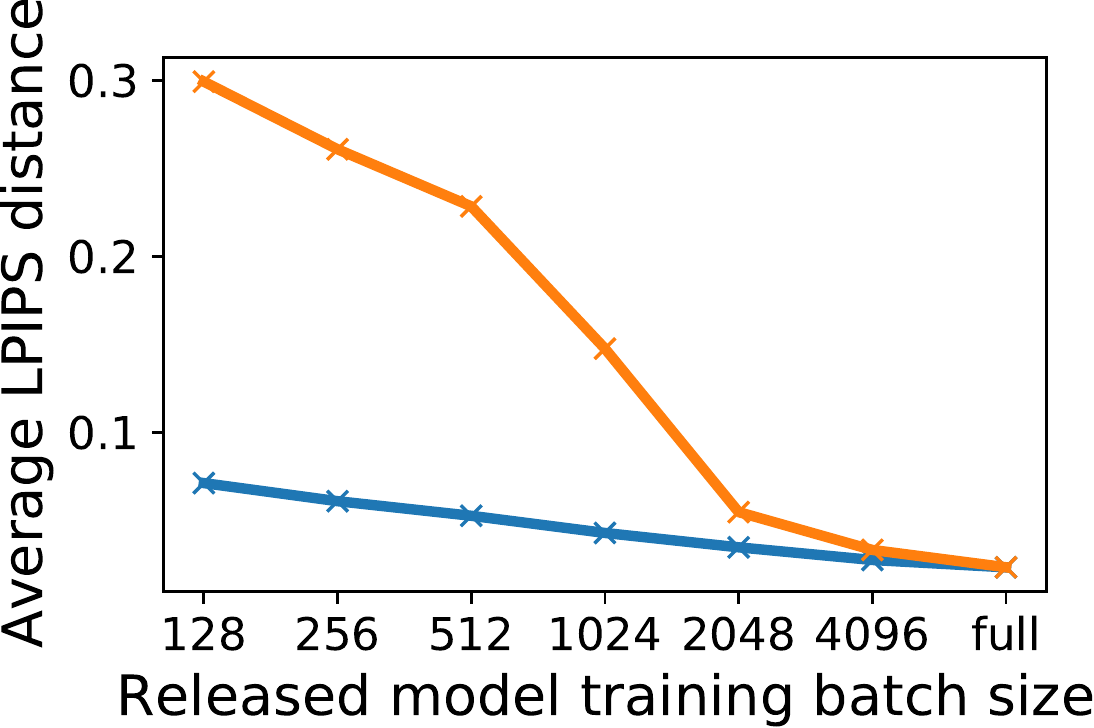}
    \caption{Learning rate $0.01$, momentum $0.9$.}
    \label{fig:lpips_lr_001_mom_09}
\end{subfigure}\begin{subfigure}[t]{.24\textwidth}
\centering
    \includegraphics[width=0.99\linewidth]{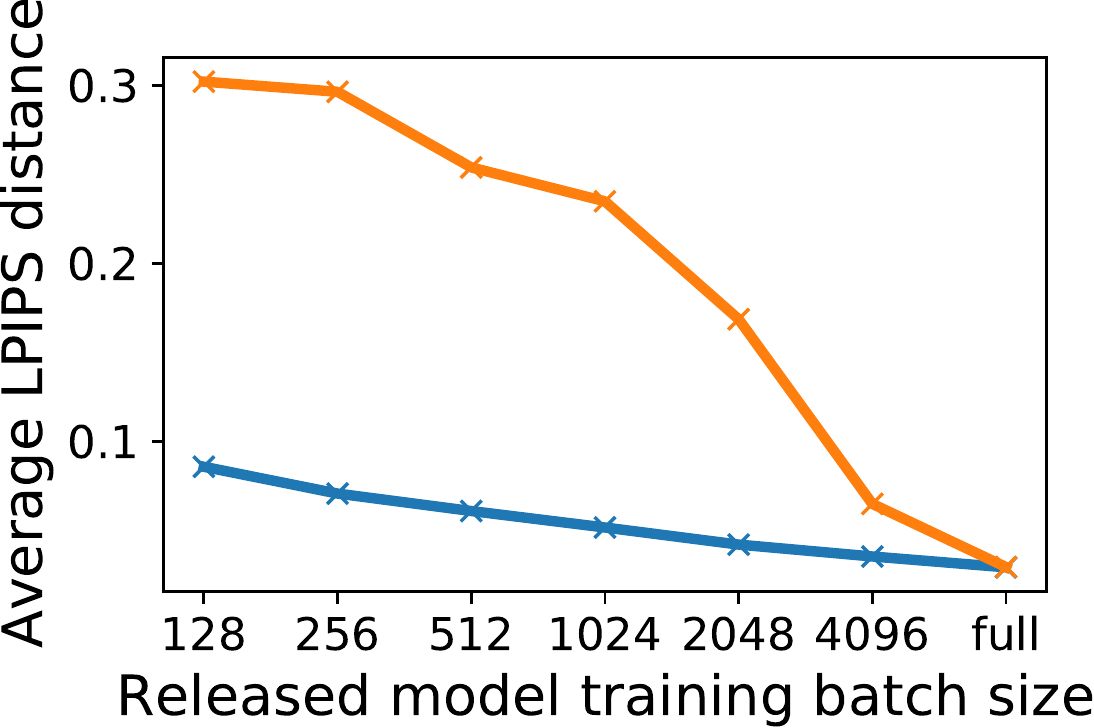}
    \caption{Learning rate $0.2$, momentum $0$.}
    \label{fig:lpips_lr_02_mom_0}
\end{subfigure}\begin{subfigure}[t]{.24\textwidth}
\centering
    \includegraphics[width=0.99\linewidth]{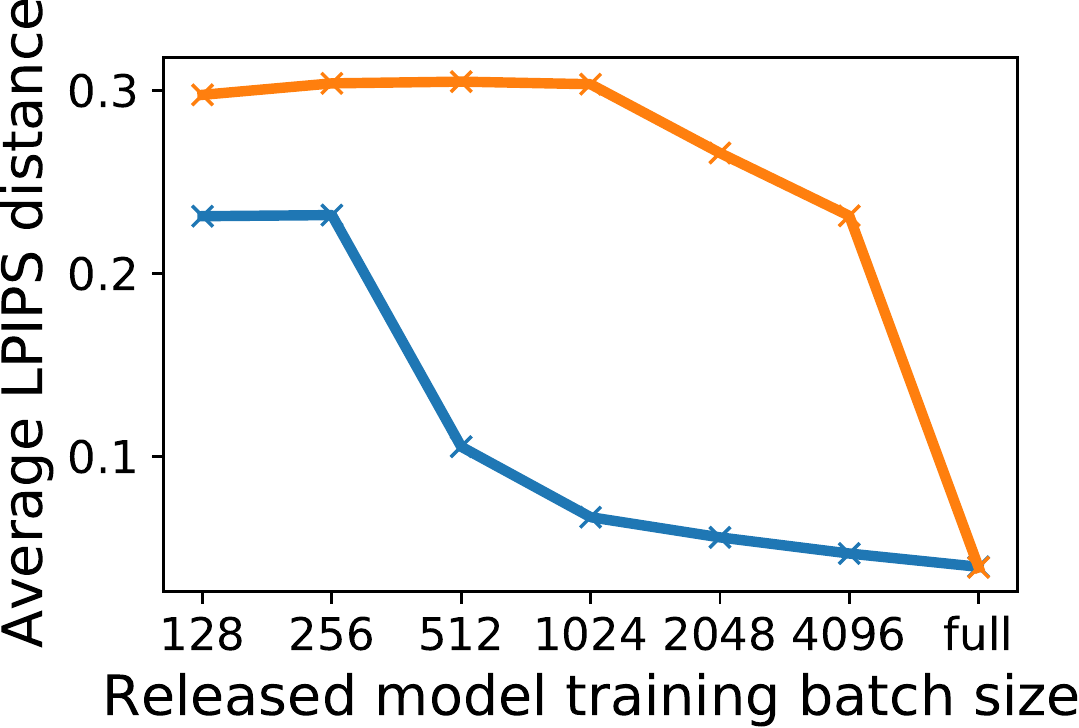}
    \caption{Learning rate $0.2$, momentum $0.9$.}
    \label{fig:lpips_lr_02_mom_09}
\end{subfigure}

\begin{subfigure}[t]{.24\textwidth}
\centering
    \includegraphics[width=0.99\linewidth]{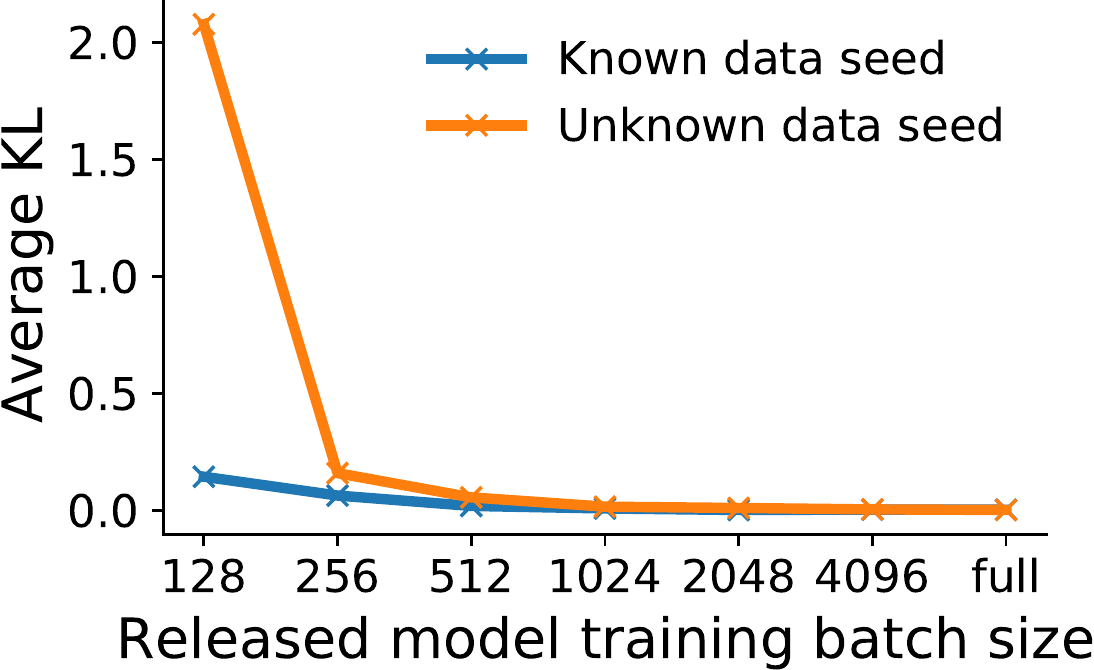}
    \caption{Learning rate $0.01$, momentum $0$.}
    \label{fig:kl_lr_001_mom_0}
\end{subfigure}\begin{subfigure}[t]{.24\textwidth}
\centering
    \includegraphics[width=0.99\linewidth]{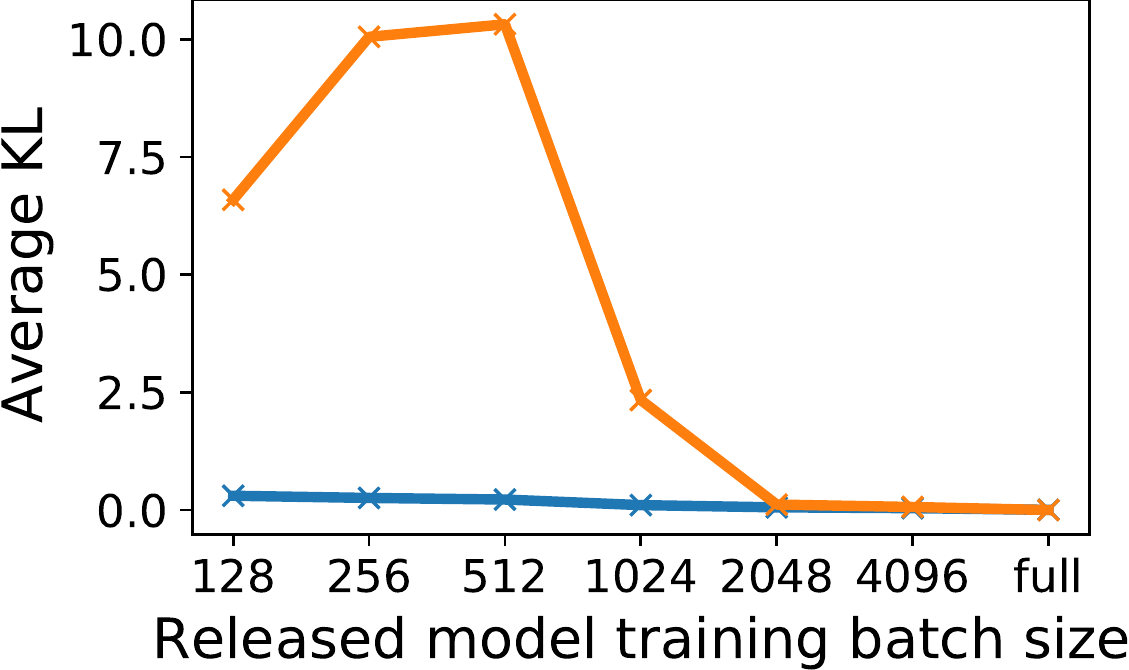}
    \caption{Learning rate $0.01$, momentum $0.9$.}
    \label{fig:kl_lr_001_mom_09}
\end{subfigure}\begin{subfigure}[t]{.24\textwidth}
\centering
    \includegraphics[width=0.99\linewidth]{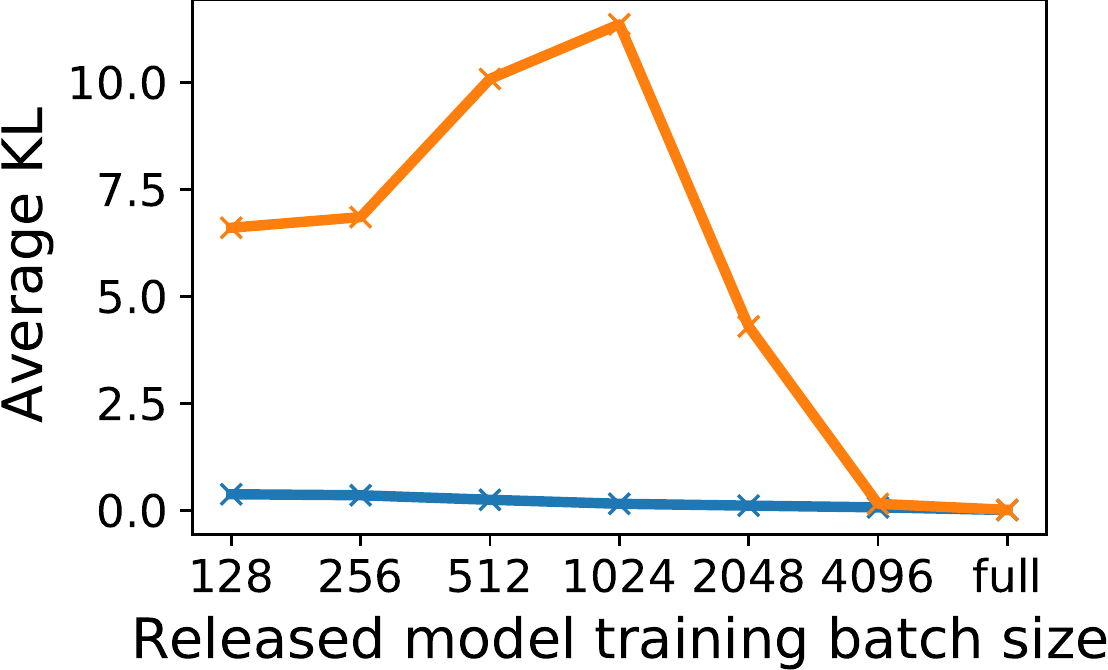}
    \caption{Learning rate $0.2$, momentum $0$.}
    \label{fig:kl_lr_02_mom_0}
\end{subfigure}\begin{subfigure}[t]{.24\textwidth}
\centering
    \includegraphics[width=0.99\linewidth]{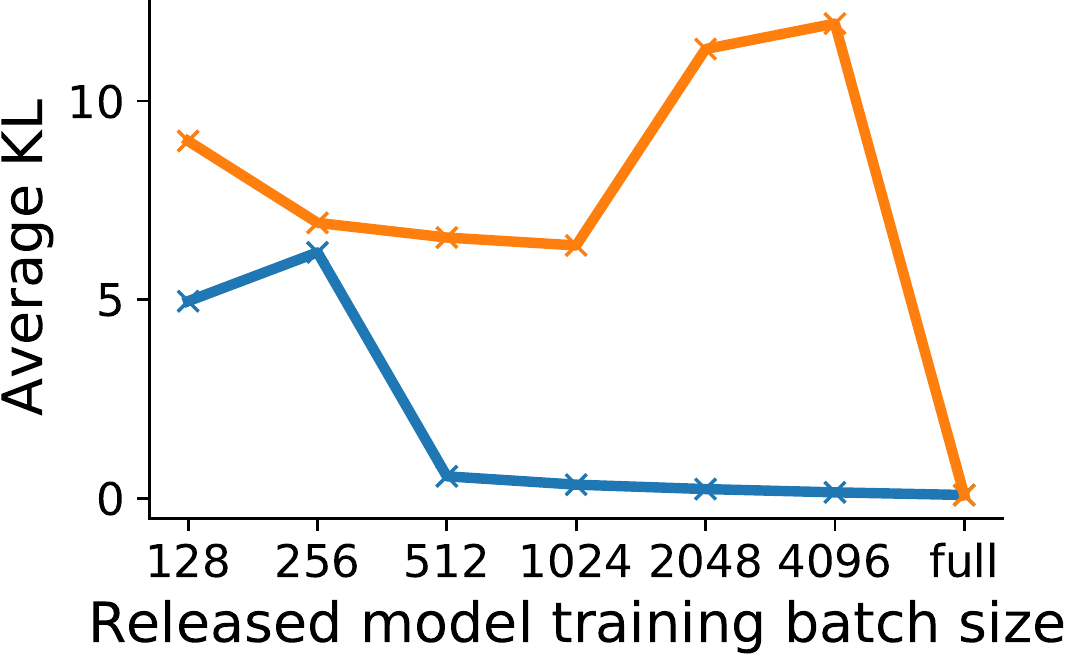}
    \caption{Learning rate $0.2$, momentum $0.9$.}
    \label{fig:kl_lr_02_mom_09}
\end{subfigure}

\caption{How randomness from data sub-sampling affects reconstructions on MNIST. Each row shows a different metric (MSE, LPIPS or KL) as a function of the batch size used in training the released model, for settings when the adversary does and does not know the seed from which data sub-sampling is initiated. Each column corresponds to a released model trained with a different learning rate and momentum setting.} 
\label{fig:mnist_batch_size_appendix}
\end{figure*}
}

\arxiv{
\begin{figure*}[t]
\captionsetup{width=1.\linewidth}
  \centering
    \includegraphics[width=0.8\linewidth]{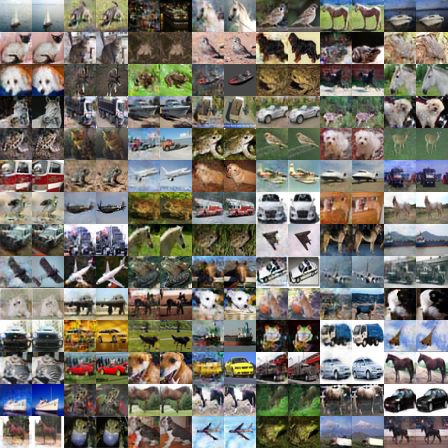}
\caption{More examples of CIFAR-10 reconstructions in the default attack setting. Odd columns are reconstructions and even columns are targets.}
\label{fig:cifar10_more_examples}
\end{figure*}
} 
\end{document}